\documentclass[12pt,reqno]{amsart}

\usepackage{amsmath,amssymb,amsthm}
\usepackage{ragged2e}
\usepackage{natbib}
\usepackage[dvips]{graphicx}
\usepackage{float}
\usepackage{epsfig,array,multirow}
\usepackage[usenames]{color}
\usepackage{geometry}
\usepackage{setspace}
\usepackage{color}
\usepackage[english]{babel}
\usepackage{lscape}
\usepackage[unicode=true,pdfusetitle,
 bookmarks=true,bookmarksnumbered=false,bookmarksopen=false,
 breaklinks=false,pdfborder={0 0 0},pdfborderstyle={},backref=false,colorlinks=true]
 {hyperref}
\hypersetup{
 linkcolor=red,urlcolor=red,citecolor=blue}
\usepackage{subfig}
\allowdisplaybreaks[4]

\theoremstyle{plain}
\newtheorem{theorem}{Theorem}[section]	
\newtheorem{lemma}{Lemma}[section]

\newtheorem{assumption}{Assumption}[section]

\theoremstyle{definition}

\newtheorem{remark}{Remark}[section]

\renewcommand{\qed}{\hfill{\tiny \ensuremath{\blacksquare} }}%

\newcommand{\Ep}{{\mathrm{E}}}

\renewcommand{\Pr}{{\mathrm{P}}}

\begin{document}

\title[Weighted-Average Quantile Regression]{Weighted-Average Quantile Regression}
\thanks{}

\author[Chetverikov]{Denis Chetverikov}
\author[Liu]{Yukun Liu}
\author[Tsyvinski]{Aleh Tsyvinski}

\address[D. Chetverikov]
{Department of Economics, UCLA, Bunche Hall, 8283, 315 Portola Plaza, Los Angeles, CA 90095, USA.}
\email{chetverikov@econ.ucla.edu}

\address[Y. Liu]
{Simon Business School, University of Rochester, 3-147 Carol Simon Hall, Rochester, NY 14611, USA.}
\email{yliu229@simon.rochester.edu}

\address[A. Tsyvinski]
{Department of Economics, Yale University, 28 Hillhouse Ave, New Haven, CT 06511, USA.}
\email{a.tsyvinski@yale.edu}

\date{\today.}

\begin{abstract}
In this paper, we introduce the weighted-average quantile regression framework, $\int_0^1 q_{Y|X}(u)\psi(u)du = X'\beta$, where $Y$ is a dependent variable, $X$ is a vector of covariates, $q_{Y|X}$ is the quantile function of the conditional distribution of $Y$ given $X$, $\psi$ is a weighting function, and $\beta$ is a vector of parameters. We argue that this framework is of interest in many applied settings and develop an estimator of the vector of parameters $\beta$. We show that our estimator is $\sqrt T$-consistent and asymptotically normal with mean zero and easily estimable covariance matrix, where $T$ is the size of available sample. We demonstrate the usefulness of our estimator by applying it in two empirical settings. In the first setting, we focus on financial data and study the factor structures of the expected shortfalls of the industry portfolios. In the second setting, we focus on wage data and study inequality and social welfare dependence on commonly used individual characteristics.  
\end{abstract}


\maketitle

\section{Introduction}
Mean and quantile regression models are among the key elements of
the econometrics toolbox. However, there is a large set of functionals beyond the mean and quantiles of a distribution that are of interest in applied work. 
It is therefore important to study other regression models as well.
To do so, we consider in this paper a broad class of regression
models, which we refer to as the \textit{weighted-average quantile
regression:}
\begin{equation}\label{eq: model}
\int_{0}^{1}q_{Y|X}(u)\psi(u)du=X'\beta,
\end{equation}
where $Y$ is a dependent variable in $\mathbb R$, $X$ is a vector of covariates in $\mathbb R^p$,
$q_{Y|X}\colon[0,1]\to\mathbb R$ is the quantile function of the conditional distribution
of $Y$ given $X$, $\psi\colon[0,1]\to\mathbb R$ is a signed
weighting function, and $\beta$ is the parameter vector in $\mathbb R^p$ to be estimated.
Such regression models are of interest in a number of applications. First,
if $Y$ is the loss of a financial portfolio and $\psi(u)=\mathbb{I}\left\{ u\geq1-\alpha\right\} /\alpha$
for, say, $\alpha = 0.1$, we obtain an example of a \textit{risk
regression}, namely an expected shortfall regression, which is of
interest in finance \cite[e.g., ][]{AB16, APPR17}.  This regression lets us study how the risk, measured by the expected shortfall, of a financial portfolio comoves with various
financial/macro variables. Second, if $Y$ is the wage and $\psi(u)=\mathbb{I}\left\{ u\leq\alpha\right\} /\alpha$
for, say, $\alpha=0.2$, we obtain a \textit{lower  wage regression},
and in the same way, we can also define middle and upper
wage regressions. These regressions are similar to the mean regression but apply the mean to separate wage classes: lower, middle, and upper classes. They may be of interest in labor economics
as parsimonious alternatives to quantile regression models. Third, if $Y$ is the wage and $\psi(u)=(\mathbb{I}\left\{ u\geq1-\alpha\right\} -\mathbb{I}\left\{ u\leq\alpha\right\} )/\alpha$
for, say, $\alpha=0.1$, we obtain an \textit{inequality regression}. The difference between the high and low income groups captures the inequality of the income distribution \cite[e.g., ][]{ACF06, BPP08, AP16}, and therefore the inequality regression may help to identify important determinants of social inequality. Similarly, if $\psi$ is a general decreasing function, we obtain a {\em social welfare regression}. Finally, if the researcher is concerned about the effect of data contamination, we can define $\psi(u)=\mathbb I\left\{ \alpha\leq u\leq1-\alpha\right\} /(1-2\alpha)$ for some small $\alpha$ to obtain a \textit{robust regression,} which may be of interest as an alternative to the Huber regression \citep{HR09}.

We assume that we have a stationary time series dataset $(X_1,Y_1),\dots,(X_T,Y_T)$
with each $(X_{t},Y_{t})$ having the same distribution as that of
the pair $(X,Y)$, and develop a $\sqrt{T}$-consistent estimator
of $\beta$. Our estimator, which we refer to as the weighted-average quantile regression estimator, consists of three steps. First, we use
machine learning to estimate the distribution function of the conditional distribution
of $Y$ given $X$. Second, we use this distribution function to construct a simple
transformation of $Y_{t}$ for all $t$. Third,
we estimate $\beta$ by running OLS of this transformation on $X_{t}$. We
prove that our estimator is asymptotically normal with mean zero and
that its asymptotic covariance matrix can be consistently estimated
by the Newey-West method on the third step, which is carried out in
all commonly used statistical software. Our estimation and inference
procedures are thus straightforward to implement.

Importantly, our approach is semi-parametric: we assume that the weighted-average conditional quantile function $x\mapsto\int_{0}^{1}q_{Y|X=x}(u)\psi(u)du$
is linear, to obtain broadly applicable results, but we do not impose any other
parametric restrictions. The latter is useful as it minimizes the
possibility of misspecification and inconsistent estimation. A parametric
alternative to our methods would be to (i) assume that each quantile
function $x\mapsto q_{Y|X=x}(u)$ is linear,
\begin{equation}
q_{Y|X}(u)=X'\beta(u),\quad\text{for all }u\in(0,1),\label{eq: linear quantiles}
\end{equation}
(ii) calculate the quantile regression estimator $\widetilde{\beta}(u)$
of each $\beta(u)$, and (iii) given that (\ref{eq: linear quantiles})
yields $\beta=\int_{0}^{1}\beta(u)\psi(u)du$, estimate $\beta$ by
$\int_{0}^{1}\widetilde{\beta}(u)\psi(u)du$. This alternative approach,
however, can lead to potential misspecification and inconsistent estimation
due to extra assumptions (\ref{eq: linear quantiles}) and, moreover,
requires estimating extreme quantile regressions, which correspond to
values of $u$ in (\ref{eq: linear quantiles}) that are close to the boundary of the $[0,1]$ interval, and which are typically difficult to estimate. In contrast,
our estimator does not require estimation of extreme quantile regressions.
Furthermore, in contrast to classical semiparametric estimation theory, by applying double/debiased machine learning techniques
\cite[e.g., ][]{CCDDHNR18}, our estimator
is able to handle the case where $X$ is moderate- or large-dimensional,
which is particularly useful in a number of applications. It is important
to note, however, that our results do not follow from the standard
results on double/debiased machine learning techniques because we
allow general weighting functions $\psi$ that in particular may feature
discontinuities. 

We apply our method to two empirical settings: a setting in finance and a setting in labor economics. In the first application, we focus on financial data and study the factor structures of the expected shortfalls of the industry portfolios. We show that the expected shortfalls of the industry portfolios have significant time-varying exposures to the factor models developed in the asset pricing literature. Importantly, the factor structures of the expected shortfalls of the industry portfolios based on the weighted-average quantile regressions can differ significantly from those estimated based on the mean and quantile regressions or based on a parametric estimator. We show that the discrepancies stem from the fact that the quantiles are not linear in the factors in the financial data.

In the second application, we apply the inequality and social welfare regressions to wage data. Using the inequality regression, we study the relationship between wage inequality and individual characteristics that are common in labor economics. We compare the inequality regression results with those based on a parametric estimator and show that the results can differ in important ways. For example, based on the weighted-average quantile regression estimator, the wage inequality is estimated to be negatively related to family size in the recent sample, but the relationship is muted using the parametric estimator. Applying the social welfare regression, we study the dependence of the weighted average wage on individual characteristics, where the weights are exponential with higher weights on the lower income. We call this a social welfare regression as it is consistent with a variety of social welfare functions used in, say, public finance that place higher weights on poorer individuals. We find that the results based on the social welfare regression differ from those from the mean regression. For example,  for the early 2000s, the magnitude of the point estimates on education based on the social welfare regression is only half as large as those using the mean regression.

\subsection*{Related Literature.}
Overall, our results generalize the commonly used mean (OLS) and quantile
regression methods to allow for a much larger class of functionals -- the
weighted-average quantiles, which are of interest in a number of applications
as outlined above and as discussed in detail in the next section.

More broadly, we contribute to the literature by providing a general principle for obtaining a double/debiased machine learning estimator of linear regression models of the form $h(F_{Y|X}) = X'\beta$, where $F_{Y|X}$ is the distribution function of the conditional distribution of $Y$ given $X$ and $h$ is a functional of interest, the weighted-average quantile being one example. Specifically, we first calculate the influence function $y\mapsto a(y,F)$ for the functional $F\mapsto h(F)$ and then estimate the vector of parameters $\beta$ by running OLS of an estimated version of $h(F_{Y|X}) + a(Y,F_{Y|X})$ on $X$. We note that although the idea of using influence function adjustments itself is not new, as it can be traced back at least to \cite{B82} and \cite{S86} and was used recently in \cite{CEINR16}, our key finding is that the adjustment term $a$ in the regression context depends only on the functional $h$ and not on the joint distribution of the pair $(X,Y)$, making the approach broadly available in applied settings. In addition, although our procedure looks similar to that in \cite{FFL09}, who propose to run OLS of $h(F_{Y}) + a(Y,F_Y)$ on $X$, where $F_Y$ is the distribution function of the {\em marginal} distribution of the random variable $Y$, the similarities are superficial: two procedures aim at estimating fundamentally different quantities and also have completely different reasons for the influence function adjustments. The detailed explanations of how our procedure relates to the literature can be found in Section \ref{sec: motivation}, after we fully describe and explain the general principle.

Our paper is also related to the financial literature on estimation of conditional risk measures. Tail risk measures, such as expected shortfall, are an important class of risk measures in finance \cite[e.g., ][]{LL13, JLN15, AB16, APPR17}. A few parametric and nonparametric methods for estimating expected shortfall regressions were proposed and analyzed in \cite{S05}, \cite{CW08}, \cite{PT08}, \cite{LPT12}, \cite{K12}, and \cite{MYT18}. However, all of these papers either assume nonparametric expected shortfall, which makes interpretations in applied work difficult, or impose the linear quantile assumption \eqref{eq: linear quantiles}, which leads to potential misspecification. Moreover, these papers consider only one risk measure: expected shortfall, whereas our methods cover a broad class of risk measures; see Section \ref{sec: examples} for details. \cite{K12} provides a nice comparison of the existing methods.

Our paper is also related to the semi-parametric methods developed in \cite{CSU12}, who consider the problem of estimating $\beta$ in the model \eqref{eq: model} with the integral over $u\in(0,1)$ replaced by the sum over a grid of values of $u$ in $(0,1)$. Clearly, the sum over a fine grid can be used to approximate the integral but the variance of their estimator depends on the inverse of the density of the conditional distribution of $Y$ given $X$ in the tails and generally blows up as we take finer grids, which makes our methods quite different from those developed in \cite{CSU12}.

Moreover, our paper is seemingly related to the methods developed in \cite{RRM14} and \cite{RR15}, who develop super-quantile regression methods. In principle, super-quantile is just another name for the expected shortfall. However, the estimators proposed in these papers do not converge to $\beta$ appearing in \eqref{eq: model} when we set $\psi(u) = \mathbb I\{u\geq 1- \alpha\}/\alpha$ to obtain the expected shortfall regression. Therefore, from the perspective of our setting, the estimators proposed in these papers are not consistent, even though they do converge to some meaningful quantities, see \cite{RRM14} for details.

\subsection*{Outline of the Paper.}
The rest of the paper is organized as follows. In the next section, we provide several examples covered by our weighted-average quantile regression framework \eqref{eq: model}. In Section \ref{sec: motivation}, we derive a general principle for obtaining double/debiased machine learning estimators of linear regression models of the form $h(F_{Y|X})=X'\beta$ and apply it to the weighted-average quantile regression. In Section \ref{sec: estimation}, we describe our estimation and inference procedures in detail. In Section, \ref{sec: asymptotic theory}, we prove consistency and derive the asymptotic distribution for our estimators. In Section \ref{sec: monte carlo}, we provide results of a small-scale Monte Carlo simulation study confirming good statistical properties of our estimators in finite samples. In Section \ref{sec: empirical applications}, we apply our procedures in two empirical settings that are of interest in finance and labor economics. In the Online Appendix, we collect all proofs, additional discussions, and extra tables and figures for the empirical applications.

\section{Examples}\label{sec: examples}
In this section, we describe various regression models covered by our general regression framework  \eqref{eq: model}.

\subsection{Risk Regression}\label{ex: risk regression}
In finance, the concept of risk measures is used to quantify the risk of
financial positions, e.g. \cite{FS02}.  Formally, any risk measure is a functional $\rho$ that is defined on
a set of random variables and that has certain desirable properties,
so that for a random variable $Y\in\mathbb{R}$ representing a loss of some financial position, the value of $\rho(Y)$ measures
the risk associated with $Y$. Most commonly used risk measures belong to the class of {\em spectral risk measures}. These risk measures have many desirable properties (positive homogeneity, translation invariance, monotonicity, sub-additivity, etc.) and take the following form \citep{A02}:
\begin{equation}\label{eq: spectral risk measure}
\rho(Y) = \int_0^1 q_Y(u)\psi(u)du,
\end{equation}
where $q_Y\colon [0,1]\to\mathbb R$ is the quantile function of the random variable $Y$ and $\psi\colon[0,1]\to\mathbb R$ is a an increasing weighting function such that (i) $\psi(u)\geq 0$ for all $u\in(0,1)$ and (ii) $\int_0^1 \psi(u)du = 1$. Here, the function $\psi$ is called the spectrum function associated with the risk measure $\rho$ and different spectrum functions $\psi$ lead to different risk measures. For example, one of the most important spectral risk measures is the expected shortfall, also known as the average value at risk, which corresponds to the spectrum function $\psi(u) = \mathbb I\{u\geq 1-\alpha\}/\alpha$, where $\mathbb I\{\cdot\}$ denotes the indicator function, and $\alpha\in(0,1)$ is a user-specified parameter, typically taking some small value such as $5\%$ or $10\%$. Other examples are exponential and polynomial risk measures, which correspond to the spectrum functions $\psi(u) = a\exp(-a(1-u))/(1-\exp(-a))$ with $a>0$ and $\psi(u) = a u^{a-1}$ with $a>1$, correspondingly, where $a$ is a user-specified parameter, e.g. \cite{L15}. A textbook-level discussion of spectral risk measures can be found in \cite{MFE15}.

To study how the risk of one random variable, say $Y$, comoves with a vector of other variables, say $X$, we can consider a {\em risk regression}
$
\rho(Y|X) = X'\beta,
$
where $\rho(Y|X) = \int_0^1 q_{Y|X}(u)\psi(u)du$ is the risk measure of the conditional distribution of $Y$ given $X$. Substituting here various functions $\psi$, we obtain various risk regressions, e.g. the expected shortfall and exponential regressions. These regressions are covered by our general framework \eqref{eq: model}.

In addition, we note also that risk measures, under different names, appear also
in behavioral economics, where they are used to rank lotteries, e.g. \cite{KT79} and \cite{Y87}, and in actuarial science, where they are used to determine premium principles, e.g. \cite{KGDD08}. Moreover, our methods can be used to estimate an expected shortfall (or any other spectral risk measure) version of CoVar, a concept introduced in \cite{AB16} to study systemic risk.

\subsection{Wage Regression}
Quantile regression methods have been used to study conditional wage distributions since \cite{B94} and \cite{C94}. If $Y$ is individual's wage and $X$ is a vector of covariates including, for example, education, running a quantile regression of $Y$ on $X$ lets us estimate the effect of education on the conditional distribution of wages for any quantile index $u$ of this distribution. This is useful because the effect of education may vary substantially depending on the quantile index. In practice, however, we may often be interested in the average effect of education for a {\em group} of quantile indices. For example, we may define the middle-wage class as the set of individuals with quantile indices within the $[20\%,80\%]$ interval on the conditional wage distribution and we may be interested in the effect of education for this particular set of individuals. In turn, quantile regression methods may not be appropriate for such parameters. Indeed, providing the quantile regression estimate for the average quantile index (50\%, in our example) may not give a representative number for the the whole group and providing the quantile regression estimates for each quantile index within the $[20\%,80\%]$ interval may not be convenient as function-valued estimates are difficult to interpret.\footnote{Moreover, averaging quantile regression estimates over quantile indices in the $[20\%,80\%]$ interval may not be a good idea either, for the reasons explained in the Introduction.} Instead, such parameters can be easily estimated by our weighted-average quantile regression methods. Specifically, by setting $\psi(u) = \mathbb I\{\alpha\leq u\leq 1-\alpha\}/(1-2\alpha)$ with $\alpha = 0.2$ in \eqref{eq: model}, we obtain a {\em middle wage regression}, and the methods developed in our paper can be used to estimate parameters of this regression, yielding in particular the average effect of education on wages for the middle-wage class. Similarly, by setting $\psi(u) = \mathbb I\{u\leq \alpha\}/\alpha$ with $\alpha = 0.2$ in \eqref{eq: model}, we obtain a {\em lower wage regression}, corresponding to the lower-wage class, and by setting $\psi(u) = \mathbb I\{u\geq 1 - \alpha\}/\alpha$, again with $\alpha = 0.2$, we obtain an {\em upper wage regression}, corresponding to the upper-wage class. More generally, since the same techniques can be used with any dependent variable $Y$, we can refer to these types of regression models simply as the lower, middle, and upper regressions. In this case, the expected shortfall regression discussed above becomes an instance of the upper regression.


\subsection{Inequality and Social Welfare Regressions}
Related to our discussion in the previous example, another reason to study conditional wage distributions is that they help us understand the dynamics of the wage inequality over time, e.g. \cite{AP08}. Again assuming that $Y$ is individual's wage and $X$ is a vector of relevant covariates, we can study wage inequality by our weighted-average quantile regression methods. Indeed, by setting $\psi(u) = (\mathbb I\{u\geq 1-\alpha\} - \mathbb I\{u\leq \alpha\})/\alpha$ in \eqref{eq: model} for some small $\alpha$, say $0.1$, we obtain an {\em inequality regression}, which allows us to study how the difference between the average wage of 10\% individuals with highest wages and the average wage of 10\% individuals with lowest wages depend on covariates. Similarly, by considering any decreasing function $\psi$, e.g. polynomial or exponential from Section \ref{ex: risk regression} with $u$ replaced by $1-u$, we obtain a {\em social welfare regression}.
 Of course, the inequality regression remains meaningful with other dependent variables as well. More broadly, it is straightforward to generalize the weighted-average quantile regression framework to include other inequality measures such as Gini's coefficient, e.g. \cite{Co11}.


\subsection{Robust Regression}
Suppose that we are interested in estimating a linear mean regression model
$$
\Ep[Y|X] = X'\beta
$$
from a stationary time series $(X_1,Y_1),\dots,(X_T,Y_T)$, where each $(X_t,Y_t)$ has the same distribution as that of the pair $(X,Y)$. Typically, we would estimate $\beta$ in this model by OLS
$$
\widehat\beta = \arg\min_{b\in\mathbb R^p} \sum_{t=1}^T (Y_t - X_t'b)^2.
$$
Suppose, however, that for some observations $t$, the values of the dependent variable $Y_t$ are corrupted. These corrupted values may significantly bias the estimator $\widehat \beta$ rendering it unreliable. This problem attracted substantial attention in the literature and led to the field called Robust Statistics, which generated many alternatives to OLS, e.g. Least Median of Squares \citep{R84}, Least Trimmed Squares \citep{RL87}, and Random Sample Consensus \citep{FB81}; see also recent advances in computer science, e.g. \cite{LCR18}. However, one of the most important methods developed in this field is the Huber estimator \citep{HR09}, which can be viewed as a modification of the OLS estimator:
$$
\widetilde\beta = \arg\min_{b\in\mathbb R^p} \sum_{t=1}^T \rho_c(Y_t - X_t'b),
$$
where
$$
\rho_c(x) = \begin{cases}
x^{2}, & \text{if }|x|\leq c,\\
2|x|c, & \text{if }|x|>c,
\end{cases}
$$
and $c>0$ is a tuning parameter. Since the derivative of the criterion function $\rho_c$ in the Huber estimator is bounded, this estimator is much more robust with respect to data corruption in the dependent variable in comparison with the OLS estimator. However, implementing this estimator requires choosing the tuning parameter $c$, which is often unclear in practice: smaller values of $c$ yield more robust but also more biased estimator. We therefore propose to use our weighted-average quantile regression estimator as an alternative. Indeed, suppose that for each observation $t$, the probability of corruption in $Y_t$ does not exceed $\alpha$ for some small user-specified value $\alpha\in(0,1)$. In this case, we can consider a {\em robust regression} by setting $w(u) = \mathbb I\{\alpha \leq u\leq 1 - \alpha\}/(1-2\alpha)$ in \eqref{eq: model}. Running our estimator based on this regression also requires the choice of the tuning parameter, $\alpha$, but in contrast to the Huber estimator, this choice is rather intuitive: the user simply needs to provide an upper bound on the fraction of corrupted observations.

\section{Motivation For Estimation Procedure}\label{sec: general principle}\label{sec: motivation}
In this section, we develop a general principle for estimating regression models
\begin{equation}\label{eq: general regression}
h(F_{Y|X})=X'\beta,
\end{equation}
where $y\mapsto F_{Y|X}(y) = \Pr(Y\leq y|X)$ is the distribution function of the conditional distribution of $Y$ given $X$ and $h\colon \mathcal F\to \mathbb R$ is a functional defined on a convex set $\mathcal F$ of distribution functions on $\mathbb R$ that includes $F_{Y|X}$ almost surely. We then apply the general principle to the weighted-average quantile regression model by substituting $h(F_{Y|X}) = \int_0^1 q_{Y|X}(u)\psi(u)du$. For clarity of the section, we leave technical regularity conditions underlying our derivations for now.

\subsection{General Principle}

To develop the principle, we need the concept of influence functions. Following \cite{HRRS86}, we say that $h$ is Gateaux differentiable at the distribution function $F\in\mathcal F$ if there exists a function $a\colon\mathbb R\to\mathbb R$ such that for all $G\in\mathcal F$, we have
\begin{equation}\label{eq: influence function}
\frac{\partial}{\partial t}h((1-t)F+tG)\Big|_{t=0_+} = \lim_{t\downarrow 0}\frac{h((1-t)F+tG) - h(F)}{t} = \int a(y)dG(y).
\end{equation}
We refer to $a$ as the influence function of $h$ at $F$. To note its dependence on $F$, we will write $a(y,F)$ instead of $a(y)$.

The influence function has two important properties. First, by substituting $G = F$ in \eqref{eq: influence function}, we have $\int a(y,F)dF(y) = 0$ and since $F\in\mathcal F$ is arbitrary, we obtain
\begin{equation}\label{eq: influence function mean zero}
\int a(y,F)dF(y) = 0,\quad\text{for all }F\in\mathcal F.
\end{equation}
Second, by substituting $(1-t)F+tG$ instead of $F$ in \eqref{eq: influence function mean zero} and taking derivative with respect to $t$ on both sides, we have
$$
\frac{\partial}{\partial t}\int a(y,(1-t)F+tG)dF(y)\Big|_{t=0_+} + \int a(y,F)dG(y) - \int a(y,F)dF(y)=0
$$
and since $\int a(y,F)dF(y) = 0$ and $G$ is arbitrary, we obtain
\begin{equation}\label{eq: influence function second property}
\frac{\partial}{\partial t}\int a(y,(1-t)F+tG)dF(y)\Big|_{t=0_+} = - \int a(y,F)dG(y),\quad\text{for all }G\in\mathcal F
\end{equation}
We will use these two properties below.

Having the concept of influence functions in mind, we propose the following principle: estimate $\beta$ in \eqref{eq: general regression} by running the OLS estimator of $h(\widehat F_{Y|X}) + a(Y,\widehat F_{Y|X})$ on $X$, where $\widehat F_{Y|X}$ is a preliminary estimator of $F_{Y|X}$. We claim that this estimator is consistent and robust (in a sense to be made precise below) with respect to the estimation error in $\widehat F_{Y|X}$. To see why this is so, observe that, under certain regularity conditions, the probability limit of such an OLS estimator will be
$$
\bar\beta = (\Ep[XX'])^{-1}\Ep\Big[X\Big(h(F_{Y|X})+ a(Y,F_{Y|X})\Big)\Big],
$$
which can be equivalently written as a set of moment conditions
\begin{equation}\label{eq: moment conditions}
\Ep\Big[X\Big(h(F_{Y|X}) + a(Y,F_{Y|X})-X'\bar\beta\Big)\Big] = 0_p,
\end{equation}
where $0_p = (0,\dots,0)'\in\mathbb R^p$. Here, we note that by applying \eqref{eq: influence function mean zero} with $F = F_{Y|X}$, we have
$$
 \int a(y,F_{Y|X})dF_{Y|X}(y) = 0.
$$
In turn, the left-hand side of this identity is equal to $\Ep[a(Y,F_{Y|X})|X]$, and so, by the law of iterated expectations,
$$
\Ep[Xa(Y,F_{Y|X})] = \Ep[X \Ep[a(Y,F_{Y|X})|X]] = 0_p.
$$
Substituting this equality into \eqref{eq: moment conditions} and recalling \eqref{eq: general regression}, 
it follows that the set of moment conditions \eqref{eq: moment conditions} can be equivalently written as
$$
\Ep[X(X'\beta - X'\bar\beta)]=0_p.
$$
As long as $\Ep[XX']$ is non-singular, it thus follows that $\bar\beta = \beta$ is the unique solution to the set of moment conditions \eqref{eq: moment conditions}. This means that our OLS estimator is consistent.

Further, fix any $\widehat F_{Y|X}$ such that $\widehat F_{Y|X}\in\mathcal F$ almost surely and write $\widetilde F = \widehat F_{Y|X} - F_{Y|X}$. By applying \eqref{eq: influence function} and \eqref{eq: influence function second property} with $F = F_{Y|X}$ and $G = \widehat F_{Y|X}$ and noting that $(1-t)F+tG = F_{Y|X} + t\widetilde F_{Y|X}$ in this case, we have
$$
\frac{\partial }{\partial t} h(F_{Y|X} + t\widetilde F)\Big|_{t = 0_+} = \int a(y,F_{Y|X})d\widehat F_{Y|X}(y)
$$
and
$$
\frac{\partial}{\partial t}\int a(y,F_{Y|X} + t\widetilde F)dF_{Y|X}(y)\Big|_{t=0_+} = -\int a(y,F_{Y|X})d\widehat F_{Y|X}(y),
$$
respectively. Summing these two identities and observing that 
$$
\int a(y,F_{Y|X} + t\widetilde F)dF_{Y|X}(y) = \Ep[a(Y,F_{Y|X} + t\widetilde F)|X],
$$
we obtain
$$
\frac{\partial }{\partial t} h(F_{Y|X} + t\widetilde F)\Big|_{t=0_+} + \frac{\partial}{\partial t}\Ep[a(Y,F_{Y|X} + t\widetilde F)|X]\Big|_{t=0_+} = 0.
$$
The latter in turn implies, via the law of iterated expectations, that
$$
\frac{\partial}{\partial t}\Ep\Big[X\Big(h(F_{Y|X}+t\widetilde F) + a(Y,F_{Y|X} + t\widetilde F)-X'\bar\beta\Big)\Big]\Big|_{t=0_+} = 0,
$$
as long as integration and differentiation can be interchanged. This means that our OLS estimator solves a system of equations having the Neyman orthogonality property with respect to $F_{Y|X}$ \citep{CCDDHNR18} and is, in this sense, robust with respect to the estimation error in $\widehat F_{Y|X}$.

Intuitively, a simple approach to estimate $\beta$ in the model \eqref{eq: general regression} would be to run the OLS estimator of $h(\widehat F_{Y|X})$ on $X$. Such an estimator can be shown to be $\sqrt T$-consistent and asymptotically normal with mean zero as long as $X$ is low-dimensional, a sufficiently simple estimator $\widehat F_{Y|X}$ is used, and its tuning parameters are chosen in a delicate way. When $X$ is moderate- or even large-dimensional, however, we have to rely on machine learning methods to obtain an estimator $\widehat F_{Y|X}$. These methods in turn yield heavily biased estimators with relatively slow convergence rates. The estimation error in $\widehat F_{Y|X}$ may then propagate into the error of the OLS estimator, leading to estimates of $\beta$ with poor properties. We deal with this problem by adding the influence function $a(Y,\widehat F_{Y|X})$ to the functional $h(\widehat F_{Y|X})$. This allows us to obtain the OLS estimator of $\beta$ that is robust with respect to the estimation error in $\widehat F_{Y|X}$, as explained above.

There are several strands of literature that use influence function adjustments for estimation. First, the so-called one-step estimators, which adjust the plug-in estimators by adding the average value of the estimated influence function, have been used in statistics for a long time as a tool of achieving semiparametric efficiency; see \cite{B82} and \cite{S86} for early references and \cite{FK21} for a recent review. Second, \cite{CEINR16} introduced the idea of adding the influence functions to obtain robust estimators in the setting where the parameter of interest is the expectation of a functional of unknown nonparametric/high-dimensional object that has to be estimated on the first step. We show that in our context the adjustment function $a$ depends only on the functional $h$ and not on the joint distribution of underlying random variables, and thus has a simple form, broadly available in applied settings. We exemplify this last point in the next subsection, where we apply our procedure to the weighted-average quantile regression. Third, \cite{FFL09} proposed a procedure that, in its simplest form, consists of running the OLS estimator of $h(\widehat F_Y) + a(Y,\widehat F_Y)$ on $X$, where $\widehat F_Y$ is a preliminary estimator of the distribution function $F_Y$ of $Y$, in order to estimate the impact of $X$ on the functionals of the counterfactual distribution of $Y$ that appears as the distribution of $X$ changes keeping the conditional distribution of $Y$ given $X$ fixed. They use the {\em marginal} distribution of $Y$, whereas we use the {\em conditional} distribution of $Y$ given $X$. This seemingly minor change has substantial consequences: two procedures aim at estimating fundamentally different quantities and have different scopes of applicability. Using an example from the Introduction of \cite{FFL09}, one can say that their procedure estimates how increasing the fraction of unionized workers affects the distribution of wages whereas our procedure estimates how the distribution of wages of unionized workers differs from the distribution of wages of non-unionized workers. In addition, the reasons for adding the influence function in two procedures are completely different. As explained above, we use the influence function adjustment to obtain an OLS estimator that is robust with respect to the estimation error in $\widehat F_{Y|X}$, whereas they use the influence function because it directly measures the impact of changing the distribution on the value of the functional, see \eqref{eq: influence function}.\footnote{Also, as long as the intercept is included, the term $h(\widehat F_Y)$ plays no role in their procedure: the slope coefficients of the OLS estimator of $h(\widehat F_Y) + a(Y,\widehat F_Y)$ on $X$ coincide with the slope coefficients of the OLS estimator of $a(Y,\widehat F_Y)$ on $X$. In contrast, dropping $h(\widehat F_{Y|X})$ in our procedure would lead to a meaningless estimator that would converge in probability to the vector of zeros.}

More generally, our approach to estimation in this section is an instance of the double/debiased machine learning method \citep{CCDDHNR18}, which gives estimation procedures based on moment conditions with the Neyman orthogonality property. Our key innovation here is that we demonstrate that although we are interested in an object $h(F_{Y|X})$ that depends on {\em conditional} distributions, with dependence going through in a potentially complicated way, e.g. via conditional quantile functions, obtaining moment conditions with the Neyman orthogonality property is actually simple: we simply have to add the influence function for the functional $h$, which can be obtained by looking at the values $h(F)$ of the functional $h$ at {\em unconditional} distributions $F$. In addition, an important issue that arises in almost all of our applications is that the functionals we consider are not twice continuously differentiable (in a Gateaux sense), thus the results of  \cite{CCDDHNR18} can not be applied.


\subsection{Application to Weighted-Average Quantile Regression}
Here, we apply the general principle described in the previous subsection to the weighted-average quantile regression model \eqref{eq: model}. To do so, we set $h(F) = \int_0^1 q_F(u)\psi(u)du$, where $q_F(u)$ denotes the $u$th quantile of the distribution function $F\in\mathcal F$. The influence function for this functional is well-known:
\begin{equation}\label{eq: original influence function}
a(y,F) = \int_0^1\frac{u-\mathbb I\{y\leq q_F(u)\}}{f(q_F(u))}\psi(u)du,
\end{equation}
where $f = F'$ is the pdf corresponding to $F$. To see why, suppose first that we are interested in the individual quantile $q_F(u)$ for some $u\in(0,1)$. Let $F$ and $G$ be two distribution functions with strictly positive derivatives $f$ and $g$, respectively. Then for any $t\in[0,1]$, we have
$$
\int_{-\infty}^{q_{(1-t)F+tG}(u)}((1-t)f(y)+tg(y))dy = u.
$$
Taking derivative of both sides with respect to $t$ at $t = 0_+$, we then obtain
$$
f(q_F(u))\frac{\partial}{\partial t}q_{(1-t)F+tG}(u)\Big|_{t=0_+} + \int_{-\infty}^{q_F(u)} (g(y)-f(y))dy = 0.
$$
Hence,
\begin{align*}
\frac{\partial}{\partial t}q_{(1-t)F+tG}(u)\Big|_{t=0_+}
& = \int_{-\infty}^{q_F(u)}\frac{f(y)-g(y)}{f(q_F(u))}dy\\
& = \frac{u}{f(q_F(u))} - \int_{-\infty}^{+\infty}\frac{\mathbb I\{y\leq q_F(u)\}g(y)}{f(q_F(u))}dy\\
& = \int_{-\infty}^{+\infty} \frac{u - \mathbb I\{y\leq q_F(u)\}}{f(q_F(u))}g(y)dy.
\end{align*}
Comparing this expression with \eqref{eq: influence function}, we thus obtain the influence function for the individual quantile $q_F(u)$:
$$
y\mapsto \frac{u - \mathbb I\{y\leq q_F(u)\}}{f(q_F(u))}.
$$
This expression in turn immediately gives \eqref{eq: original influence function} in light of linear additivity of derivatives.

Applying our general principle from the previous subsection, we conclude that in order to estimate $\beta$ in the model \eqref{eq: model}, we can run OLS of an estimated version of
\begin{equation}\label{eq: non-transformed dependent variable}
\int_0^1\left(q_{Y|X}(u) + \frac{u-\mathbb I\{Y\leq q_{Y|X}(u)\}}{f_{Y|X}(q_{Y|X}(u))}\right)\psi(u)du
\end{equation}
on $X$, where $f_{Y|X}$ is the pdf of the conditional distribution of $Y$ given $X$.  This, however, is not convenient for two reasons. First, this approach requires estimating extreme quantiles, $q_{Y|X}(u)$ for $u$ close to the boundaries of the interval $[0,1]$, which are typically difficult to estimate. Second, this approach requires estimating the pdf $f_{Y|X}(q_{Y|X}(u))$, which appears in the denominator and which may take small values near the boundaries of the interval $[0,1]$, thus leading to large estimation errors. Note also that simply truncating the interval $[0,1]$ may not lead to good results as in some cases, such as the expected shortfall or inequality regressions, extreme quantiles are of particular importance. To deal with these problems, we rewrite the integral in \eqref{eq: non-transformed dependent variable} differently.

First, observe that for all $u\in[0,1]$, we have
$$
\int_{-\infty}^{q_{Y|X}(u)}f_{Y|X}(y)dy = u,
$$
and so $f_{Y|X}(q_{Y|X}(u))q'_{Y|X}(u) = 1$ almost surely. Therefore, by applying the change of variables $u\mapsto s(u) = q_{Y|X}(u)$ and recalling that $u = F_{Y|X}(s(u))$ in this case, it follows that the integral in \eqref{eq: non-transformed dependent variable} is equal to
\begin{multline}
\int_{-\infty}^{+\infty} \left( s + \frac{F_{Y|X}(s) - \mathbb I\{Y\leq s\}}{f_{Y|X}(s)} \right)f_{Y|X}(s)\psi(F_{Y|X}(s))ds \\
=\int_{-\infty}^{+\infty} sf_{Y|X}(s)\psi(F_{Y|X}(s))ds + \int_{-\infty}^{+\infty}(F_{Y|X}(s) - \mathbb I\{Y\leq s\})\psi(F_{Y|X}(s))ds.\label{eq: two integrals}
\end{multline}
Second, using integration by parts, we can further rewrite the first integral in \eqref{eq: two integrals} as
\begin{align*}
& \int_{-\infty}^{+\infty} sf_{Y|X}(s)\psi(F_{Y|X}(s))ds \\
&\qquad = \int_{-\infty}^{0} sf_{Y|X}(s)\psi(F_{Y|X}(s))ds + \int_{0}^{+\infty} sf_{Y|X}(s)\psi(F_{Y|X}(s))ds \\
&\qquad = - \int_{-\infty}^0 \Psi(F_{Y|X}(s))ds + \int_0^{+\infty}(\bar\Psi - \Psi(F_{Y|X}(s)))ds,
\end{align*}
where $\Psi\colon[0,1]\to\mathbb R$ is the function defined by $\Psi(s) = \int_0^s \psi(u)du$ for all $s\in[0,1]$ and $\bar\Psi = \int_0^1 \psi(u)du$.
Combining these results, it follows that the integral in \eqref{eq: non-transformed dependent variable} is equal to
\begin{align}
R = & - \int_{-\infty}^0 \Psi(F_{Y|X}(s))ds + \int_0^{+\infty}(\bar\Psi - \Psi(F_{Y|X}(s)))ds  \nonumber\\
&\qquad \qquad + \int_{-\infty}^{+\infty}(F_{Y|X}(s) - \mathbb I\{Y\leq s\})\psi(F_{Y|X}(s))ds.\label{eq: definition r new}
\end{align}
We thus propose estimating the vector of parameters $\beta$ in the weighted-average quantile regression model \eqref{eq: model} by running OLS of an estimated version of $R$ on $X$. In comparison with the integral in \eqref{eq: non-transformed dependent variable}, the advantage of using $R$ is that it depends only on the distribution function $F_{Y|X}$, which is easy to estimate even in the tails.

\section{Estimation and Inference}\label{sec: estimation}
In this section, we provide a detailed discussion of our estimation and inference procedures for the vector of parameters $\beta$ in the weighted-average quantile regression model \eqref{eq: model}. We assume, throughout the rest of the paper, that we have a strictly stationary time series dataset $(X_1,Y_1),\dots,(X_T,Y_T)$, where each $(X_t,Y_t)$ has the same distribution as that of the pair $(X,Y)$. 

For all $s\in\mathbb R$ and $x$ in the support of $X$, denote $F(s|x) = \Pr(Y\leq s|X=x)$, so that $F(\cdot|X) = F_{Y|X}(\cdot)$ is the distribution function of the conditional distribution of $Y$ given $X$. Also, as in the previous section, denote $\Psi(s) = \int_0^s \psi(u)du$ for all $s\in[0,1]$ and $\bar\Psi = \int_0^1 \psi(u)du$. In addition, define $R$ as in \eqref{eq: definition r new} and
\begin{equation}\label{eq: residual definition}
e = \int_{-\infty}^{+\infty}\left(F(s|X) - \mathbb I\{Y\leq s\}\right)\psi(F(s|X))ds.
\end{equation}
By the law of iterated expectations, we then have $\Ep[e|X] = 0$. In addition,
by discussion at the end of the previous section, we also have
$$
\int_0^1 q_{Y|X}(u)\psi(u)du = - \int_{-\infty}^0 \Psi(F(s|X))ds + \int_0^{+\infty}(\bar\Psi - \Psi(F(s|X)))ds.
$$
Thus, it follows from \eqref{eq: model} that
\begin{equation}\label{eq: definition R}
R = X'\beta + e,\quad\text{where}\quad \Ep[e|X]=0.
\end{equation}
This equation reinforces our proposal in the previous section to estimate $\beta$ by the OLS method, regressing $R$ on $X$, where the distribution function $F$ appearing in $R$ is replaced by a suitable nonparametric/machine learning estimator, to be discussed later. For technical reasons, we also rely on sample splitting, so that the function $F$ and the vector $\beta$ are estimated on different subsamples of the whole sample. Formally, we define our estimator of $\beta$ as follows:



\begin{enumerate}
\item split the full sample $(X_1,Y_1),\dots,(X_T,Y_T)$ into two consequtive subsamples, say, $(X_1,Y_1),\dots,(X_{T_1},Y_{T_1})$ and $(X_{T_1+1},Y_{T_1+1}),\dots,(X_{T_1+T_2},Y_{T_1+T_2})$, where $T_1 + T_2 = T$;
\item use the first subsample, $(X_1,Y_1),\dots,(X_{T_1},Y_{T_1})$, to obtain a nonparametric/ machine learning estimator $\widehat F(s|x)$ of $F(s|x)$ for all $x\in\{X_{T_1+1},\dots,X_{T_1+T_2}\}$ and $s\in\mathbb R$;
\item calculate
\begin{align}
\widehat R_t = & \int_0^{+\infty}(\bar\Psi - \Psi(\widehat F(s|X_t)))ds - \int_{-\infty}^0 \Psi(\widehat F(s|X_t))ds \nonumber\\
&\qquad + \int_{-\infty}^{+\infty}\left(\widehat F(s|X_t) - \mathbb I\{Y_t\leq s\}\right)\psi(\widehat F(s|X_t))ds,\label{eq: rhat def}
\end{align}
for all $t = T_1+1,\dots,T_1 + T_2$;
\item calculate the OLS estimator
\begin{equation}\label{eq: risk regression estimator}
\widehat\beta = \left(\sum_{t = T_1+1}^{T_1+T_2} X_tX_t'\right)^{-1}\left(\sum_{t=T_1+1}^{T_1+T_2} X_t \widehat R_t\right).
\end{equation}
\end{enumerate}
\noindent
In this procedure, $T_1$ and $T_2$ should be chosen to be of the same order, which we assume for the rest of the paper.  In our simulations, we choose $T_1 \approx 2T_2$.


By analogy with mean and quantile regression estimators, we refer to $\widehat\beta$ as the {\em weighted-average quantile regression estimator}. By substituting various weighting functions $\psi$ (and the corresponding $\Psi$), we obtain various regression estimators. For instance, if $Y$ is the loss of a financial portfolio, by setting $\psi(u) = \mathbb I\{u \geq 1 - \alpha\}/\alpha$, we obtain an expected shortfall regression estimator, as discussed in Section \ref{sec: examples}.

We will prove in the next section that, under suitable regularity conditions, the estimator $\widehat\beta$ is $\sqrt T$-consistent and asymptotically normal with mean zero:
\begin{equation}\label{eq: main convergence}
\sqrt{T_2}(\widehat\beta - \beta) = \left(\frac{1}{T_2}\sum_{t=T_1+1}^{T_1+T_2}X_tX_t'\right)^{-1}\left(\frac{1}{\sqrt{T_2}}\sum_{t=T_1+1}^{T_1+T_2}X_te_t\right) + o_P(1)\to_d N(0,\Sigma),
\end{equation}
where 
$$
e_t = \int_{-\infty}^{+\infty}\left(F(s|X_t) - \mathbb I\{Y_t\leq s\}\right)\psi(F(s|X_t))ds,\quad\text{for all }t=T_1+1,\dots,T_1+T_2,
$$
and $\Sigma$ is the asymptotic covariance matrix. Moreover, $\Sigma$ can be consistently estimated, for example, by the Newey-West method, where each $e_t$ is replaced by the corresponding residual $\widehat e_t = \widehat R_t - X_t'\widehat\beta$; see the next section for details. This implies that once an estimator of the function $F$ is obtained using data from the first subsample, any standard statistical software can be used to obtain the estimator $\widehat\beta$ and corresponding standard errors by simply running the OLS regression of $\widehat R_t$ on $X_t$ using data from the second subsample and reporting the Newey-West standard errors. 

Next, we discuss estimation of the function $F$. To do so, fix any $s\in\mathbb R$ and observe that $F(s|x) = \Pr(Y\leq s | X=x) = \Ep[\mathbb I\{Y\leq s\}| X = x]$ for all $x\in\mathbb R^p$. Therefore, to obtain an estimator of the function $x\mapsto F(s|x)$, we can apply any standard nonparametric/machine learning estimator regressing $\mathbb I\{Y_t\leq s\}$ on $X_t$ using data from the first subsample. Applying the estimator separately for each value of $s$, we obtain an estimator $(s,x)\mapsto \widehat F(s|x)$ of the function $(s,x)\mapsto F(s|x)$. For example, for our empirical results, we use a version of the random forest method described below.\footnote{We also tried boosting and $\ell_1$-penalized methods but they did not perform as well as random forests: both methods turned out slower and the latter method suffered from potentially substantial non-linearities in the functions $x\mapsto F(s|x)$.}

Further, nonparametric/machine learning estimators will produce numerically identical results for any pair of values of $s$, say $s_1$ and $s_2$, such that there is no $Y_t$ between $s_1$ and $s_2$ since the datasets $\{(X_t,\mathbb I\{Y_t\leq s_1\})\}_{t=1}^{T_1}$ and $\{(X_t,\mathbb I\{Y_t\leq s_2\})\}_{t=1}^{T_1}$ are identical in this case. This in turn implies that there is no need to apply the nonparametric/machine learning estimator for all values of $s\in\mathbb R$, and it suffices to only consider $s\in\{Y_1,\dots,Y_{T_1}\}$ since the estimator $s\mapsto \widehat F(s|x)$ will be piecewise constant and the integrals in \eqref{eq: rhat def} will be given by the corresponding sums. More precisely, letting $s_1\leq\dots\leq s_{T_1}$ be the sequence of values of $Y_1,\dots,Y_{T_1}$ arranged in the increasing order and imposing a mild condition that $\widehat F(s|X) = 0$ for all $s < \min_{1\leq t\leq T_1} Y_{t}$ and $\widehat F(s|X) = 1$ for all $s \geq \max_{1\leq t\leq T_1} Y_{t}$ a.s., which is satisfied for any reasonable nonparametric/machine learning estimator, it follows that
\begin{equation}\label{eq: alternative r}
\widehat R_t = s_{T_1}\bar\Psi +  \sum_{j=1}^{T_1-1} (s_{j+1} - s_{j})M_{t}(s_{j}, s_{j+1}),\quad\text{for all }t=T_1+1,\dots,T_1+T_2
\end{equation}
where we denoted
\begin{equation}\label{eq: function m}
M_{t}(s_j,s_{j+1}) = - \Psi(\widehat F(s_j|X_t)) + (\widehat F(s_j|X_t)  - \widetilde {\mathbb I}\{Y_t\leq s_j,s_{j+1}\})\psi(\widehat F(s_j|X_t))
\end{equation}
and 
$$
\widetilde {\mathbb I}\{Y_t\leq s_j,s_{j+1}\} = \max\left(\min\left(\frac{s_{j+1} - Y_t}{s_{j+1} - s_j},1\right),0\right)
$$ 
for all $j = 1,\dots,T_1 - 1$ and $t = T_1+1,\dots,T_1+T_2$.

Moreover, in large samples, where calculating $\widehat R_t$ in \eqref{eq: alternative r} is computationally costly, the grid $s_1\leq\dots\leq s_{T_1}$ used in \eqref{eq: alternative r} can be replaced by a much coarser grid $\min_{1\leq t\leq T_1}Y_t = s_1^*\leq\dots\leq s_k^* = \max_{1\leq t\leq T_1}Y_t$, so that
$$
\widehat R_t = s_{k}^* +  \sum_{j=1}^{k-1} (s^*_{j+1} - s^*_{j}) M_{t}(s_j^*,s_{j+1}^*),\quad\text{for all }t=T_1+1,\dots,T_1+T_2,
$$ 
where $k$ is much smaller than $T_1$. 





\section{Asymptotic Theory}\label{sec: asymptotic theory}

In this section, we derive an asymptotic theory for the weighted-average quantile regression estimator $\widehat\beta$. To do so, we denote $D_t = (X_t',Y_t)'$ for all $t=1,\dots,T$ and $D_1^{T_1} = (D_1,\dots,D_{T_1})$. We will assume that the dataset $D_1,\dots,D_T$ is a subset of a strictly stationary time series $\{D_t\}_{t\in\mathbb Z}$. Further, for all $j\in\mathbb N$, let $\mathcal I_{-\infty}^0$ and $\mathcal I_{j}^{+\infty}$ be $\sigma$-algebras generated by $\{D_s\}_{s\leq 0}$ and $\{D_s\}_{s\geq j}$, respectively, and let
$$
\beta_j = \Ep\left[\sup\left\{ |\Pr(B\mid \mathcal I_{-\infty}^0) - \Pr(B)| \colon B\in\mathcal I_{j}^{+\infty}\right\}\right]
$$
be the $\beta$-mixing coefficients. In addition, let $\mathcal X$ be the support of $X$ and for all $x\in\mathcal X$, denote
\begin{equation}\label{eq: Delta definition}
\Delta(x) = \sup_{s\in\mathbb R}\left| \widehat F(s|x) - F(s|x) \right| + \int_{-\infty}^{+\infty} \left| \widehat F(s|x) - F(s|x)  \right| ds.
\end{equation}
Moreover, let $0 < u_0 < 1/2$, $0<c<\infty$, and $-\infty < s_1 < s_2< +\infty$ be some constants. We will use the following assumptions.

\begin{assumption}\label{as: beta mixing}
The strictly stationary time series $\{D_t\}_{t\in\mathbb Z}$ has summable $\beta$-mixing coefficients: $\sum_{j=1}^{\infty}\beta_j < \infty$.
\end{assumption}
Assumption \ref{as: beta mixing} implies that the time series $\{D_t\}_{t\in\mathbb Z}$ is absolutely regular. As explained in \cite{C11}, many econometric time series models satisfy this assumption. In fact, it is common practice to impose stronger mixing conditions. For example, \cite{CL13} require that $\beta_j \leq \beta_0 j^{-\omega}$ for all $j\geq 1$ and some $\beta_0>0$ and $\omega>2$, which clearly implies that $\sum_{j=1}^\infty \beta_j < \infty$. \cite{FHLV16} require that $\phi_j \leq \phi_0 j^{-\omega}$ for all $j\geq 1$ and some $\phi_0>0$ and $\omega > 1$, where $\phi_j$'s are $\phi$-mixing coefficients. Since $\phi_j\geq \beta_j$ for all $j\geq 1$, such a condition also implies our Assumption \ref{as: beta mixing}. Note also that Assumption \ref{as: beta mixing} holds trivially if the random vectors $D_t$ are independent across $t$, which means that our results apply for i.i.d. data settings as well. We refer an interested reader to \cite{FY05} and \cite{B05} for detailed explanations on various mixing conditions and their plausibility.
\begin{assumption}\label{as: standard convergence}
(i) Components of the random vector $X$ as well as the random variable $e$ have finite fourth moments: $\Ep[\|X\|^4] < \infty$ and $\Ep[e^4]<\infty$.
(ii) In addition, the matrix $\Ep[XX']$ is positive-definite. (iii) Moreover,
$$
\frac{1}{\sqrt{T_2}}\sum_{t=T_1+1}^{T_1+T_2}e_tX_t\to_d N(0,\Omega)
$$
for a positive-definite matrix $\Omega$.
\end{assumption}
Assumption \ref{as: standard convergence} is standard in time series econometrics.  Assumption \ref{as: standard convergence}(i) is a mild moment condition. Assumption \ref{as: standard convergence}(ii) is an identification condition. Assumption \ref{as: standard convergence}(iii) follows from a combination of mixing and moment conditions. For example, since $\beta$-mixing coefficients dominate $\alpha$-mixing coefficients, under Assumption \ref{as: beta mixing}, Assumption \ref{as: standard convergence}(iii) holds as long as the random variables $\|e_tX_t\|$ are bounded; see Theorem 2.21(ii) in \cite{FY05}. More generally, when the random variables $\|e_tX_t\|$ satisfy $\Ep[\|e_tX_t\|^{\eta}]<\infty$ for some $\eta > 2$, Assumption \ref{as: standard convergence}(iii) holds as long as $\sum_{j=1}^{\infty}\beta_j^{1-2/\eta} <\infty$; see Theorem 2.21(i) in \cite{FY05}.

\begin{assumption}\label{as: function psi}
(i) The weighting function $\psi$ has bounded variation. (ii) In addition, $\psi$ is continuously differentiable on $(0,u_0)$ and $(1-u_0,1)$ with bounded derivative.
\end{assumption}

Assumption \ref{as: function psi}(i) means that the function $\psi$ can be decomposed as $\psi = \psi_1 - \psi_2$, where both $\psi_1$ and $\psi_2$ are bounded and increasing functions. Assumption \ref{as: function psi} is thus satisfied in all our examples from Section \ref{sec: examples}.



\begin{assumption}\label{as: function F}
(i) The function $F$ is such that $F(s_1|x)< u_0/2$ and $F(s_2|x)> 1-u_0/2$ for all $x\in\mathcal X$. (ii) In addition, the function $u\mapsto F(u|x)$ is continuously differentiable on $u\in(s_1,s_2)$ with derivative $u\mapsto f(u|x)$ satisfying $f(u|x)\geq c$ for all $u\in(s_1,s_2)$ and $x\in\mathcal X$.
\end{assumption}

This assumption imposes mild regularity conditions on the conditional distribution of $Y$ given $X$. This assumption can be avoided if the function $\psi$ is smooth.

\begin{assumption}\label{as: estimator Fhat}
(i) The estimator $\widehat F$ is such that 
$$
\Pr\left(\sup_{x\in\mathcal X}\Delta(x) > u_0 / 2\right) \to 0.
$$
(ii) In addition,
\begin{equation}\label{eq: estimator Fhat - 1}
\sum_{t=T_1+1}^{T_1+T_2}\Ep\left[ (1 + \|X_t\|)^4\Delta(X_t)^4\mid D_1^{T_1}\right] = o_P(1).
\end{equation}
\end{assumption}
Assumption \ref{as: estimator Fhat} means that the estimator $\widehat F$ converges to the function $F$ sufficiently fast. To obtain some intuition about this assumption, note that under Assumptions \ref{as: beta mixing} and \ref{as: standard convergence}(i), it holds as long as $\sup_{x\in\mathcal X}\Delta(x) = o_P(T^{-1/4})$, which is plausible for nonparametric/machine learning estimators $\widehat F$. Note, however, that \eqref{eq: estimator Fhat - 1} does not actually require a bound on the supremum of the function $\Delta$ and instead uses a suitable weighted average value of this function, which is typically easier to bound.

We are now ready to state the main result of this section:

\begin{theorem}\label{thm: main result} Under Assumptions \ref{as: beta mixing}--\ref{as: estimator Fhat},
\begin{equation}\label{eq: main convergence 2}
\sqrt{T_2}(\widehat\beta - \beta) = \left(\frac{1}{T_2}\sum_{t=T_1+1}^{T_1+T_2}X_tX_t'\right)^{-1}\left(\frac{1}{\sqrt{T_2}}\sum_{t=T_1+1}^{T_1+T_2}X_te_t\right) + o_P(1)\to_d N(0,\Sigma),
\end{equation}
where $\Sigma = (\Ep[XX'])^{-1}\Omega(\Ep[XX'])^{-1}$.
\end{theorem}
\begin{remark}[Relation to Double/Debiased Machine Learning]
As discussed in Section \ref{sec: motivation}, our approach to estimation of weighted-average quantile regressions is an instance of the double/debiased machine learning method; e.g. \cite{CCDDHNR18}. In particular, we constructed the random variable $R$ in \eqref{eq: definition r new} so that (i) $\Ep[XR] = \Ep[XX']\beta$, meaning that the least squares projection of $R$ on $X$ correctly identifies $\beta$, and (ii) $\Ep[XR]$ is first-order insensitive with respect to perturbations of the function $F$, appearing in the definition of $R$. The latter condition, commonly referred to as the Neyman orthogonality, means that if we define
\begin{align*}
& R(\eta_1,\eta_2) = \int_0^{+\infty}(1-\Psi(\eta_1(s|X)))ds - \int_{-\infty}^0\Psi(\eta_1(s|X))ds \\
&\qquad\qquad\qquad + \int_{-\infty}^{+\infty}(\eta_1(s|X) - \mathbb I\{Y\leq s\})\eta_2(s|X)ds,
\end{align*}
for all functions $(s,x)\mapsto \eta_1(s|x)$ and $(s,x)\mapsto \eta_2(s|x)$, then $R(F,\psi(F)) = R$ and first-order Gateaux derivatives of the functions $\eta_1\mapsto \Ep[XR(\eta_1,\psi(F))]$ and $\eta_2\mapsto \Ep[XR(F,\eta_2)]$ at $\eta_1 = F$ and $\eta_2 = \psi(F)$, respectively, both vanish. It is this last condition that allows us to derive asymptotic normality of our estimator $\widehat\beta$ under weak conditions on the estimator $\widehat F$ of the function $F$, as specified in Assumption \ref{as: estimator Fhat} (in particular, we do not need to impose the common small bias condition). Our results do not follow from those in \cite{CCDDHNR18} because the function $\Psi$ is not necessarily continuously differentiable under our assumptions (and in fact has kinks in most examples from Section \ref{sec: examples}), and so the function $\eta_1 \mapsto \Ep[XR(\eta_1,\psi(F))]$ does not necessarily have the second-order Gateaux derivative, which is assumed to exist and is required to be suitably bounded in \cite{CCDDHNR18}.\footnote{Gateaux differentiability can be retained by assuming that the function $s\mapsto \widehat F(s|X)$ is increasing almost surely but this is unattractive because machine learning methods may or may not give increasing estimators and applying monotonization procedure may be computationally costly since the procedure would have to be carried out for each observation separately.} Instead, our results employ the smoothing properties of the integrals in the definition of $R$ in \eqref{eq: definition R}.
\qed
\end{remark}

\begin{remark}[Cross-Fitting and I.I.D. Setting]
Throughout this paper, we are assuming that the observations $D_1,\dots,D_T$ are coming from a time series under mixing conditions. Of course, this setting covers the case of i.i.d. observations as well. However, we can construct a more efficient estimator in the latter case via cross-fitting, e.g. \cite{CCDDHNR18}. Indeed, let $\widehat\beta_1$ be the estimator $\widehat\beta$ defined in \eqref{eq: risk regression estimator}. In addition, let
$$
\widehat\beta_2 = \left(\sum_{t=1}^{T_1}X_tX_t'\right)^{-1}\left( \sum_{t=1}^{T_1} X_t \widehat R_t\right),
$$
where $\widehat R_t$ is defined by \eqref{eq: rhat def} with the estimator $\widehat F$ being constructed using data from the second subsample, $D_{T_1+1},\dots,D_{T_1+T_2}$. It is then straightforward to show that the estimator $\widehat\beta = (\widehat\beta_1 + \widehat\beta_2)/2$ will satisfy 
$$
\sqrt{T}(\widehat\beta - \beta) = \left(\frac{1}{T_2}\sum_{t=1}^{T}X_tX_t'\right)^{-1}\left(\frac{1}{\sqrt{T}}\sum_{t=1}^{T}X_te_t\right) + o_P(1)\to_d N(0,\overline\Sigma),
$$
where
$$
\overline\Sigma = (\Ep[XX'])^{-1}\left(\Ep[e^2XX']\right)(\Ep[XX'])^{-1}.
$$
For estimation of $\overline\Sigma$ and construction of standard errors and confidence intervals, it is then possible to use the conventional Eicker-Huber-White formula.\qed
\end{remark}

Next, we consider consistent estimation of the covariance matrix $\Sigma$ appearing in Theorem \ref{thm: main result}. As discussed in the previous section, we focus on the Newey-West estimator:
\begin{equation}\label{eq: sigma estimation}
\widehat \Sigma = \left(\frac{1}{T_2}\sum_{t=T_1+1}^{T_1+T_2}X_tX_t'\right)^{-1}\left(\widehat\Omega_0 + \sum_{j=1}^m w(j,m)(\widehat\Omega_j + \widehat\Omega_j')\right)\left(\frac{1}{T_2}\sum_{t=T_1+1}^{T_1+T_2}X_tX_t'\right)^{-1},
\end{equation}
where $m$ is a tuning parameter, $w$ is a weighting function, and
\begin{equation}\label{eq: omega estimation}
\widehat\Omega_j = \frac{1}{T_2-T_1}\sum_{t=T_1+j+1}^{T_1+T_2} \widehat e_t\widehat e_{t-j}X_tX_{t-j}',\quad\text{for all }j=0,\dots,m.
\end{equation}
The tuning parameter $m$ is often chosen so that $m = m(T)\to\infty$ as $T\to\infty$ and the weighting function is typically defined by $w(j,m) = 1 - j/(m+1)$ for all $j=1,\dots,m$.

To prove consistency of the estimator $\widehat\Sigma$, we will need the following additional notation:
$$
\bar\Omega =  \bar\Omega_0 + \sum_{j=1}^m w(j,m)(\bar\Omega_j + \bar\Omega_j'),
$$
where
$$
\bar\Omega_j = \frac{1}{T_2-T_1}\sum_{t=T_1+j+1}^{T_1+T_2} e_t e_{t-j}X_tX_{t-j}',\quad\text{for all }j=0,\dots,m.
$$
We will impose the following assumptions.

\begin{assumption}\label{as: extra conditions}
(i) The matrix $\bar\Omega$ is consistent for $\Omega$: $\bar\Omega\to_P\Omega$. (ii) In addition, the weighting function $w$ is such that $0\leq w(j,m)\leq 1$ for all $j=1,\dots,m$. (iii) Moreover, the smoothing parameter $m$ is such that $m = o(T^{1/4})$.
\end{assumption}

Assumption \ref{as: extra conditions}(i) is a high-level condition that is familiar from the literature. Primitive conditions ensuring that this assumption is satisfied can be found in \cite{NW87}. Assumption \ref{as: extra conditions}(ii) is satisfied if we set $w(j,m) = 1 - j/(m+1)$, for example, which is a typical choice for the weighting function. Assumption \ref{as: extra conditions}(iii) is a mild growth condition meaning that the tuning parameter $m$ should not grow too fast as $T$ gets large.

In the next result, we prove consistency of the estimator $\widehat\Sigma$.

\begin{theorem}\label{thm: variance estimation}
Under Assumptions \ref{as: beta mixing}--\ref{as: extra conditions}, the estimator $\widehat \Sigma$ is consistent:
$$
\widehat\Sigma\to_P\Sigma.
$$
\end{theorem}

\begin{remark}[Transformations of $X$]
Although we focused on the case of linear weighted-average quantile regressions $\int_0^1 q_{Y|X}(u)\psi(u)du = X'\beta$ throughout the paper, inspecting the proofs reveals that our results equally apply to the more general case where we include transformations of $X$, such as interactions and other higher-order polynomial terms, on the right-hand side of the regression: $\int_0^1 q_{Y|X}(u)\psi(u)du = p(X)'\beta$, where $x\mapsto p(x) = (p_1(x),\dots,p_k(x))'$ is a vector of transformations. In this case, one should simply replace all $X_t$'s in \eqref{eq: risk regression estimator}, \eqref{eq: sigma estimation}, and \eqref{eq: omega estimation} by the corresponding $p(X_t)$'s. Theorems \ref{thm: main result} and \ref{thm: variance estimation} still apply in this case modulo obvious modifications.\qed
\end{remark}

\begin{remark}[Weighted-Average Quantile Regression Estimator as Best Linear Predictor]
When the weighted-average quantile function $x\mapsto \int_0^1 q_{Y|X=x}(u)\psi(u)du$ is not linear, i.e. there exists no $\beta$ such that $\int_0^1 q_{Y|X}(u)\psi(u)du = X'\beta$, it is straightforward to check that Theorems \ref{thm: main result} and \ref{thm: variance estimation} continue to hold with 
$$
\beta = \arg\min_{b\in\mathbb R^p} \Ep\left[\left(\int_0^1 q_{Y|X}(u)\psi(u)du - X'b\right)^2\right].
$$
It thus follows that our weighted-average quantile regression method consistently estimate parameters of the best linear approximation to $x\mapsto \int_0^1 q_{Y|X=x}(u)\psi(u)du$. In this sense, running our estimator makes sense even if it is believed that the linearity assumption of the regression model \eqref{eq: model} is satisfied not exactly but only approximately.
\qed
\end{remark}


\section{Monte Carlo Simulation Study}\label{sec: monte carlo}
In this section, we present results of a small-scale Monte Carlo simulation study that sheds light on finite-sample properties of the weighted-average quantile regression estimator.

We consider the following data-generating processes:
\begin{align*}
& \text{DGP1:}\qquad Y = \varepsilon - X'\bar\beta,\\
& \text{DGP2:}\qquad Y = (1 + 0.2X^1)\varepsilon - X'\bar\beta.
\end{align*}
Depending on the experiment, $X = (X^1,\dots,X^p)'$ and $\bar\beta = (\bar\beta_1,\dots,\bar\beta_p)'$ are vectors either in $\mathbb R^2$ or in $\mathbb R^5$, so that $p = 2$ or $5$, respectively. in the former case, we set $\bar\beta_2 = 0.5$ and vary $\bar\beta_1$ over $\{0, 0.3, 0.6, 0.9\}$. In the latter case, we set $\bar\beta_2 = 0.5$, $\bar\beta_3 = \bar\beta_4 = \bar\beta_5 = 0$, and again vary $\bar\beta_1$ over $\{0, 0.3, 0.6, 0.9\}$. We consider cases with $\varepsilon\sim N(0,1)$ and $\varepsilon \sim t(4)$. In both cases, $\varepsilon$ is independent of $X$. Note that in the latter case, Assumption \ref{as: standard convergence} is actually not satisfied, as random variables with the $t(4)$ distribution have finite moments up-to the 4th order but excluding the 4th order itself, and so this case serves as a way to check whether our methods continue to work if our asymptotic theory assumptions are slightly violated. Finally, we set $X = (X^1,\dots,X^p)'$ so that $X^1 = |\widetilde X^1|$ and $X^j = \widetilde X^j$ for all $j = 2,\dots,p$, where $\widetilde X = (\widetilde X^1,\dots,\widetilde X^p)'$ is a standard normal random vector in $\mathbb R^p$. For simplicity, we assume that the data $(X_1,Y_1),\dots,(X_T,Y_T)$ consists of $T$ i.i.d. realizations of the pair $(X,Y)$. We consider samples of size $T = 1000$ and $2000$.

As a machine learning estimator of the function $F$, we use a version of a random forest. Specifically, recall that any random forest estimator takes the weighted-average form, i.e. an estimator of $\Ep[V|Z = z]$ based on the data $(Z_1,V_1),\dots,(Z_T,V_T)$ will take the form $\widehat \Ep[V|Z = z] = \sum_{t=1}^T w_t V_t$. We do two changes to this estimator. First, once we have the weights $w_t$ from the random forest, we replace the weighted-average estimator by a local linear estimator:
$$
\widehat \Ep[V|Z = z] = z'\arg\min_b \sum_{t=1}^T w_t(V_t - Z_t'b)^2.
$$
This helps to improve the accuracy of the estimator, e.g. \cite{FTAW21}. Second, recall that we need an estimator of $x\mapsto \widehat F(s|x)$ for all $s\in\{Y_1,\dots,Y_{T_1}\}$, which is computationally straightforward but costly when $T_1$ is large. We therefore first split the interval $[\min_{1\leq t\leq T_1}Y_t,\max_{1\leq t\leq T_1}Y_t]$ into $\log T_1$ equal intervals, calculate random forest weights with $s$ being the center of each interval, and then apply the same weights for all $s$ in the same interval. This substantially reduces computational costs as we now need to calculate only $\log T_1$ random forests rather than $T_1$ of them. Closely related ideas were previously used in \cite{M06} who constructed a quantile random forest by applying a local linear quantile estimator based on weights obtained from a (mean) random forest. The main reason we rely on random forest estimators in this paper is that they are easy to train and allow for computational simplifications as explained here.

Also, for our simulations, we set $T_1 = 2T_2$, so that the random forest estimator uses twice as many observations as the OLS estimator. This is meaningful because random forest estimator, being a nonparametric estimator, has a much slower rate of convergence than that of OLS. In addition, to choose the number of leaves in each tree of the random forest estimator, we use sample splitting, namely we use $T/2$ observations to build random forest estimators corresponding to different number of leaves and we use remaining $T_1 - T/2 = T/6$ observations to choose the best random forest estimator according to the mean squared error criterion.

Note that both DGP1 and DGP2 satisfy our linear weighted-average quantile regression model \eqref{eq: model}. DGP1 corresponds to the homoscedastic case and yields $\beta$ in \eqref{eq: model} equal to $-\bar\beta\int_0^1 \psi(u)du$. For this data-generating process, we thus have $\beta = \bar\beta$ for the lower, middle, upper as well as exponential and polynomial regression models\footnote{Following Section \ref{ex: risk regression}, we define the exponential and polynomial regressions by \eqref{eq: model} with $\psi(u) = a\exp(-a(1-u))/(1-\exp(-a))$ for $a>0$ and $\psi(u) = a u^{a-1}$ for $a>1$, respectively.} and $\beta = 0_p$ for the inequality regression model. DGP2 corresponds to the heteroscedastic case and yields $\beta$ in \eqref{eq: model} such that $\beta_1 = 0.2\int_0^1 q_{\varepsilon}(u)\psi(u)du-\bar\beta_1\int_0^1\psi(u)du$ and $(\beta_2, \dots,  \beta_p)' = -(\bar\beta_2,\dots,\bar\beta)'\int_0^1\psi(u)du$, where $q_{\varepsilon}\colon[0,1]\to\mathbb R$ is the quantile function of the random variable $\varepsilon$. For this data-generating process, $(\beta_2,\dots,\beta_p)'$ thus coincides with the same vector under DGP1 but $\beta_1$ can only be calculated numerically.

For each specification of parameter values and each DGP, we repeat the experiment 500 times and estimate the coverage probability for the 90\% confidence interval for $\beta_1$ constructed using $t$-statistics. In addition, we estimate the mean absolute error $\Ep[|\widehat\beta_1 - \beta_1|]$. We present results for the coverage probability and the mean absolute error in Tables \ref{table: simulation results 1} and \ref{table: simulation results 2} of the Appendix, respectively. For each case, we give results for 4 regression models: upper, inequality, middle, and exponential regressions, which are denoted in the tables by $\psi$-type 1, 2, 3, and 4, respectively.

Overall, Table \ref{table: simulation results 1} shows that asymptotic theory from the previous section yields good approximation to the finite sample situation. In particular, the empirical coverage probability of 90\% confidence intervals is close to the nominal coverage probability. The only exception perhaps is the case of the upper and exponential regressions $(\psi = 1,4)$ with heteroscedastic noise, $T = 1000$, $p = 5$, and the $t(4)$ distribution, in which case the asymptotic confidence intervals undercover the true parameter values. However, the coverage improves substantially as we increase the sample size from $T = 1000$ to $T = 2000$. Table \ref{table: simulation results 2} also shows that the mean absolute error for the case with $p=5$ is similar to that for the case with $p = 2$, especially when $T = 2000$. This reinforces the conclusion that the asymptotic theory provides a good approximation to the finite-sample situation.

\section{Empirical Applications}\label{sec: empirical applications}

In this section, we apply our weighted-average quantile regression (WAQR) estimator in two empirical settings. In the first one, we focus on financial market data and study the expected shortfall regression. In the second one, we focus on wage data and study the inequality and social welfare regressions.

\subsection{Financial Market Data}

In this subsection, we apply the WAQR estimator in the asset pricing setting. We investigate the factor loadings of the risk measures of the industry returns. Although
our method is general, we focus on the expected shortfall
(ES), which is one of the most used risk measures in finance \cite[e.g., ][]{GL14, AB16, APPR17}. In this case, our WAQR can be referred to as the expected shortfall regression estimator.

\subsubsection{Expected Shortfall Regression Estimator}

We investigate the factor structures of the 10\% expected shortfalls
of the Fama-French 5 industries. We use the Fama-French 5-factor model
standard in the literature to capture industries' expected shortfalls \cite[e.g., ][]{FF15, FF16}.
Table \ref{tab:Industry-Factor-Loadings} reports the factor exposure
results for the 10\% ES regressions. For comparison, the table
also shows the results based on the mean regression and the 10\%
quantile regression.\footnote{The estimates in the mean and 10\% quantile regressions are multiplied
by -1 to be consistent with the risk regressions.}

The point estimates to the market excess returns based on the 10\%
ES regression are negative and statistically significant at -1.196,
-1.293, -1.400, -1.153, and -1.330 for the consumer, manufacturing,
high tech, health, and other industries, respectively. The negative
coefficients imply that the 10\% ES of the industry returns are
lower when the market excess returns are high. The point estimates for the 10\% ES regression are slightly larger in magnitude
than those of mean regressions and 10\% quantile regressions.

The exposures to the other four factors can substantially differ across
the mean, quantile, and risk measures. For the size factor, the magnitude
of the exposure for the manufacturing  industry is similar across
the mean, quantile, and risk measures. The magnitudes of the exposures
for the other industries are consistently larger for the risk measures relative to those from the mean and quantile regressions. For the value factor, the magnitudes of the exposures for the consumer,
high tech, and health industries are similar across the mean, quantile,
and risk measures. However, the magnitudes of exposures are markedly
higher for the risk measures than for the means and 10\% quantiles
for the manufacturing and the other industries. For example, for the
manufacturing industry, the coefficient estimate to the value factor
is -0.514 for the 10\% ES regression, while the coefficient estimates
are -0.110 and -0.140 for the mean and 10\% quantile, respectively.
For the profitability factor, the magnitudes of the coefficient estimates
for the 10\% ES regressions are generally larger to those of
the mean and 10\% quantile regressions. For example, for the consumer
industry, the coefficient estimate to the profitability factor is
-0.779 based on the 10\% ES regression, while the estimates are
-0.197 and -0.195 for the mean and 10\% quintile, respectively. For
the investment factor, the magnitudes of the coefficient estimates
for the 10\% ES are consistently higher than those of the mean and
10\% quantile regressions for the consumer, the high tech, and the
other industries.\footnote{Table \ref{tab:newey-west} in the Online Appendix further shows the
results of Table \ref{tab:Industry-Factor-Loadings} and uses the
Newey-West procedure to adjust the standard errors. The significance
levels are largely unchanged with the adjustment.}

In summary, the baseline results show that the 10\% ES of the industry
returns are highly exposed to the factors that are designed to explain
the mean returns. The results highlight the different dynamics between
the risk measures and the means or quantiles of returns and the importance
of the factors in capturing the variations of the risk measures
of the portfolio returns.

We further study the time-varying exposures of the 10\% ES of the
industry returns to the Fama-French 5 factors. The period used for
estimation is the past twenty years and we roll the estimation period
every year. The results are summarized in Figure \ref{fig:Rolling-10-Year}.
The exposures of the 10\% ES to the market factor have a downward
spike for the industries around the internet bubble period, implying
that the exposures of the 10\% ES to the market factor increase during
this period.

The exposures of the 10\% ES to the other factors vary substantially
during the sample period. We discuss several examples. For the health
industry, the exposures of the 10\% ES to the value factor are consistently
positive, suggesting that the risk of the health industry as measured
by 10\% ES increases when the value premium is high. The magnitude
of the coefficient estimate increases substantially to around three
from 1990 to early 2000. For the profitability factor, the exposures
of the 10\% ES of the high tech industry increase drastically in the
200s, but decrease to the pre-2000 level since 2010. For the investment
factor, the health industry tends to have negative exposures of its
10\% ES, while the high tech industry tends to have positive exposures
of their 10\% ES. In other words, the risk measured by the 10\% ES
of the health industry decreases when the investment premium is high,
while the risk measured by 10\% ES of the high tech industry increases
when the investment premium is high. The magnitudes of the 10\% ES
of the industries all increased during the early 2000s or the burst
of the internet bubble period.

The time-series results suggest that exposures of the 10\% ES to the
factors varied substantially during the 1990s to the early 2000 period,
coinciding with the beginning and the subsequent burst of the internet
bubble. However, the exposures were relatively stable during the 2008-2009
financial crisis period, although the market experienced drastic volatility
during the crisis period.

\subsubsection{Parametric Estimator}

As discussed above, a (potentially inconsistent) alternative to our WAQR estimator is a parametric estimator. This alternative method estimates exposures of risk measures to a set of covariates by taking the weighted average of the point estimates from the individual quantile regressions. Here, we compare our estimation results with those based on the parametric estimator. The results are documented in
Table \ref{tab:naive_method}.\footnote{The individual quantile regression results from 1\% to 10\% quantiles
are reported in Figure \ref{fig:Coefficient-by-Percentiles}.} Table \ref{tab:naive_method} shows the estimation results for the
10\% ES using the parametric estimator alongside those from our WAQR estimator.

For the exposures to the market factor, the point estimates of the
10\% ES to the market factor from the WAQR estimator are slightly larger than those based on the parametric estimators. However, the
coefficient estimates of the 10\% ES based on the parametric estimator
and the WAQR estimator can differ significantly for the other factors.
We provide several examples of the differences. For the high tech
industry, the coefficient estimate of the 10\% ES to the size factor
is 0.119 based on the WAQR estimator but is 0.062 based on the parametric estimator, which is close in magnitude to that of the
mean and 10\% quantile regressions. For the manufacturing industry,
the coefficient estimate of the 10\% ES to the value factor is -0.514
based on the WAQR estimator but is only -0.084 based on the parametric estimator. Again, the coefficient estimate based on the
parametric estimator is close to those based on the mean and 10\%
quantile regressions.

Overall, we find that the WAQR and parametric
estimators can differ substantially. In particular, when the exposures of the risk
measures differ from those of the mean and quantiles, the parametric
estimator tends to underestimate the exposures. The magnitudes of the coefficient
estimates based on the parametric estimator tend to fall between those
based on the WAQR estimator and those from mean and quantile regressions.

Furthermore, we investigate the underlying reasons behind this discrepancy
between our WAQR estimator and the simple parametric estimator. Relative to the parametric estimator, an important assumption our WAQR estimator relaxed is the assumption that the quantiles are linear in the covariates. So far, we have shown that the WAQR estimator and the parametric estimator tend to provide different coefficient estimates in financial data. We directly test whether the differences stem from the violation
of the linearity assumption the parametric method imposes.

We test whether the 10\% quantile and the 5\% quantile of the industry
returns are significantly exposed to the higher moments and interactions
of the Fama-French 5 factors. We document the results in Table \ref{tab:Quantile-Regression-higher-order}
in the Online Appendix. Panel A of the table shows the quantile regression
results of regressing the industry returns to the first, second, and
third moments of the Fama-French 5 factors. Inconsistent with the
assumption that the quantiles are linear with the covariates, we find
that the 10\% and 5\% quantiles of the industry returns are significantly
exposed to many of the second and third moments of the Fama-French
5 factors. Panel B of the table reports the quantile regression results
of regressing the industry returns to the standalone and interactions
of the Fama-French 5 factors. Again, inconsistent with the assumption
that the quantiles are linear with the covariates, we show that the
10\% and 5\% quantiles of the industry returns are significantly exposed
to a number of the interaction terms of the Fama-French 5 factors.

Overall, we conclude that the discrepancies of the results between
the WAQR estimator and the parametric estimator stem
from the fact that the quantiles are not linear in the covariates, and that estimates from the parametric estimator method are not reliable
in the financial setting.

\subsection{Wage Data}

In this subsection, we apply our method to study wage inequality and social welfare. We start with the wage inequality.

\subsubsection{Inequality Regression}
 We consider
several standard individual characteristics in the literature when
wage or consumption is studied \cite[e.g., ][]{ACF06, BPP08, AP16}, including family size, an indicator
variable for no children, age, and education. The sample goes from
2001 to 2018. The sample and variables are discussed in detail in
Appendix \ref{sec: data}. We apply the WAQR estimator for the inequality regression each year using all the independent variables, and show coefficient estimates in Figure \ref{fig:Coefficient-by-Year-Inequality-1}. For comparison, we also show the estimates based on a parametric estimator which is the differences
of the coefficient estimates for the 90\% quantile and the 10\% quantile
regressions.

We discuss the time trends of the coefficient estimates based on the WAQR estimator. The coefficient estimates for the family size
decrease over time, going from positive to significantly negative
since 2011. That is, the inequality, or the average wage difference
between the top and bottom of the distribution, decreases in family
size in the latter part of the sample. When the parametric estimator
is used, the point estimates slightly decrease over time but stay
positive even towards the end of the sample. The WAQR coefficient estimates
for the indicator variable of no children increase over time, going
from significantly negative to insignificantly positive. The coefficient
estimates for age are relatively stable over time. The point estimates
based on the WAQR estimator are relatively similar to those
based on the parametric estimator for these two variables. 

The point WAQR estimates of education stay significantly positive over
time, suggesting that inequality
increases with education. When the parametric estimator is used the
point estimates are also positive for all years. However, the point
estimates based on the parametric estimator are markedly lower than those
based on the WAQR estimator before 2010. For example, the
magnitude of the point estimate based on the parametric estimator is only
half of that based on the WAQR estimator  for year 2001. In
the latter part of the sample, the point estimates based on the two
methods are relatively similar.

\subsubsection{Social Welfare Regression}

We now apply our WAQR estimator to study the relationship
between the weighted average wage and the individual characteristics. We assume that the weights
are exponential with more weights being placed to the lower income (specifically, we use the exponential weighting function $\psi$ from Section \ref{ex: risk regression} with $u$ replaced by $1-u$ and $a = 10$).

We report the WAQR estimation results in Figure \ref{fig:Coefficient-by-Year-Lowincome} below. For comparison, we also provide the mean regression (OLS) results. The blue line shows the point estimates based on the WAQR estimator over time, while the red line documents the point estimates for the mean regression (OLS).

We discuss the time trends of the coefficient estimates based on the WAQR. The point estimates for the family
size variable increase over time, going from insignificantly negative to positive.
For the mean regressions, the point estimates also increase over time
but stay negative even towards the end of the sample. The point
estimates for the indicator variable of no children decrease over
time, and are lower than those based on the mean regression
in the latter part of the sample. The point estimates based on the WAQR are similar to those based on the mean regressions for the age variable.

The point estimates for the education variable stay significantly positive over
time based on the WAQR, suggesting that the average
wage of the low income group increases with education. When OLS is
used, the point estimates are also positive for all years. However,
the point estimates based on the WAQR are significantly
lower than those based on the mean regression before 2010.
The point estimates converge towards the latter part of the sample.

\section{Conclusion}
We introduce the weighted-average quantile regression that significantly generalizes the commonly used mean and quantile regressions. We develop estimators of such regressions that are straightforward to apply in a variety of empirical settings. In the examples of risk, inequality, and social welfare regressions, the weighted-average quantile regression estimators yield results that are different from those based on both mean and quantile regression methods.

\newpage

\begin{landscape}
\thispagestyle{empty}
\begin{table}[H]
{\caption{\small{}Industry Factor Loadings\label{tab:Industry-Factor-Loadings}}
}

\caption*{
This table shows the industry factor loadings for the Fama-French
5 industries. The ``Mean'' rows reports the results for mean regressions.
The ``10\% VAR'' rows report the results for 10\% quantile regressions.
The ``10\% ES'' rows report results for the 10\% ES risk regressions.
The estimates in the mean and 10\% quantile regressions are multiplied
by -1 to be consistent with the risk regressions. The Fama-French
5-factor model is used. Standard errors are reported in parentheses.
}

\begin{tabular}{lccccccccccccc}
\hline 
{\scriptsize{}Cnsmr} & {\scriptsize{}MKTRF} &  & {\scriptsize{}SMB} &  & {\scriptsize{}HML} &  & {\scriptsize{}RMW} &  & {\scriptsize{}CMA} &  & {\scriptsize{}Constant} &  & {\scriptsize{}$R^{2}$}\tabularnewline
\hline 
{\scriptsize{}Mean} & {\scriptsize{}-0.869} & {\scriptsize{}(0.004)} & {\scriptsize{}-0.059} & {\scriptsize{}(0.008)} & {\scriptsize{}0.165} & {\scriptsize{}(0.007)} & {\scriptsize{}-0.278} & {\scriptsize{}(0.012)} & {\scriptsize{}-0.197} & {\scriptsize{}(0.015)} & {\scriptsize{}-0.003} & {\scriptsize{}(0.005)} & {\scriptsize{}0.904}\tabularnewline
{\scriptsize{}10\% VAR} & {\scriptsize{}-0.873} & {\scriptsize{}(0.008)} & {\scriptsize{}-0.074} & {\scriptsize{}(0.016)} & {\scriptsize{}0.177} & {\scriptsize{}(0.014)} & {\scriptsize{}-0.294} & {\scriptsize{}(0.023)} & {\scriptsize{}-0.195} & {\scriptsize{}(0.030)} & {\scriptsize{}0.374} & {\scriptsize{}(0.009)} & \tabularnewline
{\scriptsize{}10\% ES} & {\scriptsize{}-1.196} & {\scriptsize{}(0.022)} & {\scriptsize{}-0.122} & {\scriptsize{}(0.043)} & {\scriptsize{}0.214} & {\scriptsize{}(0.039)} & {\scriptsize{}-0.779} & {\scriptsize{}(0.061)} & {\scriptsize{}-0.250} & {\scriptsize{}(0.081)} & {\scriptsize{}0.909} & {\scriptsize{}(0.025)} & {\scriptsize{}0.388}\tabularnewline
\hline 
{\scriptsize{}Manuf} & {\scriptsize{}MKTRF} &  & {\scriptsize{}SMB} &  & {\scriptsize{}HML} &  & {\scriptsize{}RMW} &  & {\scriptsize{}CMA} &  & {\scriptsize{}Constant} &  & {\scriptsize{}$R^{2}$}\tabularnewline
\hline 
{\scriptsize{}Mean} & {\scriptsize{}-1.026} & {\scriptsize{}(0.006)} & {\scriptsize{}-0.075} & {\scriptsize{}(0.011)} & {\scriptsize{}-0.110} & {\scriptsize{}(0.010)} & {\scriptsize{}-0.440} & {\scriptsize{}(0.016)} & {\scriptsize{}-0.182} & {\scriptsize{}(0.021)} & {\scriptsize{}0.008} & {\scriptsize{}(0.007)} & {\scriptsize{}0.879}\tabularnewline
{\scriptsize{}10\% VAR} & {\scriptsize{}-0.997} & {\scriptsize{}(0.012)} & {\scriptsize{}-0.101} & {\scriptsize{}(0.023)} & {\scriptsize{}-0.140} & {\scriptsize{}(0.021)} & {\scriptsize{}-0.406} & {\scriptsize{}(0.033)} & {\scriptsize{}-0.099} & {\scriptsize{}(0.043)} & {\scriptsize{}0.527} & {\scriptsize{}(0.013)} & \tabularnewline
{\scriptsize{}10\% ES} & {\scriptsize{}-1.293} & {\scriptsize{}(0.029)} & {\scriptsize{}-0.066} & {\scriptsize{}(0.057)} & {\scriptsize{}-0.514} & {\scriptsize{}(0.051)} & {\scriptsize{}-0.832} & {\scriptsize{}(0.081)} & {\scriptsize{}0.122} & {\scriptsize{}(0.106)} & {\scriptsize{}1.188} & {\scriptsize{}(0.033)} & {\scriptsize{}0.345}\tabularnewline
\hline 
{\scriptsize{}Hitec} & {\scriptsize{}MKTRF} &  & {\scriptsize{}SMB} &  & {\scriptsize{}HML} &  & {\scriptsize{}RMW} &  & {\scriptsize{}CMA} &  & {\scriptsize{}Constant} &  & {\scriptsize{}$R^{2}$}\tabularnewline
\hline 
{\scriptsize{}Mean} & {\scriptsize{}-1.074} & {\scriptsize{}(0.005)} & {\scriptsize{}0.068} & {\scriptsize{}(0.010)} & {\scriptsize{}0.316} & {\scriptsize{}(0.009)} & {\scriptsize{}0.348} & {\scriptsize{}(0.014)} & {\scriptsize{}-0.045} & {\scriptsize{}(0.018)} & {\scriptsize{}-0.008} & {\scriptsize{}(0.006)} & {\scriptsize{}0.922}\tabularnewline
{\scriptsize{}10\% VAR} & {\scriptsize{}-1.075} & {\scriptsize{}(0.011)} & {\scriptsize{}0.056} & {\scriptsize{}(0.021)} & {\scriptsize{}0.339} & {\scriptsize{}(0.019)} & {\scriptsize{}0.370} & {\scriptsize{}(0.030)} & {\scriptsize{}-0.049} & {\scriptsize{}(0.040)} & {\scriptsize{}0.421} & {\scriptsize{}(0.012)} & \tabularnewline
{\scriptsize{}10\% ES} & {\scriptsize{}-1.400} & {\scriptsize{}(0.027)} & {\scriptsize{}0.119} & {\scriptsize{}(0.052)} & {\scriptsize{}0.297} & {\scriptsize{}(0.046)} & {\scriptsize{}1.125} & {\scriptsize{}(0.074)} & {\scriptsize{}-0.280} & {\scriptsize{}(0.097)} & {\scriptsize{}1.082} & {\scriptsize{}(0.030)} & {\scriptsize{}0.445}\tabularnewline
\hline 
{\scriptsize{}Hlth} & {\scriptsize{}MKTRF} &  & {\scriptsize{}SMB} &  & {\scriptsize{}HML} &  & {\scriptsize{}RMW} &  & {\scriptsize{}CMA} &  & {\scriptsize{}Constant} &  & {\scriptsize{}$R^{2}$}\tabularnewline
\hline 
{\scriptsize{}Mean} & {\scriptsize{}-0.816} & {\scriptsize{}(0.007)} & {\scriptsize{}0.021} & {\scriptsize{}(0.013)} & {\scriptsize{}0.319} & {\scriptsize{}(0.012)} & {\scriptsize{}0.084} & {\scriptsize{}(0.019)} & {\scriptsize{}-0.124} & {\scriptsize{}(0.025)} & {\scriptsize{}-0.003} & {\scriptsize{}(0.008)} & {\scriptsize{}0.765}\tabularnewline
{\scriptsize{}10\% VAR} & {\scriptsize{}-0.803} & {\scriptsize{}(0.013)} & {\scriptsize{}0.016} & {\scriptsize{}(0.025)} & {\scriptsize{}0.327} & {\scriptsize{}(0.022)} & {\scriptsize{}0.114} & {\scriptsize{}(0.035)} & {\scriptsize{}-0.098} & {\scriptsize{}(0.046)} & {\scriptsize{}0.592} & {\scriptsize{}(0.014)} & \tabularnewline
{\scriptsize{}10\% ES} & {\scriptsize{}-1.153} & {\scriptsize{}(0.026)} & {\scriptsize{}0.096} & {\scriptsize{}(0.051)} & {\scriptsize{}0.251} & {\scriptsize{}(0.045)} & {\scriptsize{}0.190} & {\scriptsize{}(0.072)} & {\scriptsize{}-0.045} & {\scriptsize{}(0.094)} & {\scriptsize{}1.207} & {\scriptsize{}(0.029)} & {\scriptsize{}0.322}\tabularnewline
\hline 
{\scriptsize{}Other} & {\scriptsize{}MKTRF} &  & {\scriptsize{}SMB} &  & {\scriptsize{}HML} &  & {\scriptsize{}RMW} &  & {\scriptsize{}CMA} &  & {\scriptsize{}Constant} &  & {\scriptsize{}$R^{2}$}\tabularnewline
\hline 
{\scriptsize{}Mean} & {\scriptsize{}-1.060} & {\scriptsize{}(0.004)} & {\scriptsize{}0.005} & {\scriptsize{}(0.007)} & {\scriptsize{}-0.632} & {\scriptsize{}(0.007)} & {\scriptsize{}0.153} & {\scriptsize{}(0.011)} & {\scriptsize{}0.311} & {\scriptsize{}(0.014)} & {\scriptsize{}0.002} & {\scriptsize{}(0.004)} & {\scriptsize{}0.962}\tabularnewline
{\scriptsize{}10\% VAR} & {\scriptsize{}-1.059} & {\scriptsize{}(0.007)} & {\scriptsize{}0.012} & {\scriptsize{}(0.013)} & {\scriptsize{}-0.604} & {\scriptsize{}(0.012)} & {\scriptsize{}0.175} & {\scriptsize{}(0.018)} & {\scriptsize{}0.308} & {\scriptsize{}(0.024)} & {\scriptsize{}0.338} & {\scriptsize{}(0.007)} & \tabularnewline
{\scriptsize{}10\% ES} & {\scriptsize{}-1.330} & {\scriptsize{}(0.033)} & {\scriptsize{}0.349} & {\scriptsize{}(0.064)} & {\scriptsize{}-1.685} & {\scriptsize{}(0.058)} & {\scriptsize{}0.278} & {\scriptsize{}(0.092)} & {\scriptsize{}0.765} & {\scriptsize{}(0.120)} & {\scriptsize{}1.043} & {\scriptsize{}(0.037)} & {\scriptsize{}0.427}\tabularnewline
\hline 
\end{tabular}{\scriptsize\par}
\end{table}
\end{landscape}

\begin{landscape}
\thispagestyle{empty}
\begin{table}[H]
{\caption{\small{}Compare with Parametric Estimator\label{tab:naive_method}}
}

\caption*{
This table shows the industry factor loadings for the 10\% ES regression for the Fama-French 5 industries based on the parametric estimator and the WAQR estimator. The ``10\% Parametric'' rows report results using the parametric
estimator. The ``10\% ES'' rows report results using the WAQR estimator. The Fama-French 5-factor model is used.  Standard errors are reported in parentheses.
}

\begin{tabular}{lccccccccccccc}
\hline 
{\scriptsize{}Cnsmr} & {\scriptsize{}MKTRF} &  & {\scriptsize{}SMB} &  & {\scriptsize{}HML} &  & {\scriptsize{}RMW} &  & {\scriptsize{}CMA} &  & {\scriptsize{}Constant} &  & {\scriptsize{}$R^{2}$}\tabularnewline
\hline 
{\scriptsize{}10\% Parametric} & {\scriptsize{}-0.862} &  & {\scriptsize{}-0.094} &  & {\scriptsize{}0.179} &  & {\scriptsize{}-0.294} &  & {\scriptsize{}-0.181} &  & {\scriptsize{}0.548} &  & \tabularnewline
{\scriptsize{}10\% ES} & {\scriptsize{}-1.196} & {\scriptsize{}(0.022)} & {\scriptsize{}-0.122} & {\scriptsize{}(0.043)} & {\scriptsize{}0.214} & {\scriptsize{}(0.039)} & {\scriptsize{}-0.779} & {\scriptsize{}(0.061)} & {\scriptsize{}-0.250} & {\scriptsize{}(0.081)} & {\scriptsize{}0.909} & {\scriptsize{}(0.025)} & {\scriptsize{}0.388}\tabularnewline
\hline 
{\scriptsize{}Manuf} & {\scriptsize{}MKTRF} &  & {\scriptsize{}SMB} &  & {\scriptsize{}HML} &  & {\scriptsize{}RMW} &  & {\scriptsize{}CMA} &  & {\scriptsize{}Constant} &  & {\scriptsize{}$R^{2}$}\tabularnewline
\hline 
{\scriptsize{}10\% Parametric} & {\scriptsize{}-1.016} &  & {\scriptsize{}-0.111} &  & {\scriptsize{}-0.084} &  & {\scriptsize{}-0.440} &  & {\scriptsize{}-0.038} &  & {\scriptsize{}0.782} &  & \tabularnewline
{\scriptsize{}10\% ES} & {\scriptsize{}-1.293} & {\scriptsize{}(0.029)} & {\scriptsize{}-0.066} & {\scriptsize{}(0.057)} & {\scriptsize{}-0.514} & {\scriptsize{}(0.051)} & {\scriptsize{}-0.832} & {\scriptsize{}(0.081)} & {\scriptsize{}0.122} & {\scriptsize{}(0.106)} & {\scriptsize{}1.188} & {\scriptsize{}(0.033)} & {\scriptsize{}0.345}\tabularnewline
\hline 
{\scriptsize{}Hitec} & {\scriptsize{}MKTRF} &  & {\scriptsize{}SMB} &  & {\scriptsize{}HML} &  & {\scriptsize{}RMW} &  & {\scriptsize{}CMA} &  & {\scriptsize{}Constant} &  & {\scriptsize{}$R^{2}$}\tabularnewline
\hline 
{\scriptsize{}10\% Parametric} & {\scriptsize{}-1.077} &  & {\scriptsize{}0.062} &  & {\scriptsize{}0.360} &  & {\scriptsize{}0.440} &  & {\scriptsize{}-0.070} &  & {\scriptsize{}0.639} &  & \tabularnewline
{\scriptsize{}10\% ES} & {\scriptsize{}-1.400} & {\scriptsize{}(0.027)} & {\scriptsize{}0.119} & {\scriptsize{}(0.052)} & {\scriptsize{}0.297} & {\scriptsize{}(0.046)} & {\scriptsize{}1.125} & {\scriptsize{}(0.074)} & {\scriptsize{}-0.280} & {\scriptsize{}(0.097)} & {\scriptsize{}1.082} & {\scriptsize{}(0.030)} & {\scriptsize{}0.445}\tabularnewline
\hline 
{\scriptsize{}Hlth} & {\scriptsize{}MKTRF} &  & {\scriptsize{}SMB} &  & {\scriptsize{}HML} &  & {\scriptsize{}RMW} &  & {\scriptsize{}CMA} &  & {\scriptsize{}Constant} &  & {\scriptsize{}$R^{2}$}\tabularnewline
\hline 
{\scriptsize{}10\% Parametric} & {\scriptsize{}-0.813} &  & {\scriptsize{}0.009} &  & {\scriptsize{}0.342} &  & {\scriptsize{}0.110} &  & {\scriptsize{}-0.128} &  & {\scriptsize{}0.900} &  & \tabularnewline
{\scriptsize{}10\% ES} & {\scriptsize{}-1.153} & {\scriptsize{}(0.026)} & {\scriptsize{}0.096} & {\scriptsize{}(0.051)} & {\scriptsize{}0.251} & {\scriptsize{}(0.045)} & {\scriptsize{}0.190} & {\scriptsize{}(0.072)} & {\scriptsize{}-0.045} & {\scriptsize{}(0.094)} & {\scriptsize{}1.207} & {\scriptsize{}(0.029)} & {\scriptsize{}0.322}\tabularnewline
\hline 
{\scriptsize{}Other} & {\scriptsize{}MKTRF} &  & {\scriptsize{}SMB} &  & {\scriptsize{}HML} &  & {\scriptsize{}RMW} &  & {\scriptsize{}CMA} &  & {\scriptsize{}Constant} &  & {\scriptsize{}$R^{2}$}\tabularnewline
\hline 
{\scriptsize{}10\% Parametric} & {\scriptsize{}-1.050} &  & {\scriptsize{}0.000} &  & {\scriptsize{}-0.609} &  & {\scriptsize{}0.182} &  & {\scriptsize{}0.284} &  & {\scriptsize{}0.486} &  & \tabularnewline
{\scriptsize{}10\% ES} & {\scriptsize{}-1.330} & {\scriptsize{}(0.033)} & {\scriptsize{}0.349} & {\scriptsize{}(0.064)} & {\scriptsize{}-1.685} & {\scriptsize{}(0.058)} & {\scriptsize{}0.278} & {\scriptsize{}(0.092)} & {\scriptsize{}0.765} & {\scriptsize{}(0.120)} & {\scriptsize{}1.043} & {\scriptsize{}(0.037)} & {\scriptsize{}0.427}\tabularnewline
\hline 
\end{tabular}{\scriptsize\par}
\end{table}
\end{landscape}

\begin{figure}[H]
\caption{Rolling 20 Year Coefficient\label{fig:Rolling-10-Year}}

{\tiny{}\medskip{}}
\caption*{
This figure shows the WAQR estimates for the time-varying exposures of the 10\% ES of the industry returns to the Fama-French 5 factors. The industries include
the Fama-French 5 industries: consumer, manufacturing, high tech,
health, and other. The period used for estimation is the past 20 years and the
estimation period is rolled over every year.
}
{\tiny{}\medskip{}}
\begin{centering}
\subfloat{\includegraphics[width=7cm]{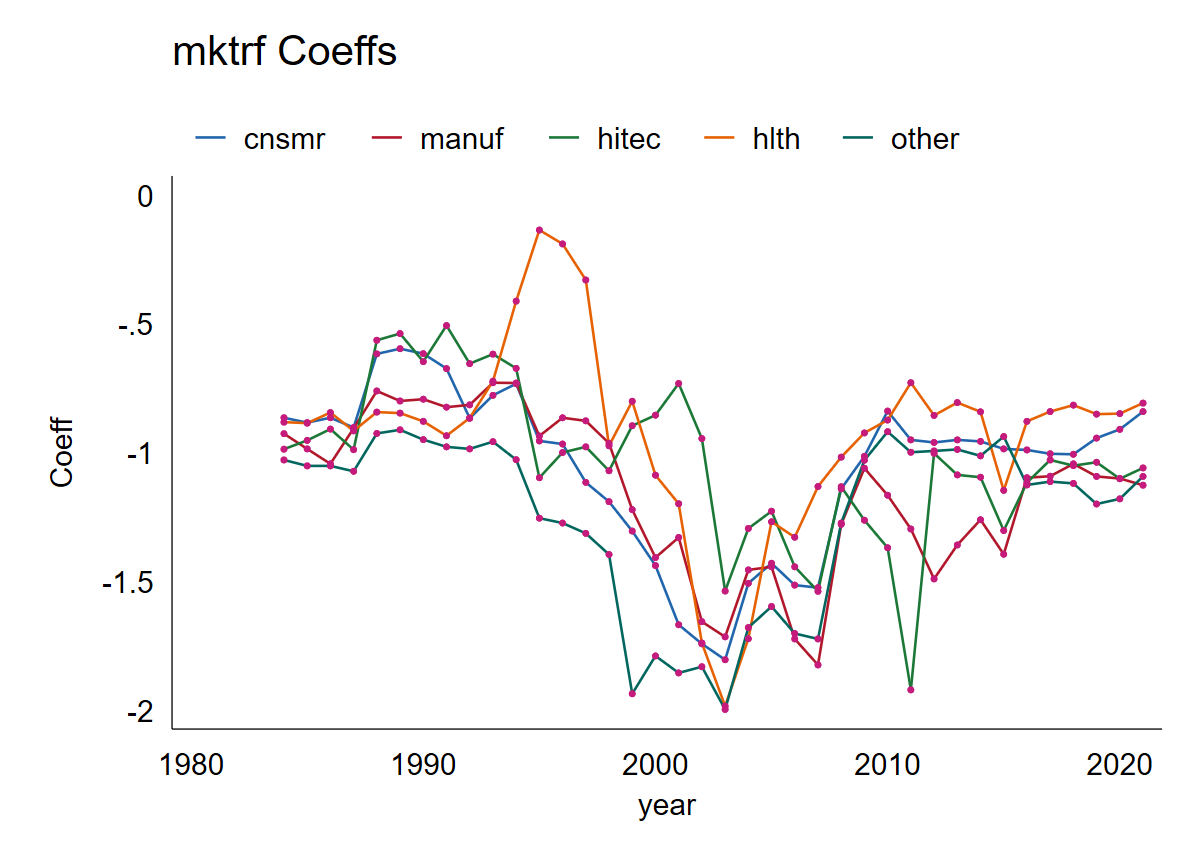}}\subfloat{\includegraphics[width=7cm]{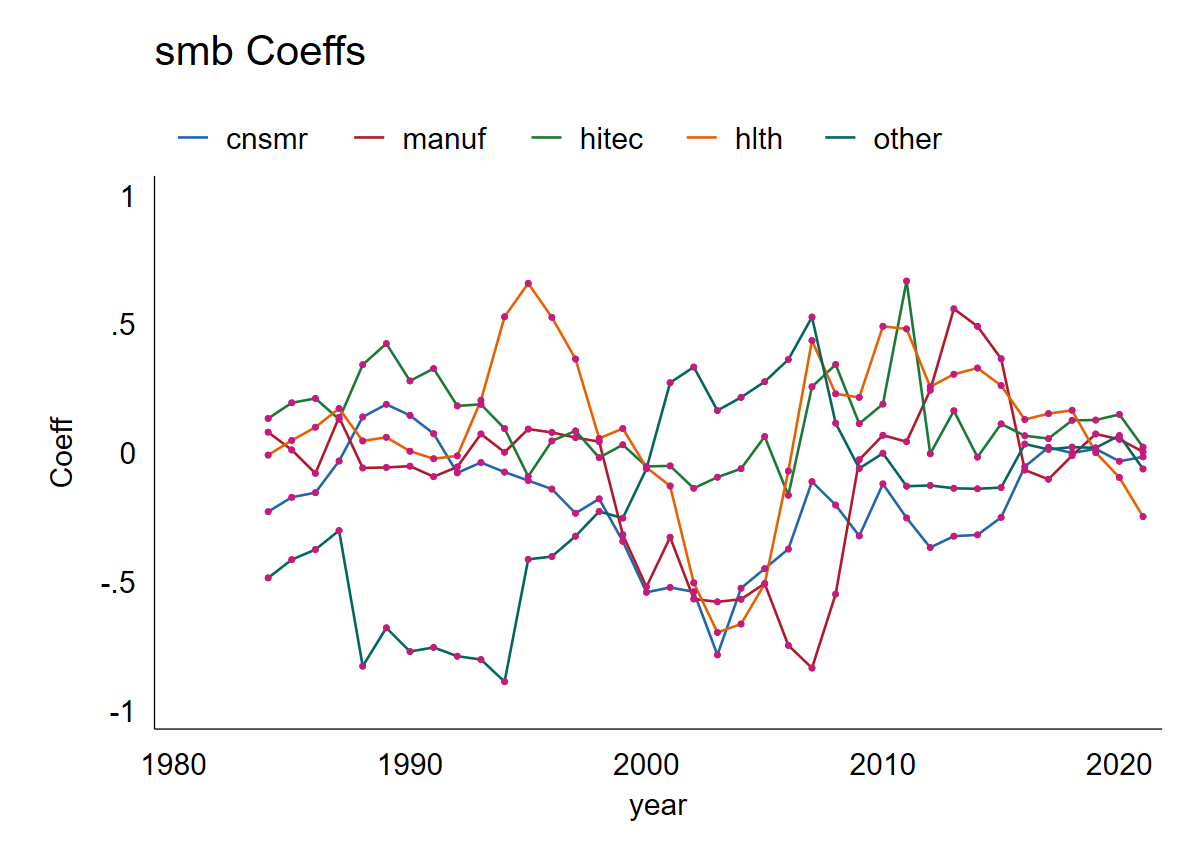}}
\end{centering}
\begin{centering}
\subfloat{\includegraphics[width=7cm]{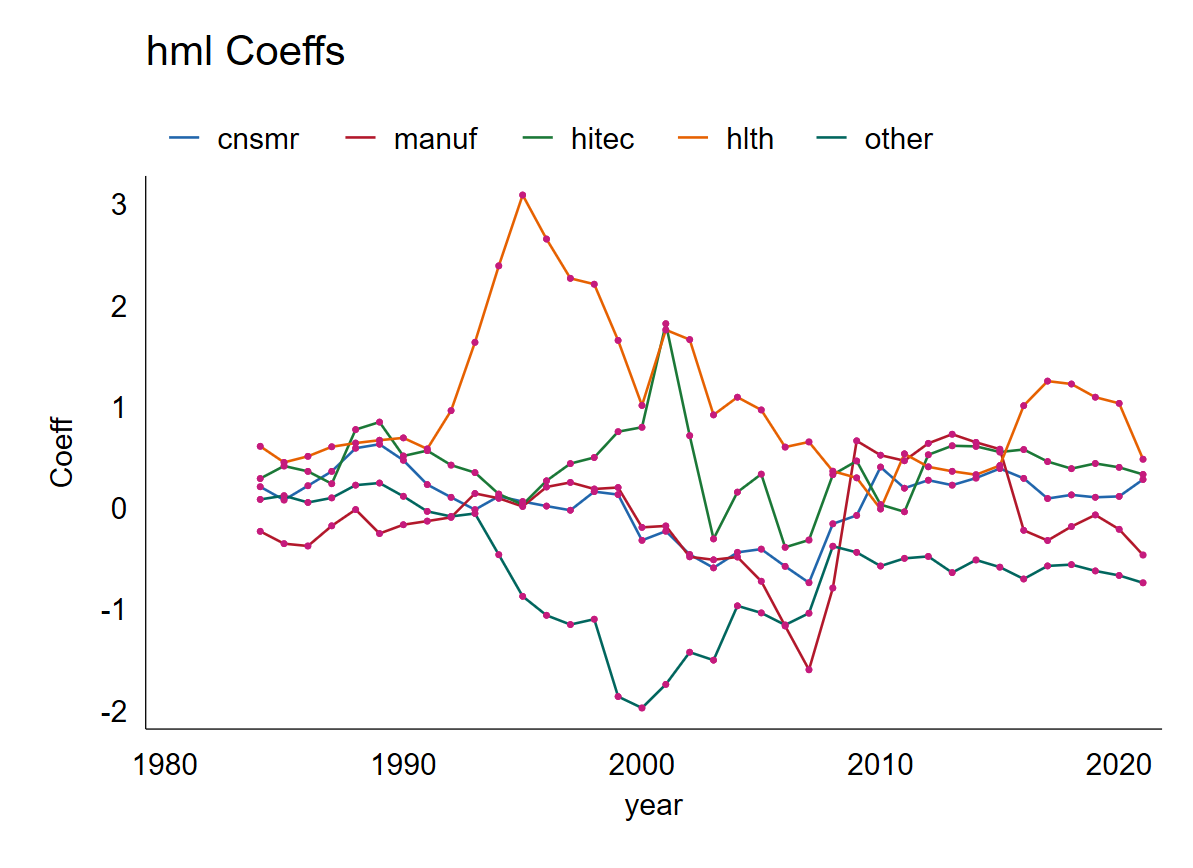}}\subfloat{\includegraphics[width=7cm]{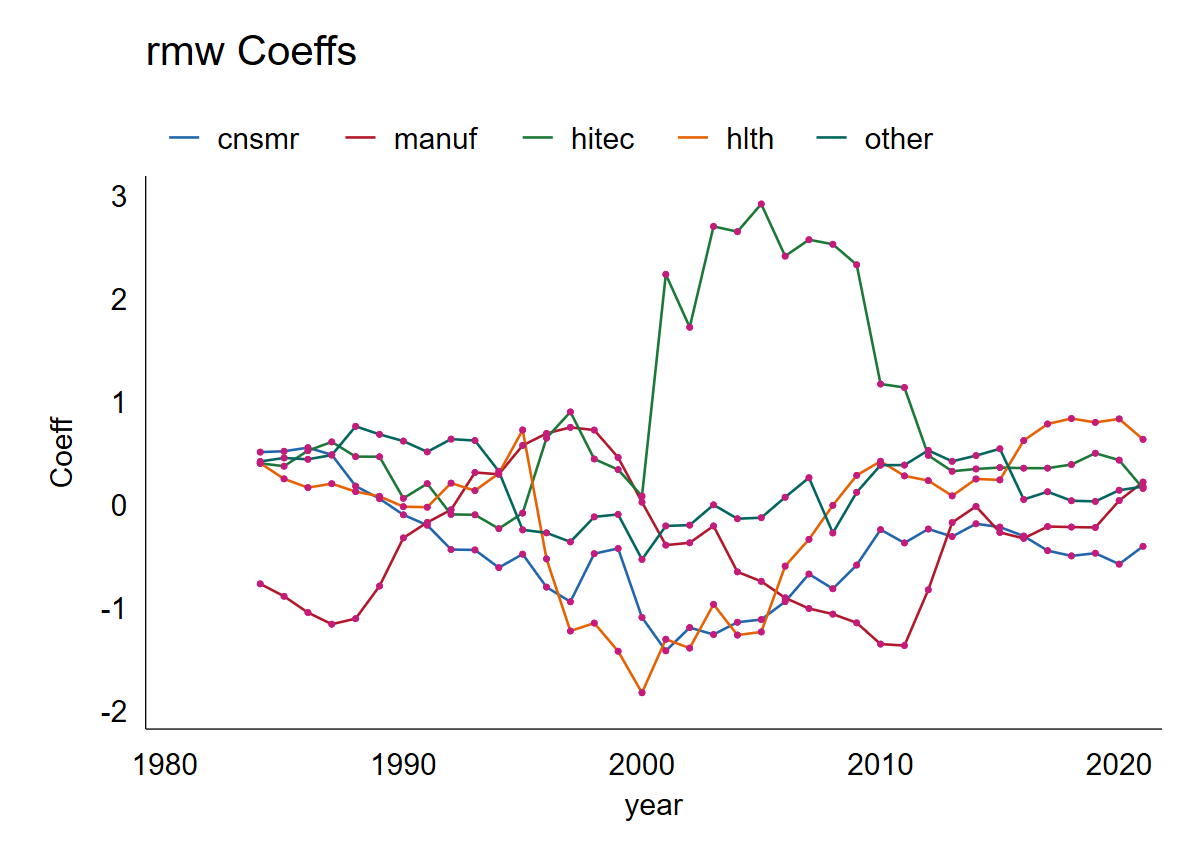}}
\end{centering}
\begin{centering}
\subfloat{\includegraphics[width=7cm]{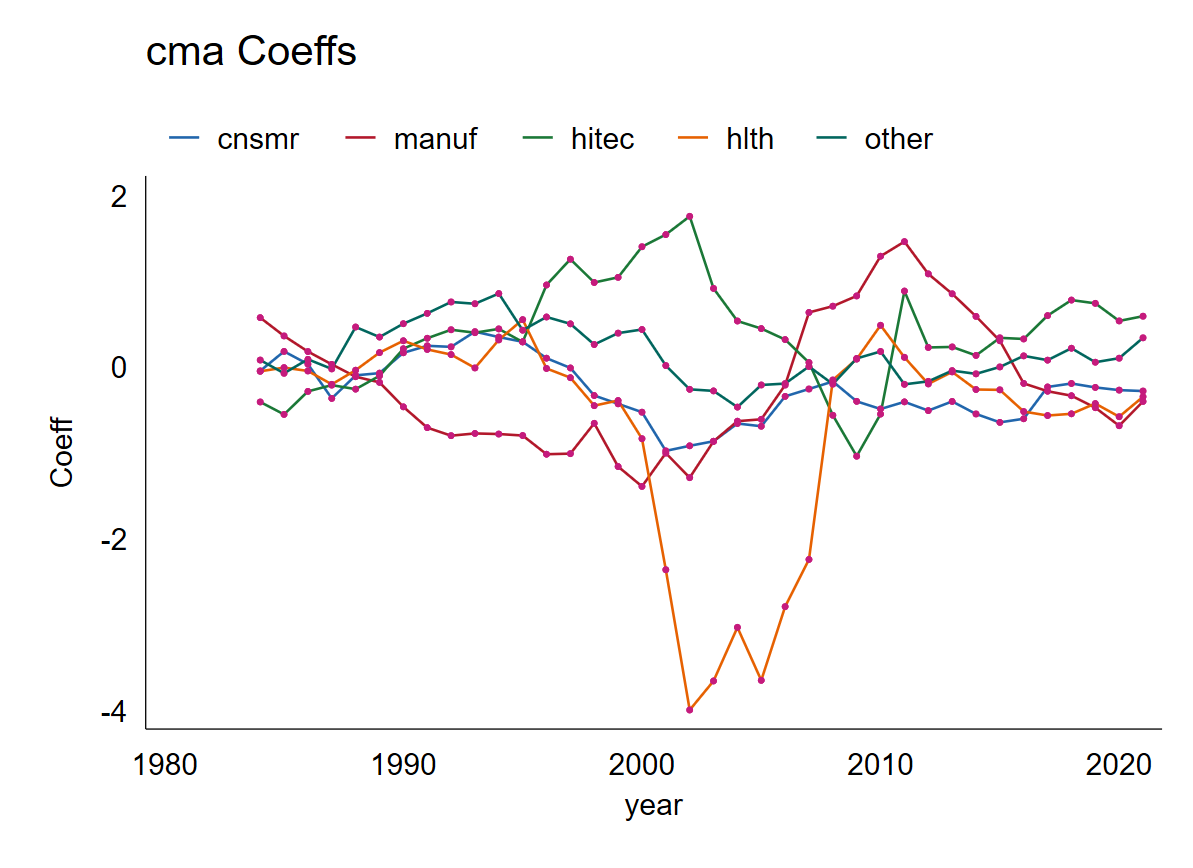}}
\subfloat{\includegraphics[width=7cm]{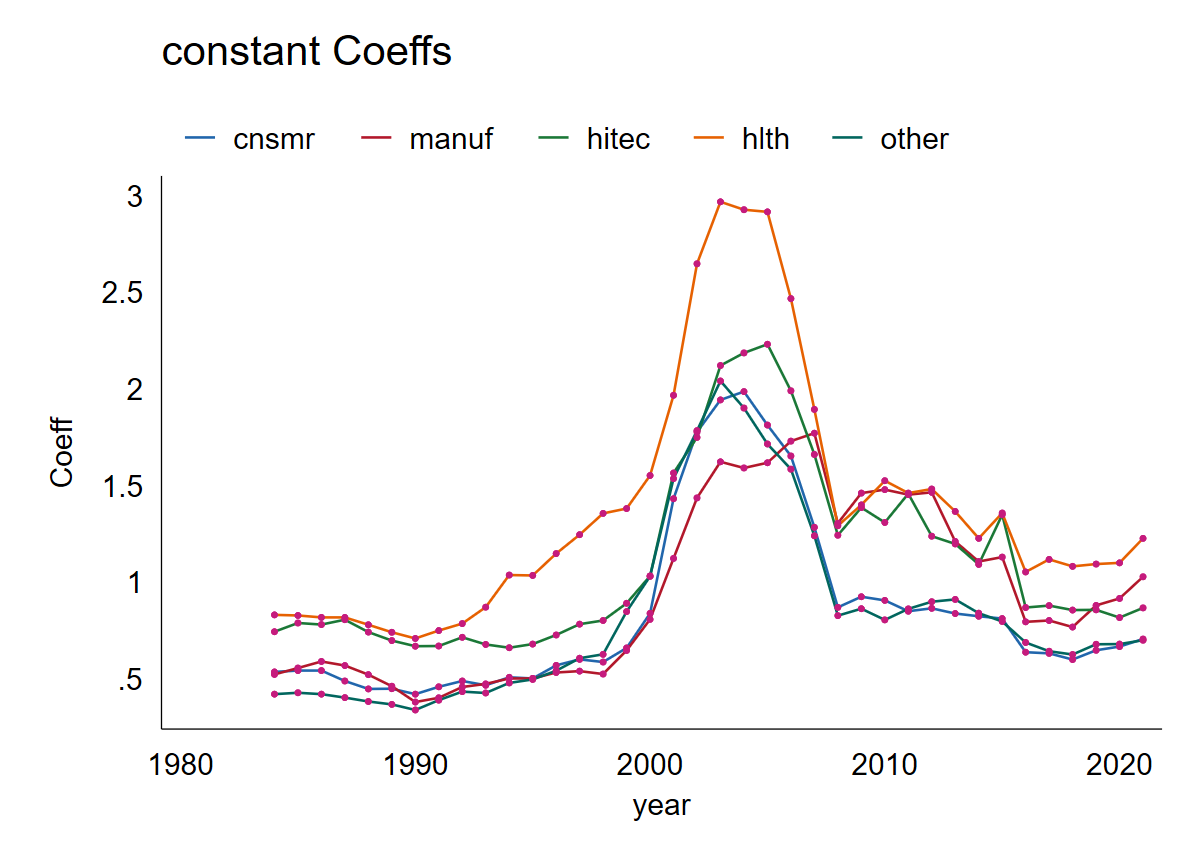}}
\end{centering}
\end{figure}

\begin{figure}[H]
\caption{Inequality Regression Coefficients\label{fig:Coefficient-by-Year-Inequality-1}}
\caption*{
This figure shows the time-varying WAQR coefficient estimates of the difference
between the average logged weekly wages of the top and bottom 10\%
to several standard individual characteristics, including family size,
an indicator variable of no children, age, and education. Each year,
one inequality regression is conducted with all the independent variables.
The blue lines represent coefficient estimates based on the inequality
regression and red lines represent coefficient estimates based on
a native method. The 95\% confidence intervals for the point estimates
based on the inequality regressions are plotted in dash lines. 
}

\begin{centering}
\subfloat{\includegraphics[width=7cm]{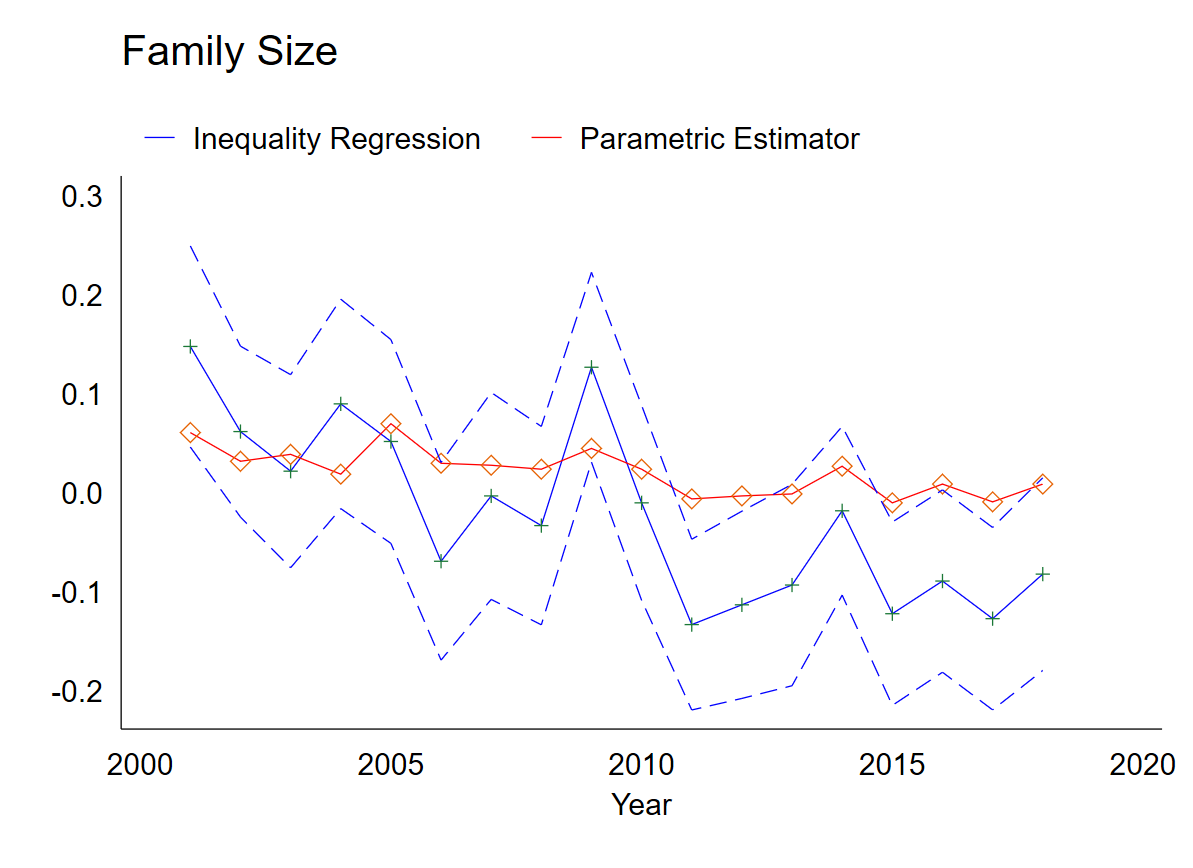}}\subfloat{\includegraphics[width=7cm]{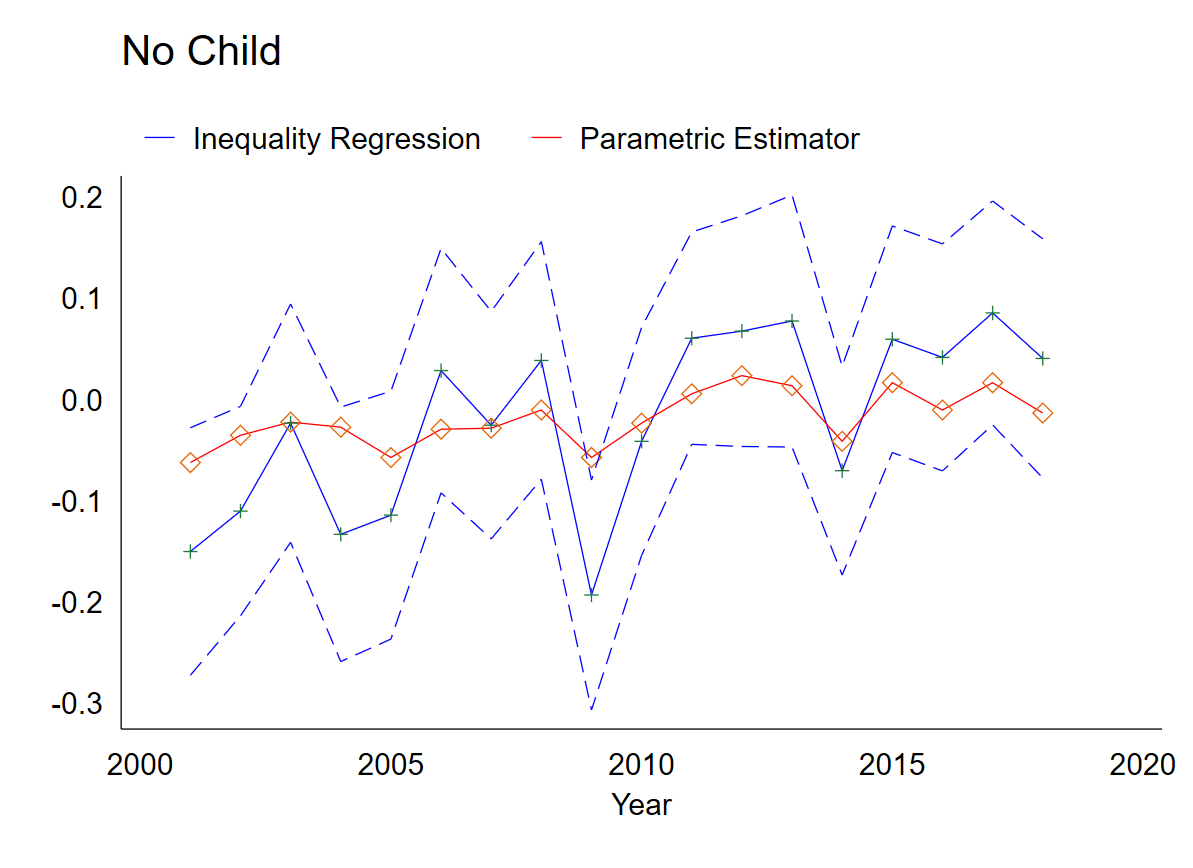}}
\end{centering}
\begin{centering}
\subfloat{\includegraphics[width=7cm]{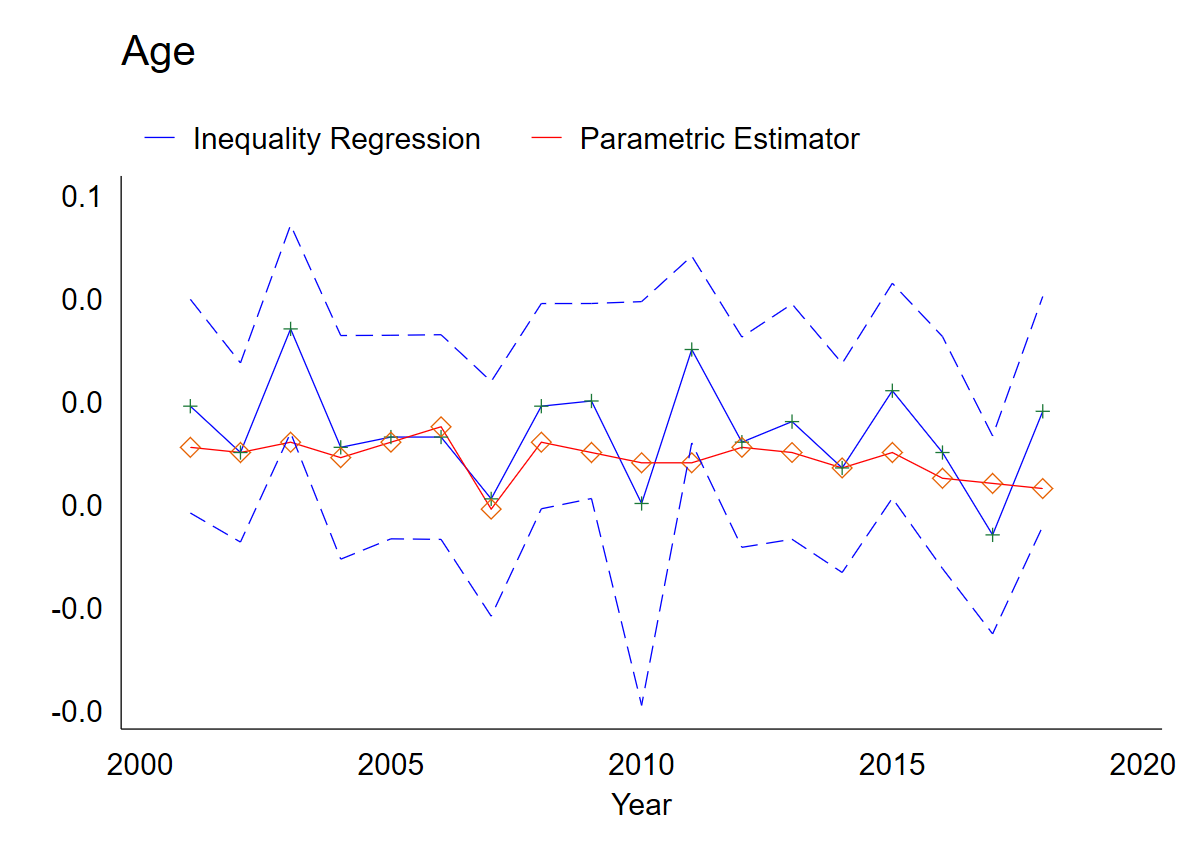}}\subfloat{\includegraphics[width=7cm]{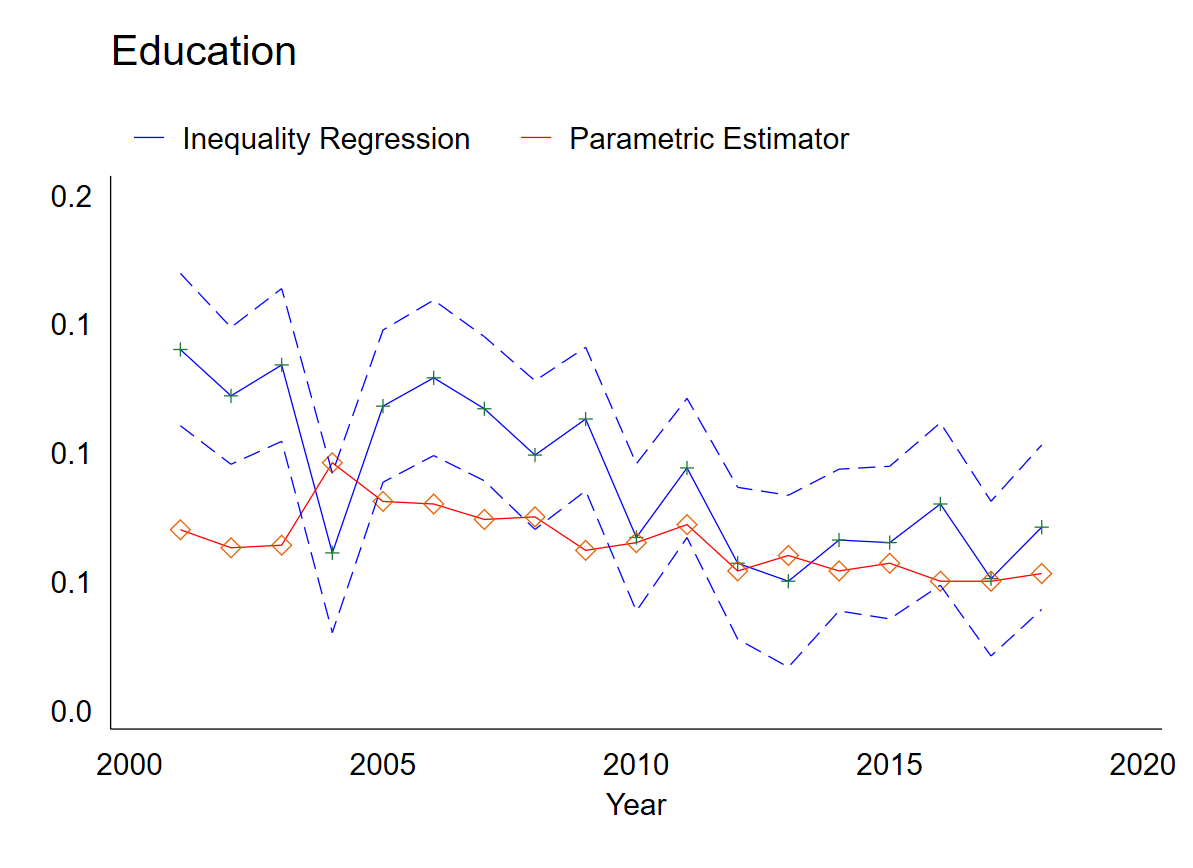}}
\end{centering}
\end{figure}

\begin{figure}[H]
\caption{Social Welfare Regression Coefficients\label{fig:Coefficient-by-Year-Lowincome}}

\caption*{
This figure shows the time-varying WAQR coefficient estimates for the social welfare (exponential) regression, with dependent variable being wage and the vector of independent variables including individual characteristics (family size,
an indicator variable of no children, age, and education). Each year,
one social welfare (exponential) regression is conducted with all the independent variables.
The blue lines represent coefficient estimates based on the social welfare (exponential)
regression and red lines represent coefficient estimates based on
the mean regression. The 95\% confidence intervals for the point estimates
based on the social welfare (exponential) regressions are plotted in dash lines. 
}

\begin{centering}
\subfloat{\includegraphics[width=7cm]{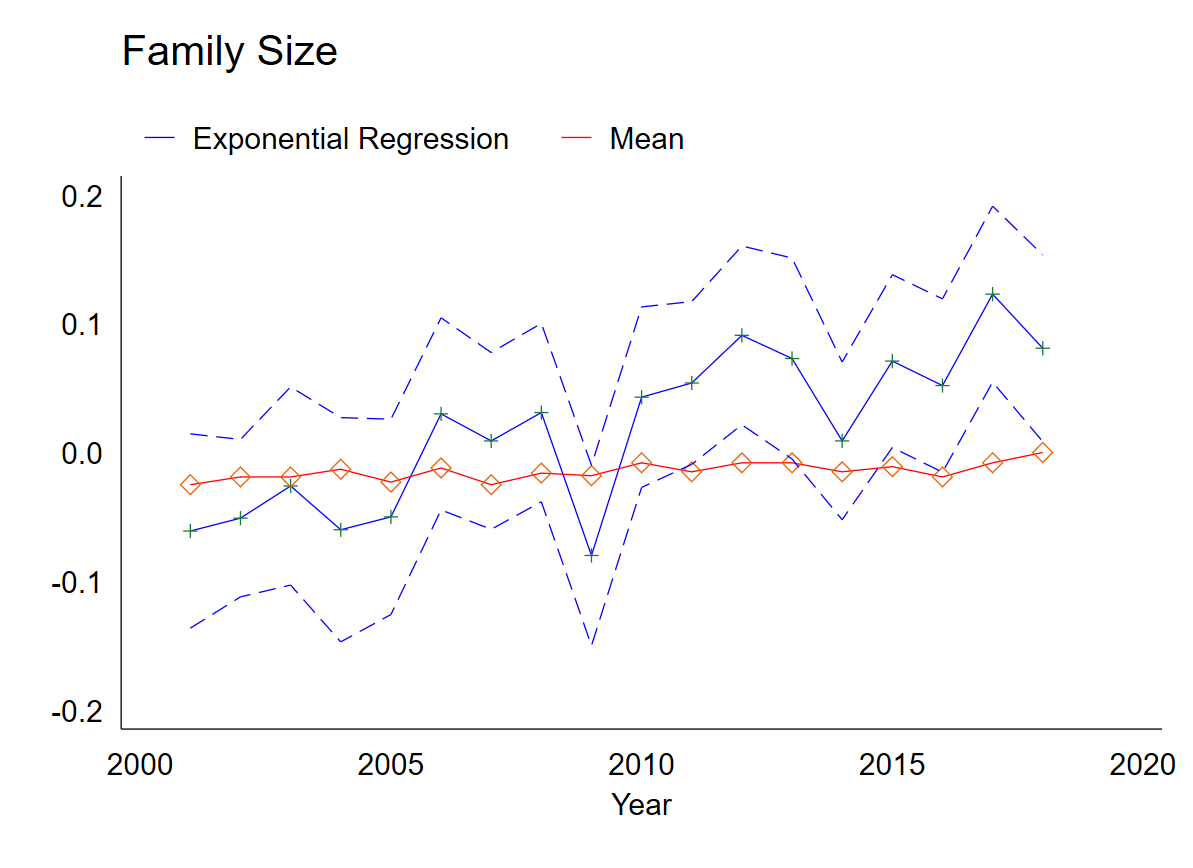}}\subfloat{\includegraphics[width=7cm]{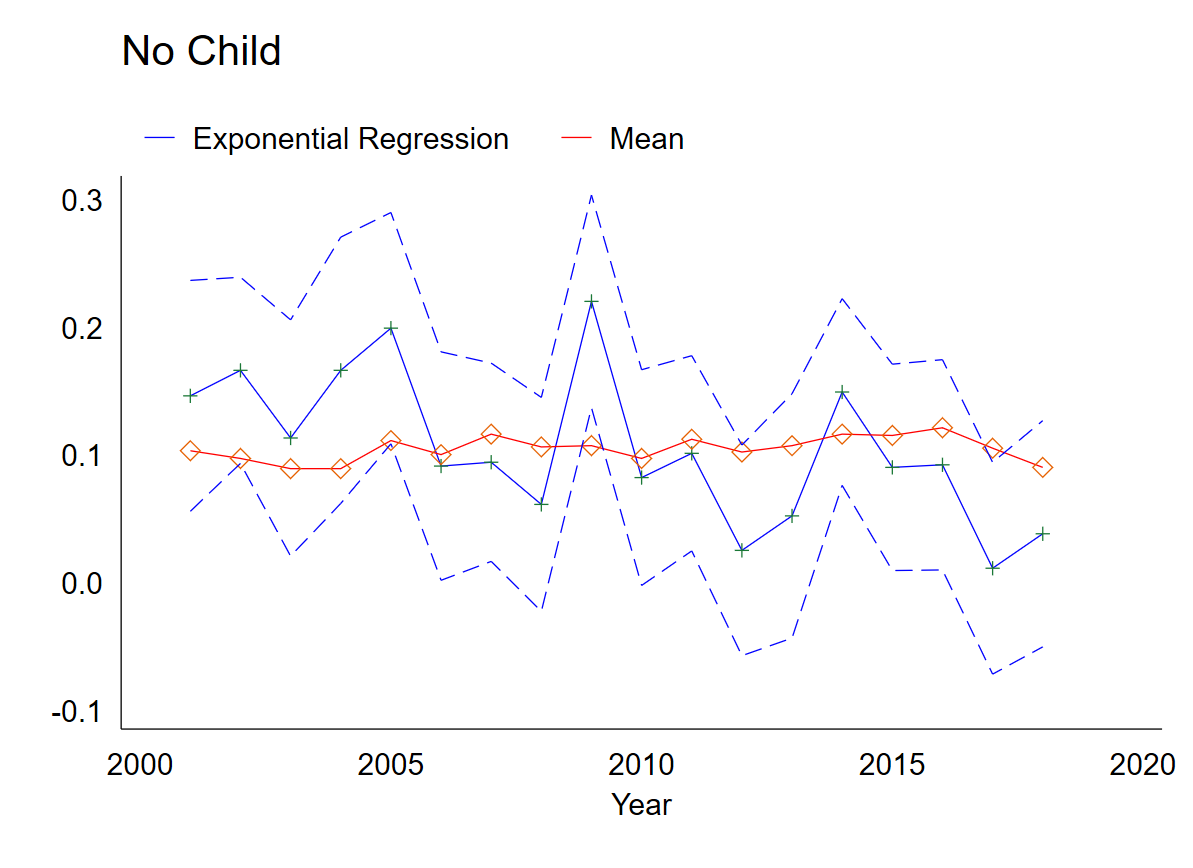}}
\end{centering}
\begin{centering}
\subfloat{\includegraphics[width=7cm]{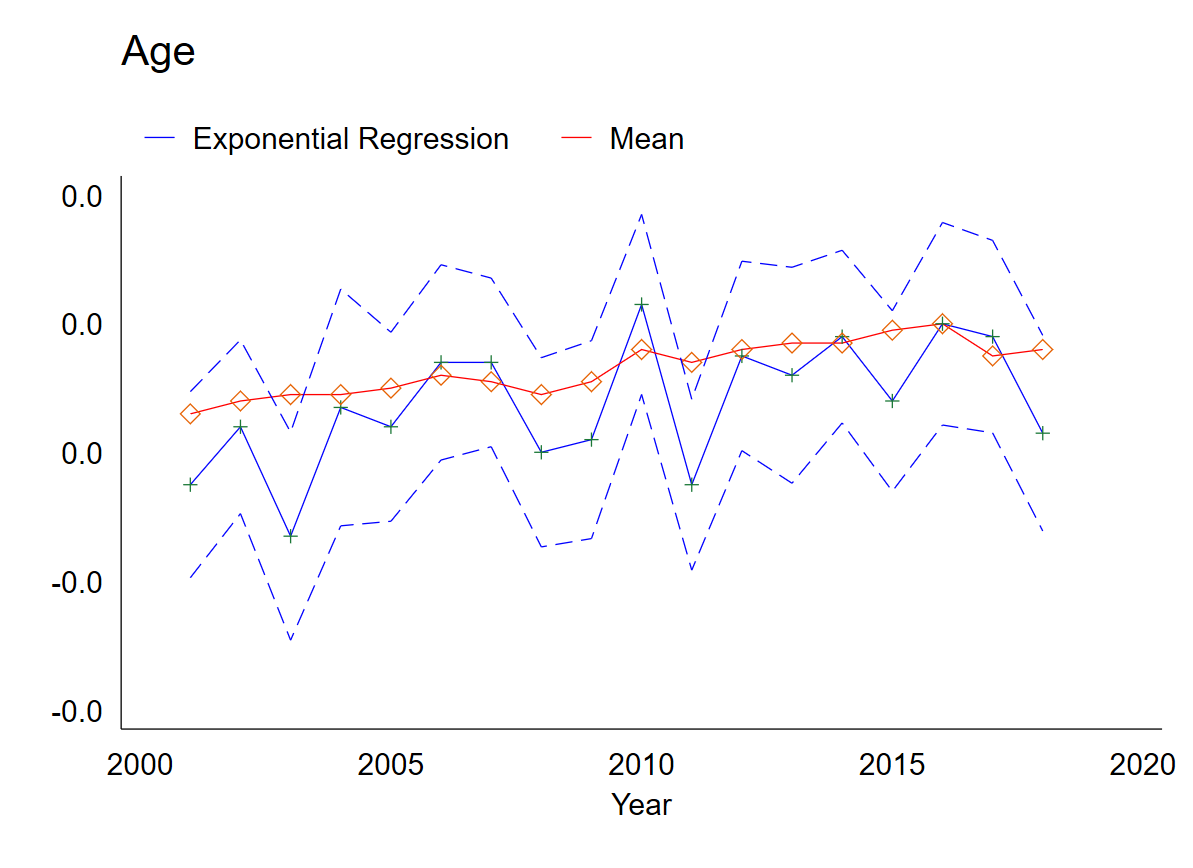}}\subfloat{\includegraphics[width=7cm]{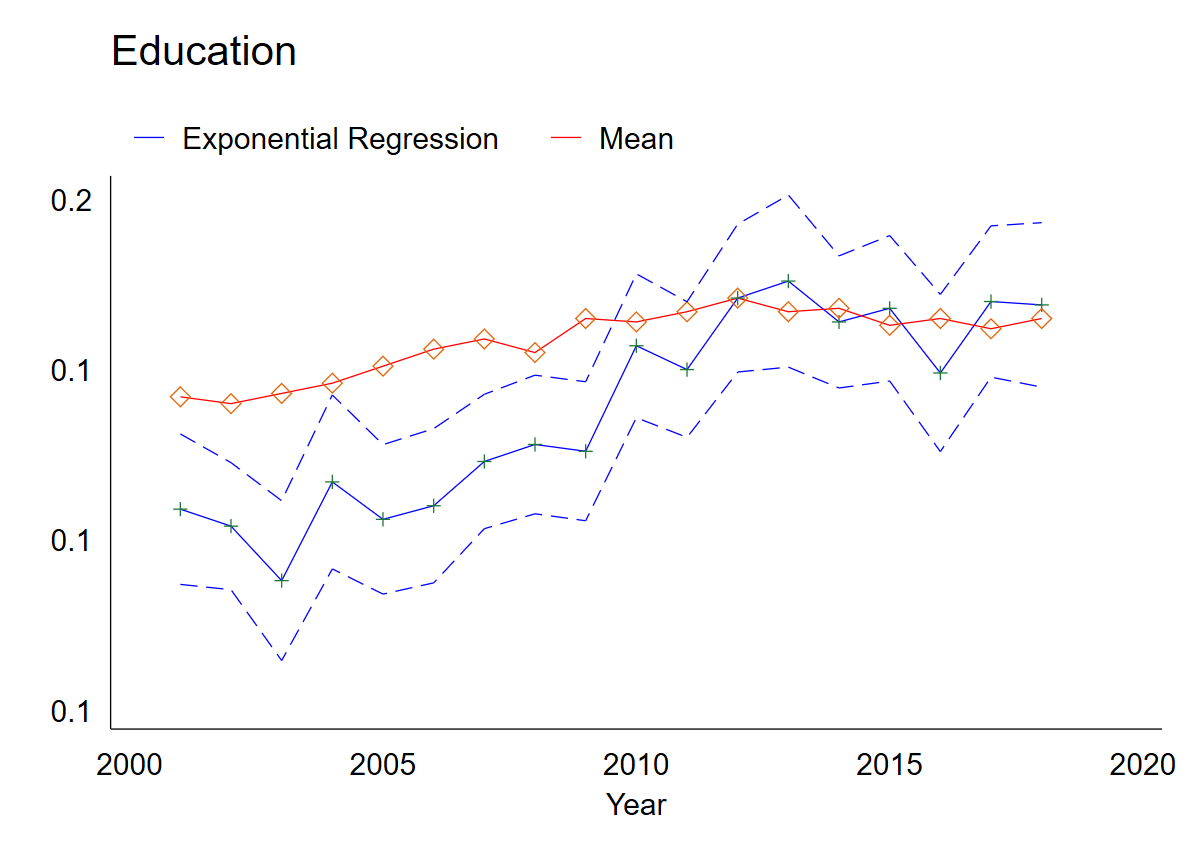}}
\end{centering}
\end{figure}

\clearpage

\appendix

\begin{center}
\textbf{{\Large For Online Publication}}
\end{center}

\section{Proof of Theorem \ref{thm: main result}}
The proof of the theorem is long. We therefore start with a sequence of useful lemmas. Throughout this section, we will assume, without loss of generality, that $D = (X',Y)'$ is independent of $\{D_t\}_{t\in\mathbb Z}$.

\begin{lemma}\label{lem: auxiliary 1}
Under Assumptions \ref{as: function psi}, \ref{as: function F}, and \ref{as: estimator Fhat}(i), we have
\begin{equation}\label{eq: psi-norm bound lem 1}
\int_{-\infty}^{+\infty}\left| \psi(\widehat F(s|x)) - \psi(F(s|x)) \right| ds \leq C\Delta(x)
\end{equation}
for all $x\in\mathcal X$ with probability approaching one, where $C>0$ is some constant.
\end{lemma}

\begin{proof}
By Assumption \ref{as: function psi}(i), the function $\psi$ can be decomposed as $\psi = \psi_1 - \psi_2$, where both $\psi_1\colon[0,1]\to\mathbb R$ and $\psi_2\colon[0,1]\to\mathbb R$ are bounded and increasing functions. Moreover, by Assumption \ref{as: function psi}(ii), we can choose $\psi_1$ and $\psi_2$ such that they are both continuously differentiable on $(0,u_0)$ and $(1-u_0,1)$ with bounded derivatives. Moreover, by suitable shifting these functions, we can assume, without loss of generality, that they are both non-negative. We will show how to prove that \eqref{eq: psi-norm bound lem 1} holds with $\psi$ replaced by $\psi_1$. We will then note that the same argument applies in the case of $\psi_2$ as well, and so \eqref{eq: psi-norm bound lem 1} will follow from the triangle inequality.

For all $x\in\mathcal X$, denote
$$
\Delta_1(x) = \sup_{s\in\mathbb R}\left| \widehat F(s|x) - F(s|x) \right|\text{ and }\Delta_2(x) = \int_{-\infty}^{+\infty} \left| \widehat F(s|x) - F(s|x)  \right| ds
$$
so that $\Delta(x) = \Delta_1(x) + \Delta_2(x)$. In addition, extend the function $\psi_1$ from $[0,1]$ to $\mathbb R$ by setting $\psi_1(u) = \psi_1(0)$ for all $u < 0$ and $\psi_1(u) = \psi_1(1)$ for all $u>1$. Defined this way, the function $\psi_1$ is bounded, increasing, and non-negative on $\mathbb R$.

Now, since $\psi_1$ is increasing, we have for all $x\in\mathcal X$ and $s\in\mathbb R$ that
$$
\psi_1(\widehat F(s|x)) - \psi_1(F(s|x)) \leq \psi_1(F(s|x)+\Delta_{1}(x)) - \psi_1(F(s|x))
$$
and
$$
\psi_1(F(s|x)) - \psi_1( \widehat  F(s|x)) \leq \psi_1(F(s|x)) - \psi_1(F(s|x)-\Delta_{1}(x)).
$$
Hence,
\begin{align*}
&\left| \psi_1(\widehat F(s|x)) - \psi_1(F(s|x)) \right| \\
&\qquad \leq \max\Big(\psi_1(F(s|x)+\Delta_{1}(x)) - \psi_1(F(s|x)),\psi_1(F(s|x)) - \psi_1(F(s|x)-\Delta_{1}(x))\Big)\\
&\qquad \leq \psi_1(F(s|x)+\Delta_{1}(x)) -  \psi_1(F(s|x)-\Delta_{1}(x)).
\end{align*}
Therefore,
\begin{align}
&\int_{s_1}^{s_2}|\psi_1(\widehat F(s|x)) - \psi_1(F(s|x))|ds \nonumber \\
&\qquad \leq\int_{s_1}^{s_2}\left( \psi_1(F(s|x) + \Delta_{1}(x)) - \psi_1(F(s|x)-\Delta_{1}(x)) \right) ds \nonumber\\
& \qquad \leq \frac{1}{c} \int_{F(s_1|x)}^{F(s_2|x)}(\psi_1(z+\Delta_{1}(x)) - \psi_1(z-\Delta_{1}(x)))dz\nonumber\\
& \qquad = \frac{1}{c}\int_{F(s_2|x)-\Delta_{1}(x)}^{F(s_2|x)+\Delta_{1}(x)}\psi_1(z)dz - \frac{1}{c}\int_{F(s_1|x)-\Delta_{1}(x)}^{F(s_1|x)+\Delta_{1}(x)}\psi_1(z)dz   \leq \frac{2\Delta_{1}(x)\psi_1(1)}{c}, \label{eq: sup-integral bound 1}
\end{align}
where the second inequality follows from Assumption \ref{as: function F}(ii) by carrying out the change of variables $s\mapsto F(s|x) = z$, and the third from the fact that $\psi_1$ is non-negative. 

Moreover, by Assumptions \ref{as: function F}(i) and \ref{as: estimator Fhat}(i), we have
$$
\sup_{s\leq s_1} \widehat F(s|x) \leq \sup_{s\leq s_1} F(s|x) + \Delta_{1}(x) \leq F(s_1|x) + \Delta_{1}(x) < u_0/2 + u_0/2 \leq u_0
$$
and
$$
\inf_{s\geq s_2} \widehat F(s|x) \geq \inf_{s\geq s_2} F(s|x) - \Delta_{1}(x) \geq F(s_2|x) - \Delta_{1}(x) > 1 - u_0/2 - u_0/2 = 1 - u_0
$$
with probability approaching one uniformly over $x\in\mathcal X$. Therefore, by Assumptions \ref{as: function psi}(iii) and \ref{as: function F}(i), for some constant $C_{\psi}>0$,
\begin{equation}\label{eq: sup-integral bound 2}
\int_{-\infty}^{s_1}\left| \psi_1(\widehat F(s|x)) - \psi_1(F(s|x)) \right| ds \leq C_{\psi}\int_{-\infty}^{s_1}\left|\widehat F(s|x) - F(s|x) \right| ds \leq C_{\psi}\Delta_{2}(x)
\end{equation}
and
\begin{equation}\label{eq: sup-integral bound 3}
\int_{s_2}^{\infty}\left| \psi_1(\widehat F(s|x)) - \psi_1(F(s|x)) \right| ds \leq C_{\psi}\int_{s_2}^{\infty}\left|\widehat F(s|x) - F(s|x) \right| ds \leq C_{\psi}\Delta_{2}(x)
\end{equation}
with probability approaching one uniformly over $x\in\mathcal X$. Combining \eqref{eq: sup-integral bound 1}, \eqref{eq: sup-integral bound 2}, and \eqref{eq: sup-integral bound 3} gives \eqref{eq: psi-norm bound lem 1} with $\psi$ replaced by $\psi_1$. In addition, we can prove by the same argument that \eqref{eq: psi-norm bound lem 1} holds with $\psi$ replaced by $\psi_2$ as well. The asserted claim now follows from the triangle inequality.
\end{proof}

\begin{lemma}\label{lem: conditional variance}
Consider a sequence of functions $\{f_T\}_{T\geq 2}$ such that for all $T\geq 2$, the function $f_T$ is mapping $\mathcal D\times\mathcal D^{T_1}$ into $\mathbb R$, where $\mathcal D$ is the support of $D$. Then
$$
\textrm{{\em Var}}\left( \sum_{t=T_1+1}^{T_1+T_2} f_T(D_t,D_1^{T_1}) \mid D_1^{T_1} \right) = o_P(T)
$$
as long as 
\begin{equation}\label{eq: fourth moment convergence}
\sum_{t=T_1+1}^{T_1+T_2}\Ep\left[|f_T(D_t,D_1^{T_1})|^4 \mid D_1^{T_1}\right] = o_P(1)
\end{equation}
and Assumption \ref{as: beta mixing} is satisfied.
\end{lemma}
\begin{proof}
For brevity of notations, for all $t = T_1+1,\dots,T_1 + T_2$, we will write $f_t$ instead of $f_T(D_t,D_1^{T_1})$ throughout the proof. Then it follows from \eqref{eq: fourth moment convergence} that there exists $\gamma_T\to0$ as $T\to\infty$ such that 
\begin{equation}\label{eq: fourth moment probability bound}
\Pr\left(\sum_{t=T_1+1}^{T_1+T_2}\Ep\left[|f_t|^4 \mid D_1^{T_1}\right] \leq \gamma_T^2\right) \geq 1 - \gamma_T.
\end{equation}
Next, for all $t = T_1+1,\dots,T_1+T_2$, denote $f_{t,+} = f_t \mathbb I\{f_t \geq 0\}$ and $f_{t,-} = - f_t \mathbb I\{f_t < 0\}$, so that $f_t = f_{t,+} - f_{t,-}$. Then
$$
f_t = \int_{0}^{+\infty} (\mathbb I\{f_{t,+} > s\} - \mathbb I\{f_{t,-} > s\})ds,
$$
and so
\begin{align}
& \textrm{Var}\left( \sum_{t=T_1+1}^{T_1+T_2} f_t \mid D_1^{T_1} \right) 
 = \sum_{t_1,t_2=T_1+1}^{T_1+T_2}\textrm{Cov}(f_{t_1},f_{t_2}\mid D_1^{T_1}) \nonumber\\
& \qquad  = \int_0^{+\infty}\int_0^{+\infty} \sum_{t_1,t_2=T_1+1}^{T_1+T_2}\textrm{Cov}\Big( \mathbb I\{f_{t_1,+} > s_1\} - \mathbb I\{f_{t_1,-} > s_1\}, \nonumber \\
&\qquad\qquad\qquad\qquad\qquad\qquad\qquad \mathbb I\{f_{t_2,+} > s_2\} - \mathbb I\{f_{t_2,-} > s_2\} \mid D_1^{T_1}  \Big)ds_1ds_2. \label{eq: conditional variance formula}
\end{align}
Denoting the integrand here by $R(s_1,s_2)$, we now derive three different bounds on it. 

First, observe that for any $t_1<t_2$ and any random variables $Z_1$ and $Z_2$ such that $Z_1$ depends only on $D_{t_1}$ and $D_1^{T_1}$ and $Z_2$ depends only on $D_{t_2}$ and $D_1^{T_1}$, we have
\begin{align*}
\textrm{Cov}(Z_1,Z_2\mid D_1^{T_1})
& = \Ep[Z_1Z_2\mid D_1^{T_1}] - \Ep[Z_1\mid D_1^{T_1}]\Ep[Z_2\mid D_1^{T_1}] \\
& = \Ep\left[Z_1(\Ep[Z_2\mid D_{t_1}, D_1^{T_1}] - \Ep[Z_2\mid D_1^{T_1}]) \mid D_1^{T_1} \right],
\end{align*}
and so if $|Z_1|\leq 1$ a.s., then
\begin{align*}
|\textrm{Cov}(Z_1,Z_2\mid D_1^{T_1})|
& \leq \Ep[|\Ep[Z_2\mid D_{t_1}, D_1^{T_1}] - \Ep[Z_2\mid D_1^{T_1}]| \mid D_1^{T_1} ].
\end{align*}
Substituting here $Z_1 = \mathbb I\{f_{t_1,+} > s_1\} - \mathbb I\{f_{t_1,-} > s_1\}$ and $Z_2 = \mathbb I\{f_{t_2,+} > s_2\} - \mathbb I\{f_{t_2,-} > s_2\}$, we obtain
\begin{align*}
&\left| \textrm{Cov}\Big( \mathbb I\{f_{t_1,+} > s_1\} - \mathbb I\{f_{t_1,-} > s_1\}, \mathbb I\{f_{t_2,+} > s_2\} - \mathbb I\{f_{t_2,-} > s_2\} \mid D_1^{T_1}  \Big) \right|\\
&\quad \leq  \Ep\left[ \left| \Pr( f_T(D_{t_2},D_1^{T_1}) > s_2 \mid D_{t_1},D_1^{T_1}) - \Pr( f_T(D_{t_2},D_1^{T_1}) > s_2 \mid D_1^{T_1}) \right| \mid D_1^{T_1} \right] \\
&\quad\quad +  \Ep\left[ \left| \Pr( f_T(D_{t_2},D_1^{T_1}) < - s_2 \mid D_{t_1},D_1^{T_1}) - \Pr( f_T(D_{t_2},D_1^{T_1}) < - s_2 \mid D_1^{T_1}) \right| \mid D_1^{T_1} \right] \\
& \quad \leq 2\Ep\left[\sup_{B}\left| \Pr(D_{t_2} \in B \mid D_{t_1},D_1^{T_1}) - \Pr(D_{t_2} \in B \mid D_1^{T_1}) \right| \mid D_1^{T_1}\right],
\end{align*}
where the penultimate inequality follows from the definition of the $\beta$-mixing coefficients. Therefore,
\begin{align*}
\Ep\left[\sup_{s_1,s_2\in(0,\infty)}|R(s_1,s_2)|\right] 
&\leq \sum_{t = T_1+1}^{T_1 + T_2}1 + 2\sum_{t_1 = T_1 + 1}^{T_1+T_2 - 1}\sum_{t_2 = t_1+1}^{T_1 + T_2}2(\beta_{t_2-t_1}+\beta_{t_2 - T_1}) \\
& \leq (T_2 - T_1)\left(1 + 8\sum_{t=1}^{\infty} \beta_t\right)
\end{align*}
by the definition of the $\beta$-mixing coefficients. Hence, by Markov's inequality,
\begin{equation}\label{eq: mixing bound on R}
|R(s_1,s_2)| \leq \frac{T_2 - T_1}{\gamma_T}\left(1 + 8\sum_{t=1}^{\infty} \beta_t\right)
\end{equation}
with probability at least $1 - \gamma_T$ uniformly over $s_1,s_2\in(0,\infty)$, which is our first bound on $R(s_1,s_2)$.

Next, observe that for any random variables $Z_1$ and $Z_2$ such that $|Z_1|\leq 1$ a.s., we have
\begin{align*}
|\textrm{Cov}(Z_1,Z_2\mid D_1^{T_1})|
& = | \Ep[ (Z_1 - \Ep[Z_1\mid D_1^{T_1}])(Z_2 - \Ep[Z_2\mid D_1^{T_1}]) \mid D_1^{T_1}] | \\
& = | \Ep[ Z_1(Z_2 - \Ep[Z_2\mid D_1^{T_1}]) \mid D_1^{T_1}] |  \leq \Ep[ |Z_2 - \Ep[Z_2\mid D_1^{T_1}]| \mid D_1^{T_1}] \\
& \leq \Ep[|Z_2|\mid D_1^{T_1}] + |\Ep[Z_2\mid D_1^{T_1}]| \leq 2\Ep[|Z_2|\mid D_1^{T_1}].
\end{align*}
Substituting here $Z_1 = \mathbb I\{f_{t_1,+} > s_1\} - \mathbb I\{f_{t_1,-} > s_1\}$ and $Z_2 = \mathbb I\{f_{t_2,+} > s_2\} - \mathbb I\{f_{t_2,-} > s_2\}$ again, we obtain
\begin{align*}
&\left| \textrm{Cov}\Big( \mathbb I\{f_{t_1,+} > s_1\} - \mathbb I\{f_{t_1,-} > s_1\}, \mathbb I\{f_{t_2,+} > s_2\} - \mathbb I\{f_{t_2,-} > s_2\} \mid D_1^{T_1}  \Big) \right|\\
& \qquad \leq 2\Big(\Pr( f_{t_2} > s_2 \mid D_1^{T_1}) + \Pr( f_{t_2} < - s_2 \mid D_1^{T_1})\Big) \leq 2\Pr( |f_{t_2}| > s_2 \mid D_1^{T_1}).
\end{align*}
Therefore,
$$
|R(s_1,s_2)| \leq 2\sum_{t_1,t_2=T_1+1}^{T_1 + T_2} \Pr( |f_{t_2}| > s_2 \mid D_1^{T_1}) \leq 2(T_2-T_1)\sum_{t=T_1+1}^{T_1+T_2} \Pr( |f_{t}| > s_2 \mid D_1^{T_1}),
$$
and so, by Markov's inequality and \eqref{eq: fourth moment probability bound},
$$
|R(s_1,s_2)| \leq \frac{2(T_2-T_1)}{s_2^4}\sum_{t=T_1+1}^{T_1+T_2}\Ep[|f_t|^4\mid D_1^{T_1}] \leq \frac{2\gamma_T^2(T_2-T_1)}{s_2^4}
$$
with probability at least $1 - \gamma_T$ uniformly over $s_1,s_2\in(0,\infty)$, which is our second bound on $R(s_1,s_2)$. In addition, by the same argument, with interchanged $Z_1$ and $Z_2$,
$$
|R(s_1,s_2)| \leq 2(T_2-T_1)\sum_{t=T_1+1}^{T_1+T_2} \Pr( |f_{t}| > s_1 \mid D_1^{T_1}) \leq \frac{2\gamma_T^2(T_2-T_1)}{s_1^4},
$$
with probability at least $1 - \gamma_T$ uniformly over $s_1,s_2\in(0,\infty)$, which is our third bound on $R(s_1,s_2)$.

Now, denoting the right-hand side of \eqref{eq: mixing bound on R} by $\bar R$ and combining all three bounds, we have
$$
|R(s_1,s_2)|\leq \int_0^{\bar R}\mathbb I\left\{ u\leq \frac{2\gamma_T^2(T_2-T_1)}{s_1^4} \right\}\mathbb I\left\{ u\leq \frac{2\gamma_T^2(T_2-T_1)}{s_2^4} \right\}du
$$
with probability at least $1 - 3\gamma_T$. Substituting this bound into \eqref{eq: conditional variance formula}, we obtain
\begin{align*}
& \textrm{Var}\left( \sum_{t=T_1+1}^{T_1+T_2} f_t \mid D_1^{T_1} \right) 
 \leq \int_{0}^{+\infty}\int_{0}^{+\infty} |R(s_1,s_2)|ds_1ds_2\\
& \qquad \leq \int_{0}^{+\infty}\int_{0}^{+\infty} \int_0^{\bar R}\mathbb I\left\{ u\leq \frac{2\gamma_T^2(T_2-T_1)}{s_1^4} \right\}\mathbb I\left\{ u\leq \frac{2\gamma_T^2(T_2-T_1)}{s_2^4} \right\}duds_1ds_2 \\
&\qquad = \int_0^{\bar R} \left(\frac{2\gamma_T^2(T_2-T_1)}{u}\right)^{1/4}\left(\frac{2\gamma_T^2(T_2-T_1)}{u}\right)^{1/4}du \\
&\qquad= 2\gamma_T\sqrt{2(T_2-T_1)}\sqrt{\bar R} \leq 2\sqrt{2\gamma_T\left(1+\sum_{t=1}^{\infty}\beta_t\right)}(T_2-T_1)
\end{align*}
with probability at least $1-3\gamma_T$. Since $\gamma_T\to0$ and $\sum_{t=1}^{\infty}\beta_t < \infty$ by Assumption \ref{as: beta mixing}, the asserted claim follows.
\end{proof}

\begin{lemma}\label{lem: auxiliary lemma 3}
Consider a sequence of functions $\{f_T\}_{T\geq 2}$ such that for all $T\geq 2$, the function $f_T$ is mapping $\mathcal D\times\mathcal D^{T_1}$ into $\mathbb R$, where $\mathcal D$ is the support of $D$. Let $\{A_T\}_{T\geq 2}$ be a sequence of positive numbers. Also, let $D = (X',Y)'$ be independent of $\{D_t\}_{t\in\mathbb Z}$. Finally, let $\delta>0$ be some number. Then
$$
\sum_{t=T_1+1}^{T_1+T_2}\left(\Ep\left[ f_T(D_t,D_1^{T_1}) \mid D_1^{T_1}\right] - \Ep\left[ f_T(D,D_1^{T_1}) \mid D_1^{T_1}\right]\right) = o_P (A_T)
$$
as long as
\begin{equation}\label{eq: 1 plus delta moment}
\sum_{t=T_1+1}^{T_1+T_2}\left(\Ep\left[ | f_T(D_t,D_1^{T_1}) |^{1+\delta} \mid D_1^{T_1} \right] + \Ep\left[ | f_T(D,D_1^{T_1}) |^{1+\delta} \mid D_1^{T_1} \right]\right) = o_P(A_T^{1+\delta})
\end{equation}
and Assumption \ref{as: beta mixing} is satisfied.
\end{lemma}
\begin{proof}
For brevity of notations, for all $t = T_1+1,\dots,T_2$, we will write $f_t$ instead of $f_T(D_t,D_1^{T_1})$ throughout the proof. In addition, we will write $\tilde f$ instead of $f_T(D,D_1^{T_1})$. Then it follows from \eqref{eq: 1 plus delta moment} that there exists $\gamma_T\to0$ as $T\to\infty$ such that 
\begin{equation}\label{eq: 1 plus delta probability bound}
\Pr\left(\sum_{t=T_1+1}^{T_1+T_2}\left(\Ep\left[ | f_t|^{1+\delta} \mid D_1^{T_1} \right] + \Ep\left[ | \tilde f |^{1+\delta} \mid D_1^{T_1} \right]\right) \leq (\gamma_T A_T)^{1+\delta}\right) \geq 1 - \gamma_T.
\end{equation}
Next, for all $t = T_1+1,\dots,T_2$, denote $f_{t,+} = f_t \mathbb I\{f_t \geq 0\}$ and $f_{t,-} = -f_t \mathbb I\{f_t < 0\}$, so that
\begin{equation}\label{eq: f decomposition positive-negative 1}
f_t = f_{t,+} - f_{t,-} = \int_{0}^{+\infty} (\mathbb I\{f_{t,+} > s\} - \mathbb I\{f_{t,-} > s\})ds.
\end{equation}
Similarly, denote $\tilde f_{+} = \tilde f \mathbb I\{\tilde f \geq 0\}$ and $\tilde f_{-} = -\tilde f \mathbb I\{\tilde f < 0\}$, so that
\begin{equation}\label{eq: f decomposition positive-negative 2}
\tilde f = \tilde f_{+} - \tilde f_{-} = \int_{0}^{+\infty} (\mathbb I\{\tilde f_{+} > s\} - \mathbb I\{\tilde f_{-} > s\})ds.
\end{equation}
Further, for all $s>0$, denote
$$
R_1(s) = \sum_{t=T_1+1}^{T_1+T_2}\left| \Pr(f_{t,+} > s\mid D_1^{T_1}) - \Pr(\tilde f_{+} > s\mid D_1^{T_1})  \right|
$$
and
$$
R_2(s) = \sum_{t=T_1+1}^{T_1+T_2}\left| \Pr(f_{t,-} > s\mid D_1^{T_1}) - \Pr(\tilde f_{-} > s\mid D_1^{T_1})  \right|.
$$
Then it follows from \eqref{eq: f decomposition positive-negative 1}, \eqref{eq: f decomposition positive-negative 2}, and the triangle inequality that
\begin{equation}\label{eq: sum r decomposition}
\left| \sum_{t=T_1+1}^{T_1+T_2} (\Ep[f_t\mid D_1^{T_1}] - \Ep[\tilde f\mid D_1^{T_1}] )\right| \leq \int_0^{+\infty} R_1(s)ds + \int_0^{+\infty} R_2(s)ds.
\end{equation}
We will bound $\int_0^{+\infty} R_1(s)ds$ and note that $\int_0^{+\infty} R_2(s)ds$ can be bounded by the same argument.

Observe that
\begin{align*}
R_1(s) 
& \leq \sum_{t=T_1+1}^{T_1+T_2} \sup_B \left|\Pr(D_t \in B\mid D_1^{T_1}) - \Pr(D\in B\mid D_1^{T_1})\right| \\
& = \sum_{t=T_1+1}^{T_1+T_2} \sup_B \left|\Pr(D_t \in B\mid D_1^{T_1}) - \Pr(D\in B)\right|\\
& = \sum_{t=T_1+1}^{T_1+T_2} \sup_B \left|\Pr(D_t \in B\mid D_1^{T_1}) - \Pr(D_t\in B)\right|,
\end{align*}
and so 
$$
\Ep\left[\sup_{s\in(0,\infty)}R_1(s)\right] \leq \sum_{t=1}^{\infty} \beta_t.
$$
Therefore, $R_1(s) \leq \sum_{t=1}^{\infty}\beta_t / \gamma_T$ with probability at least $1-\gamma_T$ uniformly over $s\in (0,\infty)$ by Markov's inequality. In addition,
\begin{align*}
R_1(s) 
& \leq \sum_{t=T_1+1}^{T_1+T_2} \left( \Pr( f_{t,+} > s\mid D_1^{T_1} ) + \Pr(\tilde f_{+} > s\mid D_1^{T_1} )\right) \\
& \leq \frac{1}{s^{1+\delta}}\sum_{t=T_1+1}^{T_1+T_2}\left( \Ep[| f_{t,+}|^{1+\delta}\mid D_1^{T_1}] + \Ep[|\tilde f_{+}|^{1+\delta}\mid D_1^{T_1}]\right) \leq \frac{(\gamma_T A_T)^{1+\delta}}{s^{1+\delta}}
\end{align*}
with probability at least $1 - \gamma_T$ uniformly over $s\in(0,\infty)$ by Markov's inequality and \eqref{eq: 1 plus delta probability bound}. Hence, for any $s_0>0$, we have
$$
\int_0^{\infty} R_1(s)ds = \int_0^{s_0} R_1(s)ds + \int_{s_0}^{\infty} R_1(s)ds \leq \frac{s_0}{\gamma_T}\sum_{t=1}^{\infty}\beta_t + \frac{(\gamma_T A_T)^{1+\delta}}{\delta s_0^{\delta}}
$$
with probability at least $1 - 2\gamma_T$. Therefore, setting $s_0 = A_T(\gamma_T^{2+\delta}/\delta\sum_{t=1}^{\infty}\beta_t)^{1/(1+\delta)}$, it follows that
$$
\int_0^{\infty} R_1(s)ds \leq 2A_t\left(\sum_{t=1}^{\infty}\beta_t\right)^{\delta/(1+\delta)}\left(\frac{\gamma_T}{\delta}\right)^{1/(1+\delta)}
$$
with probability at least $1 - 2\gamma_T$. Hence, given that $\sum_{t=1}^{\infty}\beta_t<\infty$ by Assumption \ref{as: beta mixing}, it follows that $\int_0^{\infty} R_1(s)ds = o_P(A_T)$ and, by the same argument, $\int_0^\infty R_2(s)ds = o_P(A_T)$. Substituting these bounds into \eqref{eq: sum r decomposition}, we obtain the asserted claim.
\end{proof}

\begin{lemma}\label{lem: fourth lemma}
Under Assumptions \ref{as: beta mixing}, \ref{as: standard convergence}, and \ref{as: estimator Fhat},
$$
\Ep\left[ \|X\|^2\Delta(X)^2\mid D_1^{T_1} \right] = o_P(1).
$$
\end{lemma}
\begin{proof}
By Jensen's inequality,
$$
\Ep\left[ \|X\|^2\Delta(X)^2\mid D_1^{T_1} \right] \leq \sqrt{\Ep[\|X\|^4\mid D_1^{T_1}]}\sqrt{\Ep[\Delta(X)^4\mid D_1^{T_1}]}.
$$
Therefore, given that $\Ep[\|X\|^4\mid D_1^{T_1}] = O_P(1)$ by Assumption \ref{as: standard convergence}(i) and Markov's inequality, it suffices to prove that $\Ep[\Delta(X)^4\mid D_1^{T_1}] = o_P(1)$. To do so, observe that by Assumption \ref{as: estimator Fhat}(ii), there exists $\gamma_T\to0$ as $T\to\infty$ such that
$$
\Pr\left( \sum_{t=T_1+1}^{T_1+T_2} \Ep[\Delta(X_t)^4\mid D_1^{T_1}] \leq \gamma_T \right) \geq 1 - \gamma_T.
$$
Hence,
\begin{align*}
\sum_{t=T_1+1}^{T_1+T_2} \Pr(\Delta(X_t) > \gamma_T^{1/8}\mid D_1^{T_1})
& =\sum_{t=T_1+1}^{T_1+T_2} \Pr(\Delta(X_t)^4 > \sqrt{\gamma_T} \mid D_1^{T_1}) \\
& \leq \frac{1}{\sqrt{\gamma_T}}\sum_{t=T_1+1}^{T_1+T_2}\Ep[\Delta(X_t)^4\mid D_1^{T_1}] \leq \sqrt{\gamma_T}
\end{align*}
with probability at least $1 - \gamma_T$. Also,
$$
\sum_{t=T_1+1}^{T_1+T_2} \Ep\left[ |\Pr(\Delta(X_t) > \gamma_T^{1/8} \mid D_1^{T_1}) - \Pr(\Delta(X) > \gamma_T^{1/8} \mid D_1^{T_1})| \right] \leq \sum_{t=1}^{\infty}\beta_t,
$$
and so, by Markov's inequality,
$$
\frac{1}{T_2-T_1}\sum_{t=T_1+1}^{T_1+T_2} |\Pr(\Delta(X_t) > \gamma_T^{1/8} \mid D_1^{T_1}) - \Pr(\Delta(X) > \gamma_T^{1/8} \mid D_1^{T_1})| \leq \frac{1}{\sqrt{T_2-T_1}}\sum_{t=1}^{\infty}\beta_t
$$ 
with probability at least $1-1/\sqrt{T_2-T_1}$. Therefore, by the union bound,
$$
\Pr(\Delta(X) > \gamma_T^{1/8} \mid D_1^{T_1}) \leq \frac{1}{\sqrt{T_2-T_1}}\sum_{t=1}^{\infty}\beta_t + \frac{\sqrt{\gamma_T}}{T_2 - T_1}
$$
with probability at least $1 - \gamma_T - 1/\sqrt{T_2 - T_1}$. Combining this bound with Assumption \ref{as: estimator Fhat}(i) shows that $\Ep[\Delta(X)^4\mid D_1^{T_1}] = o_P(1)$ and completes the proof of the lemma.
\end{proof}

We are now ready to proof Theorem \ref{thm: main result}:
\begin{proof}[Proof of Theorem \ref{thm: main result}]
Observe that
\begin{equation}\label{eq: law of large numbers}
\frac{1}{T_2}\sum_{t=T_1+1}^{T_1+T_2}X_tX_t'\to_P \Ep[XX']
\end{equation}
by Assumptions \ref{as: beta mixing} and \ref{as: standard convergence}(i) and Proposition 2.8 in \cite{FY05} since $\beta$-mixing coefficients dominate $\alpha$-mixing coefficients. Combining this result with Assumptions \ref{as: standard convergence}(ii,iii) and using the continuous mapping theorem and the Slutsky lemma gives the second convergence result in \eqref{eq: main convergence 2}.

To prove the first convergence result in \eqref{eq: main convergence 2}, denote
$$
r_{t1} = \int_{-\infty}^{+\infty}(\Psi(F(s|X_t))-\Psi(\widehat F(s|X_t)))ds + \int_{-\infty}^{+\infty}(\widehat F(s|X_t) - F(s|X_t))\psi(\widehat F(s|X_t))ds
$$
and
$$
r_{t2} = \int_{-\infty}^{+\infty}(F(s|X_t)-\mathbb I\{Y_t\leq s\})(\psi(\widehat F(s|X_t))-\psi(F(s|X_t)))ds
$$
for all $t=T_1+1,\dots,T_1+T_2$. Then
\begin{align*}
\sqrt{T_2}(\widehat\beta - \beta) & = \left(\frac{1}{T_2}\sum_{t=T_1+1}^{T_1+T_2}X_tX_t'\right)^{-1}\left(\frac{1}{\sqrt{T_2}}\sum_{t=T_1+1}^{T_1+T_2}X_te_t\right) \\
&\quad + \left(\frac{1}{T_2}\sum_{t=T_1+1}^{T_1+T_2}X_tX_t'\right)^{-1}\left(\frac{1}{\sqrt{T_2}}\sum_{t=T_1+1}^{T_1+T_2}X_t(r_{t1} + r_{t2})\right).
\end{align*}
Therefore, given that $T_2^{-1}\sum_{t=T_1+1}^{T_1+T_2} X_t X_t'$ converges to a positive-definite matrix by \eqref{eq: law of large numbers} and Assumption \ref{as: standard convergence}(ii), we only need to prove that
\begin{equation}\label{eq: need to prove thm 1 key}
E_1 = \frac{1}{\sqrt{T_2}}\sum_{t=T_1+1}^{T_1+T_2}X_tr_{t1} = o_P(1)\quad\text{and}\quad E_2 = \frac{1}{\sqrt{T_2}}\sum_{t=T_1+1}^{T_1+T_2}X_tr_{t2} = o_P(1).
\end{equation}
We do so in turn. In addition, as in the proof of Lemma \ref{lem: auxiliary 1}, we can decompose the function $\psi$ as $\psi = \psi_1 - \psi_2$, where the functions $\psi_1$ and $\psi_2$ are both bounded, increasing, and non-negative. Therefore, given that both $r_{t1}$ and $r_{t2}$ are linear in $\psi$ (and $\Psi$), it suffices to prove \eqref{eq: need to prove thm 1 key} assuming that the function $\psi$ is itself bounded, increasing, and non-negative. This is what we do below. (Note also that the new function $\psi$ still satisfies Assumption \ref{as: function psi}, and so Lemma \ref{lem: auxiliary 1} is still applicable.)

We start with $E_1$. Since $\Psi(s) = \int_0^s \psi(u)du$ and $\psi$ is increasing, the function $\Psi$ is convex, and so
$$
\Psi(F(s|X_t)) - \Psi(\widehat F(s|X_t)) = \Psi'(\widetilde F(s|X_t))(F(s|X_t) - \widehat F(s|X_t)),
$$ 
where $\widetilde F(s|X_t)$ belongs to the interval connecting $F(s|X_t)$ and $\widehat F(s|X_t)$, and $\Psi'(\widetilde F(s|X_t))$ is an element of the sub-differential of $\Psi(\widetilde F(s|X_t))$. Hence,
$$
r_{t1} = \int_{-\infty}^{+\infty}(\widehat F(s|X_t) - F(s|X_t))(\psi(\widehat F(s|X_t)) - \Psi'(\widetilde F(s|X_t))) ds,
$$
and so, for some constant $C>0$,
\begin{align*}
|r_{t1}| & \leq \int_{-\infty}^{+\infty}|\widehat F(s|X_t) - F(s|X_t)|\times |\psi(\widehat F(s|X_t)) - \Psi'(\widetilde F(s|X_t))| ds \\
& \leq \int_{-\infty}^{+\infty}|\widehat F(s|X_t) - F(s|X_t)|\times |\psi(\widehat F(s|X_t)) - \psi(F(s|X_t))| ds \leq C\Delta(X_t)^2,
\end{align*}
with probability approaching one uniformly over $t = T_1+1,\dots,T_1+T_2$, where the second inequality follows from the facts that the function $\psi$ is increasing and that $\Psi'(\widetilde F(s|X_t))\in [\psi(\widetilde F(s|X_t) - 0),\psi(\widetilde F(s|X_t) + 0)]$ and the third from Lemma \ref{lem: auxiliary 1}. Therefore,
\begin{equation}\label{eq: e1 - 1}
\|E_1\| \leq \frac{1}{\sqrt{T_2}}\sum_{t=T_1+1}^{T_1+T_2}\|X_t\|\times |r_{t1}| \leq \frac{C}{\sqrt{T_2}}\sum_{t=T_1+1}^{T_1+T_2}\|X_t\|\Delta(X_t)^2
\end{equation}
with probability approaching one. In addition,
\begin{equation}\label{eq: e1 - 2}
\frac{1}{\sqrt{T_2}}\sum_{t=T_1+1}^{T_1+T_2}\Ep[ \|X_t\|\Delta(X_t)^2 \mid D_1^{T_1}] \leq \sqrt{\sum_{t=T_1+1}^{T_1+T_2} \Ep[\|X_t\|^2\Delta(X_t)^4\mid D_1^{T_1}] } = o_P(1)
\end{equation}
by the Cauchy-Schwarz inequality and Assumption \ref{as: estimator Fhat}(ii). Combining \eqref{eq: e1 - 1} and \eqref{eq: e1 - 2} with Markov's inequality gives $E_1 = o_P(1)$.

Next, we consider $E_2$. Observe that
$$
\Ep\left[ X\int_{-\infty}^{+\infty}(F(s|X) - \mathbb I\{Y\leq s\})(\psi(\widehat F(s|X)) - \psi(F(s|X)))ds \mid D_1^{T_1} \right] = 0.
$$
Also, for some constant $C>0$,
\begin{align*}
&\Ep\left[ \left( \|X\|\int_{-\infty}^{+\infty}(F(s|X) - \mathbb I\{Y\leq s\})(\psi(\widehat F(s|X)) - \psi(F(s|X)))ds \right)^2  \mid D_1^{T_1} \right] \\
& \qquad \leq C\Ep\left[ \|X\|^2 \Delta(X)^2  \mid D_1^{T_1} \right]  + o_P(1) = o_P(1)
\end{align*}
by Lemmas \ref{lem: auxiliary 1} and \ref{lem: fourth lemma}. In addition,
\begin{align*}
&  \sum_{t=T_1+1}^{T_1+T_2} \Ep\left[ \left( \|X_t\|\int_{-\infty}^{+\infty}(F(s|X_t) - \mathbb I\{Y_t\leq s\})(\psi(\widehat F(s|X_t)) - \psi(F(s|X_t)))ds \right)^2 \mid D_1^{T_1} \right] \\
& \qquad \leq C \sum_{t=T_1+1}^{T_1+T_2} \Ep\left[ \|X_t\|^2\Delta(X_t)^2 \mid D_1^{T_1} \right] + o_P(1) \\
& \qquad \leq \sqrt{T_2-T_1}\sqrt{\sum_{t=T_1+1}^{T_1+T_2} \Ep\left[ \|X_t\|^4\Delta(X_t)^4\mid D_1^{T_1} \right] } + o_P(1) = o_P(T)
\end{align*}
by Lemma \ref{lem: auxiliary 1}, the Cauchy-Schwarz inequality, and Assumption \ref{as: estimator Fhat}(ii). Hence,
\begin{equation}\label{eq: mean e2}
\|\Ep[E_2\mid D_1^{T_1}] \|= \left\|\Ep\left[ \frac{1}{\sqrt{T_2}}\sum_{t=T_1+1}^{T_1+T_2} X_t r_{t2} \mid D_1^{T_1} \right] \right\|= o_P(1)
\end{equation}
by Lemma \ref{lem: auxiliary lemma 3}. Further,
$$
\sum_{t=T_1+1}^{T_1+T_2}\Ep\left[ (\|X_t\| \times |r_{t2}|)^4  \mid D_1^{T_1} \right] \leq C^2\sum_{t=T_1+1}^{T_1+T_2} \Ep\left[ \|X_t\|^4 \Delta(X_t)^4  \mid D_1^{T_1} \right] + o_P(1) = o_P(1)
$$
 by Lemma \ref{lem: auxiliary 1} and Assumption \ref{as: estimator Fhat}(ii). Hence, by Lemma \ref{lem: conditional variance},
\begin{equation}\label{eq: variance e2}
\|\textrm{Var}(E_2\mid D_1^{T_1})\| = o_P(1).
\end{equation}
Combining \eqref{eq: mean e2} and \eqref{eq: variance e2} gives $E_2 = o_P(1)$ and completes the proof of the theorem.
\end{proof}

\section{Proof of Theorem \ref{thm: variance estimation}}
Denote $w(0,m) = 1/2$, so that $\bar\Omega = \sum_{j=0}^m w(j,m)(\bar\Omega_j + \bar\Omega_j')$. Also, denote 
$$
\widehat\Omega = \sum_{j=0}^m w(j,m)(\widehat\Omega_j'+\widehat\Omega_{j}'),
\text{ so that }
 \widehat\Sigma = \left(\frac{1}{T_2}\sum_{t=T_1+1}^{T_1+T_2} X_tX_t'\right)^{-1}\widehat\Omega \left(\frac{1}{T_2}\sum_{t=T_1+1}^{T_1+T_2} X_tX_t'\right)^{-1}.
 $$
Then, recalling \eqref{eq: law of large numbers} from the proof of Theorem \ref{thm: main result} and observing that $\bar\Omega\to_P\Omega$ by Assumption \ref{as: extra conditions}(i), it follows that $\widehat\Sigma\to_P \Sigma$ as long as $\widehat\Omega - \bar\Omega \to_P 0$. Thus, it suffices to prove that $\sum_{j=0}^m w(j,m)(\widehat\Omega_j - \bar\Omega_j)\to_P0$. To do so, observe that for all $j=0,\dots,m$ and $t=T_1+j+1,\dots,T_1+T_2$, we have
$$
\widehat e_t\widehat e_{t-j} - e_te_{t-j} = e_{t-j}(\widehat e_t - e_t) + \widehat e_t (\widehat e_{t-j} - e_{t-j}),
$$
and so, denoting
$$
\mathcal S_1 = \sum_{j=0}^{m}\frac{w(j,m)}{T_2-T_1}\sum_{t=T_1+j+1}^{T_1+T_2} e_{t-j}(\widehat e_t - e_t)X_t X_{t-j}'
$$
and
$$
\mathcal S_2 = \sum_{j=0}^{m}\frac{w(j,m)}{T_2-T_1}\sum_{t=T_1+j+1}^{T_1+T_2}\widehat e_t (\widehat e_{t-j} - e_{t-j}) X_tX_{t-j}',
$$
we have $\sum_{j=0}^m w(j,m)(\widehat\Omega_j - \bar\Omega_j) = \mathcal S_1 + \mathcal S_2$. We will prove that $\mathcal S_1 \to_P 0$ and note that $\mathcal S_2 \to_P 0$ by a similar argument.

By the Cauchy-Schwarz inequality and Assumption \ref{as: extra conditions}(ii), 
\begin{align*}
\|\mathcal S_1\| & \leq \sum_{j=0}^m \frac{w(j,m)}{T_2-T_1}\sqrt{\sum_{t=T_1+j+1}^{T_1+T_2}\|e_{t-j}X_{t-j}\|^2}\sqrt{\sum_{t=T_1+j+1}^{T_1+T_2}\| (\widehat e_t - e_t)X_t \|^2} \\
& \leq \frac{m}{T_2-T_1}\sqrt{\sum_{t=T_1+1}^{T_1+T_2}\|e_{t}X_{t}\|^2}\sqrt{\sum_{t=T_1+1}^{T_1+T_2}\| (\widehat e_t - e_t)X_t \|^2}.
\end{align*}
Here, given that $\Ep[\|eX\|^2]\leq\sqrt{\Ep[e^4]\Ep[\|X\|^4]}<\infty$ by Assumption \ref{as: standard convergence}(i),
$$
\frac{1}{T_2-T_1}\sum_{t=T_1+1}^{T_1+T_2}\|e_tX_t\|^2 = O_P(1)
$$
by Assumptions \ref{as: beta mixing} and Proposition 2.8 in \cite{FY05}. Also,
\begin{align*}
&\frac{1}{T_2-T_1}\sum_{t=T_1+1}^{T_1+T_2}\|(\widehat e_t - e_t)X_t\|^2 \\
&\qquad \leq \frac{2}{T_2-T_1}\sum_{t=T_1+1}^{T_1+T_2}|r_{t1}+r_{t2}|^2\|X_t\|^2 + \frac{2}{T_2-T_1}\sum_{t=T_1+1}^{T_1+T_2}\|\widehat\beta - \beta\|^2\|X_t\|^4
\end{align*}
for $r_{t1}$ and $r_{t2}$ defined in the proof of Theorem \ref{thm: main result}. Moreover,
$$
\frac{1}{T_2-T_1}\sum_{t=T_1+1}^{T_1+T_2}\|X_t\|^4 = O_P(1)
$$
by Assumptions \ref{as: beta mixing} and \ref{as: standard convergence}(i) and Proposition 2.8 in \cite{FY05}. Therefore,
$$
\frac{1}{T_2-T_1}\sum_{t=T_1+1}^{T_1+T_2}\|\widehat\beta - \beta\|^2\|X_t\|^4 = O_P\left(\frac{1}{T}\right)
$$
by Theorem \ref{thm: main result}. In addition, as in the proof of Theorem \ref{thm: main result}, for some constant $C>0$, we have $|r_{t1}+r_{t2}| \leq C(\Delta(X_t)^2 + \Delta(X_t))$ with probability approaching one uniformly over $t=T_1+1,\dots,T_1+T_2$. Hence,
$$
\frac{1}{T_2-T_1}\sum_{t=T_1+1}^{T_1+T_2}|r_{t1}+r_{t2}|^2\|X_t\|^2 \leq \frac{2C^2}{T_2-T_1}\sum_{t=T_1+1}^{T_1+T_2}(\Delta(X_t)^4 + \Delta(X_t)^2)\|X_t\|^2 = o_P\left(\frac{1}{\sqrt T}\right)
$$
by Assumption \ref{as: estimator Fhat}(ii). We therefore conclude that $\|\mathcal S_1\| = o_P(m/T^{1/4}) = o_P(1)$ by Assumption \ref{as: extra conditions}(iii). Thus, given that $\|\mathcal S_2\| = o_P(1)$ by a similar argument, it follows that $\sum_{j=0}^m w(j,m)(\widehat\Omega_j - \bar\Omega_j)\to_P 0$, which completes the proof of the theorem.

\section{Weighted-Average Quantile Regression Estimators versus Parametric Estimators}\label{sec: comparison}
In this section, we compare our weighted-average quantile regression estimators with parametric estimators outlined in the Introduction. Recall that given a weighting function $\psi$, we define the parametric estimator by
$$
\widetilde\beta = \int_0^1 \widetilde\beta(u)\psi(u)du,
$$
where each $\widetilde\beta(u)$ is the classical (linear) $u$-quantile regression estimator of $Y$ on $X$. 

This parametric estimator is rather intuitive and is simple to implement. However, the key advantage of our weighted-average quantile regression estimator $\widehat\beta$ over the parametric estimator $\widetilde\beta$ is that our estimator is much more robust with respect to possible misspecification. In particular, our estimator requires fewer parametric assumptions for consistency. Indeed, we claim that consistency of the parametric estimator $\widetilde\beta$ can only be guaranteed under a {\em continuum} of constraints, namely $q_{Y|X}(u) = X'\beta(u)$ for all $u\in(0,1)$, whereas consistency of our estimator $\widehat\beta$, as discussed in the previous section, requires only one constraint: $\int_0^1 q_{Y|X}(u)\psi(u)du = X'\beta$.

To prove this claim, suppose that $\int_0^1 q_{Y|X}(u)\psi(u)du = X'\beta$ and recall that the classical $u$-quantile regression estimator
\begin{equation}\label{eq: quantile regression estimator naive}
\widetilde\beta(u) = \arg\min_{b\in\mathbb R^d} \left(\frac{u}{T} \sum_{t=1}^{T} (Y_t - X_t'b)_{+} + \frac{1-u}{T}\sum_{t=1}^{T} (Y_t - X_t'b)_{-}\right)
\end{equation}
converges in probability to
\begin{equation}\label{eq: quantile regression limit}
\bar\beta(u) = \arg\min_{b\in\mathbb R^d}\Big( u\Ep[(Y-X'b)_{+}] + (1-u)\Ep[(Y-X'b)_{-}] \Big),
\end{equation}
where for any random variable $Z$, we use $Z_{+} = Z\mathbb I\{Z\geq 0\}$ and $Z_{-} = Z\mathbb I\{Z<0\}$ to denote its positive and negative parts. Whenever $q_{Y|X}(u)$ is linear in $X$, i.e. $q_{Y|X}(u) = X'\beta(u)$ for some $\beta(u)$ almost surely, it is a standard exercise to show that $\bar\beta(u) = \beta(u)$ by taking the first-order conditions of \eqref{eq: quantile regression limit}, meaning that $\widetilde\beta(u)\to_P\beta(u)$, and so 
$$
\widetilde \beta = \int_0^1 \widetilde\beta(u)\psi(u)du \to_P \int_0^1 \beta(u)\psi(u)du = \beta,
$$
where the last equality follows from substituting $q_{Y|X}(u) = X'\beta(u)$ into the regression model $\int_0^1 q_{Y|X}(u)\psi(u)du = X'\beta$. On the other hand, whenever $q_{Y|X}(u)$ is not linear in $X$, we still have $\widetilde\beta(u)\to_P \bar\beta(u)$, so that $\widetilde\beta = \int_0^1\widetilde\beta(u)\psi(u)du \to_P \int_0^1 \bar\beta(u)\psi(u)du$, but in general $\int_0^1\bar\beta(u)\psi(u)du \neq \beta$ in this case. Indeed, consider the following data-generating process:
$$
Y = X\beta + X^2\gamma(U),
$$
where $X\sim U[0,2]$ and $U\sim U[0,1]$ are independent random variables, $\beta$ is any constant, and $\gamma(u) = 4u - 3$ for all $u\in[0,1]$. Suppose that $\psi(u) = 2\mathbb I\{u > 1/2\}$ for all $u\in (0,1)$. It is then easy to check that $\int_0^1 q_{Y|X}(u)\psi(u)du = X'\beta$ but, as we show below, $2\int_{1/2}^1\bar\beta(u)du = \beta + 19/6 - 8/\sqrt 6$. Therefore, the parametric estimator $\widetilde\beta$ is not consistent in this case, whereas our estimator $\widehat\beta$ is. Of course, the problem for the parametric estimator here is that $q_{Y|X}(u) = X\beta(u) + X^2\gamma(u)$ is not linear in $X$.

In addition, another advantage of our estimator $\widehat\beta$ over the parametric estimator $\widetilde\beta$ is that the latter requires estimating $u$-quantile regressions for values of $u$ that are close to the boundaries of the interval $[0,1]$. This is problematic because such quantile regression estimators may have a slow rate of convergence, undermining the properties of the estimator $\widetilde\beta$. In principle, one could consider a truncated version of $\widetilde\beta$, namely
$$
\widetilde\beta^{\varepsilon} = \int_{\varepsilon}^{1-\varepsilon} \widetilde\beta(u)\psi(u)du
$$
for some $\varepsilon = \varepsilon_T \to 0$ as $T\to\infty$ but in this case, one has to find a data-driven method to choose $\varepsilon$, and we are not aware of such methods. In contrast, although our estimator $\widehat\beta$ requires estimating the function $F$ via nonparametric/machine learning methods, which also rely on tuning parameters, there is a variety of methods in the literature, such as sample splitting and cross-validation, to choose these tuning parameters.

We now prove that $2\int_{1/2}^1 \bar\beta(u)du = 2\int_{1/2}^2 \beta(u)du + 19/6 - 8/\sqrt 6$. This calculation demonstrates that the parametric estimator described above is not consistent. Fix $u\in(0,1)$ and $b\in\mathbb R$ and denote $\tilde b = b - \beta$. First, consider the case $b \geq \beta$. In this case, we have
\begin{align*}
\Ep[(Y - Xb)_{+}\mid X] 
&= \Ep[( X^2\gamma(U) - X\tilde b )_{+}\mid X] = \int_0^{\infty} \Pr( X^2\gamma(U) - X\tilde b > s \mid X)ds \\
&  = \int_0^{\infty} \Pr\left( U > \frac{1}{4}\left(3 + \frac{\tilde b}{X} + \frac{s}{X^2}\right) \mid X\right)ds =
\begin{cases}
0 & \text{if }X\leq\tilde{b},\\
\frac{(X-\tilde{b})^{2}}{8} & \text{if }X>\tilde{b},
\end{cases}
\end{align*}
and
$$
\Ep[ (Y - Xb)_{-} \mid X] = \int_0^{\infty}\Pr\left( U < \frac{1}{4}\left(3 + \frac{\tilde b}{X} - \frac{s} {X^2}\right) \mid X\right) =
\begin{cases}
X^2 + X\tilde b & \text{if }X\leq\tilde{b},\\
\frac{(3X+\tilde{b})^{2}}{8} & \text{if }X>\tilde{b}.
\end{cases}
$$
Therefore, for $\beta \leq b < \beta + 2$,
$$
\frac{d}{db}\Ep[(Y - Xb)_{+}] = \Ep\left[\frac{d}{db}\Ep[(Y - Xb)_{+}\mid X]\right] = \Ep\left[\frac{\tilde b - X}{4}\mathbb I\{X > \tilde b\}\right] = -\frac{\tilde b^2}{16} + \frac{\tilde b}{4} - \frac{1}{4}
$$
and
\begin{align*}
\frac{d}{db}\Ep[(Y - Xb)_{-}]
& = \Ep\left[\frac{d}{db}\Ep[(Y - Xb)_{-}\mid X]\right] \\ 
& = \Ep\left[\frac{3X + \tilde b}{4}\mathbb I\{X > \tilde b\} + X\mathbb I\{X \leq \tilde b\}\right] = -\frac{\tilde b^2}{16} + \frac{\tilde b}{4} + \frac{3}{4},
\end{align*}
whereas for $b\geq \beta + 2$,
$$
\frac{d}{db}\Ep[(Y - Xb)_{+}] = 0\text{ and }\frac{d}{db}\Ep[(Y - Xb)_{-}] = \Ep[X] = 1.
$$
Next, consider the case $b < \beta$. In this case, we have
$$
\Ep[(Y - Xb)_{+}\mid X]  = 
\begin{cases}
- X^2 - X\tilde b & \text{if }X\leq - \tilde{b}/3,\\
\frac{(X-\tilde{b})^{2}}{8} & \text{if }X>-\tilde{b}/3.
\end{cases}
$$
and
$$
\Ep[ (Y - Xb)_{-} \mid X] = 
\begin{cases}
0 & \text{if }X\leq - \tilde{b}/3,\\
\frac{(3X+\tilde{b})^{2}}{8} & \text{if }X>-\tilde{b}/3.
\end{cases}
$$
Therefore, for $\beta - 6 < b < \beta$, 
\begin{align*}
\frac{d}{db}\Ep[(Y - Xb)_{+}] 
& = \Ep\left[\frac{d}{db}\Ep[(Y - Xb)_{+}\mid X]\right] \\
& = \Ep\left[\frac{\tilde b - X}{4}\mathbb I\{X > -\tilde b/3\} - X\mathbb I\{X \leq -\tilde b / 3\}\right] = \frac{\tilde b^2}{48} + \frac{\tilde b}{4} - \frac{1}{4}
\end{align*}
and
$$
\frac{d}{db}\Ep[(Y - Xb)_{-}] = \Ep\left[\frac{d}{db}\Ep[(Y - Xb)_{-}\mid X]\right] = \Ep\left[\frac{3X + \tilde b}{4}\mathbb I\{X > -\tilde b/3\}\right] = \frac{\tilde b^2}{48} + \frac{\tilde b}{4} + \frac{3}{4},
$$
whereas for $b \leq \beta - 6$,
$$
\frac{d}{db}\Ep[(Y - Xb)_{+}] = -\Ep[X] =  -1\text{ and }\frac{d}{db}\Ep[(Y - Xb)_{-}] = 0.
$$
Hence,
$$
\frac{d}{db}\Big\{u\Ep[(Y - Xb)_{+}] + (1-u)\Ep[(Y - Xb)_{-}]\Big\} = 
\begin{cases}
-1 & \text{if }b \leq \beta - 6,\\
\frac{\tilde b^2}{48} +\frac{\tilde b}{4} + \frac{3}{4} - u & \text{if }\beta - 6 < b <\beta,\\
-\frac{\tilde b^2}{16} + \frac{\tilde b}{4} + \frac{3}{4} - u & \text{if }\beta\leq b<\beta + 2,\\
+1 &\text{ if } b \geq \beta + 2.
\end{cases}
$$
Thus, by the first-order conditions, the solution to the optimization problem in \eqref{eq: quantile regression limit} is
$$
\bar\beta(u) = 
\begin{cases}
\beta -6 + 4\sqrt{3u} &\text{if } u < 3/4,\\
\beta + 2 - 4\sqrt{1 - u} & \text{if } u \geq 3/4.
\end{cases}
$$
We conclude that
$$
2\int_{1/2}^{1}\bar\beta(u)du = \beta + 2\int_{1/2}^{3/4}(-6+4\sqrt{3u})du + 2\int_{3/4}^1(2-4\sqrt{1-u})du = \beta + \frac{19}{6} - \frac{8}{\sqrt 6}\neq \beta.
$$
This means that the parametric estimator described above is not consistent.

\section{Additional Tables \& Figures}

\setcounter{table}{0} \renewcommand{\thetable}{A.\arabic{table}} \setcounter{figure}{0} \renewcommand{\thefigure}{A.\arabic{figure}}

\begin{table}[H]
\caption{{\small Results of Monte Carlo simulation study for the coverage probability of 90\% confidence intervals.\label{table: simulation results 1}}}
{\footnotesize
\begin{tabular}{cc|cccc|cccc}
\hline
\hline 
\multicolumn{10}{c}{Panel A: Homoscedastic Noise}\tabularnewline
\hline 
\hline 
\multirow{3}{*}{$\psi$-type} & \multirow{3}{*}{$\beta_{1}$} & \multicolumn{4}{c|}{$e\sim N(0,1)$} & \multicolumn{4}{c}{$e\sim t(4)$}\tabularnewline
 &  & \multicolumn{2}{c}{$p=2$} & \multicolumn{2}{c|}{$p=5$} & \multicolumn{2}{c}{$p=2$} & \multicolumn{2}{c}{$p=5$}\tabularnewline
 &  & $T=1000$ & $T=2000$ & $T=1000$ & $T=2000$ & $T=1000$ & $T=2000$ & $T=1000$ & $T=2000$\tabularnewline
\hline 
\multirow{4}{*}{1} & .0 & 0.912 & 0.918 & 0.894 & 0.874 & 0.896 & 0.898 & 0.874 & 0.874\tabularnewline
 & .3 & 0.926 & 0.9 & 0.918 & 0.872 & 0.892 & 0.888 & 0.86 & 0.882\tabularnewline
 & .6 & 0.906 & 0.89 & 0.906 & 0.888 & 0.916 & 0.898 & 0.866 & 0.878\tabularnewline
 & .9 & 0.912 & 0.884 & 0.886 & 0.876 & 0.886 & 0.892 & 0.868 & 0.886\tabularnewline
\hline 
\multirow{4}{*}{2} & .0 & 0.882 & 0.908 & 0.922 & 0.888 & 0.91 & 0.888 & 0.89 & 0.882\tabularnewline
 & .3 & 0.89 & 0.908 & 0.896 & 0.898 & 0.908 & 0.902 & 0.876 & 0.9\tabularnewline
 & .6 & 0.9 & 0.906 & 0.878 & 0.866 & 0.89 & 0.89 & 0.874 & 0.898\tabularnewline
 & .9 & 0.892 & 0.892 & 0.87 & 0.89 & 0.902 & 0.884 & 0.854 & 0.886\tabularnewline
\hline 
\multirow{4}{*}{3} & .0 & 0.884 & 0.874 & 0.892 & 0.904 & 0.874 & 0.896 & 0.888 & 0.898\tabularnewline
 & .3 & 0.878 & 0.878 & 0.884 & 0.88 & 0.872 & 0.912 & 0.878 & 0.896\tabularnewline
 & .6 & 0.89 & 0.886 & 0.876 & 0.89 & 0.878 & 0.898 & 0.882 & 0.908\tabularnewline
 & .9 & 0.88 & 0.874 & 0.886 & 0.886 & 0.884 & 0.896 & 0.87 & 0.894\tabularnewline
\hline 
\multirow{4}{*}{4} & .0 & 0.92 & 0.906 & 0.912 & 0.884 & 0.906 & 0.896 & 0.866 & 0.89\tabularnewline
 & .3 & 0.92 & 0.914 & 0.914 & 0.876 & 0.902 & 0.888 & 0.86 & 0.882\tabularnewline
 & .6 & 0.918 & 0.914 & 0.912 & 0.874 & 0.914 & 0.882 & 0.848 & 0.896\tabularnewline
 & .9 & 0.928 & 0.906 & 0.9 & 0.876 & 0.898 & 0.894 & 0.872 & 0.892\tabularnewline
\hline 
\hline 
\multicolumn{10}{c}{Panel B: Heteroscedastic noise}\tabularnewline
\hline 
\hline 
\multirow{3}{*}{$\psi$-type} & \multirow{3}{*}{$\beta_{1}$} & \multicolumn{4}{c|}{$e\sim N(0,1)$} & \multicolumn{4}{c}{$e\sim t(4)$}\tabularnewline
 &  & \multicolumn{2}{c}{$p=2$} & \multicolumn{2}{c|}{$p=5$} & \multicolumn{2}{c}{$p=2$} & \multicolumn{2}{c}{$p=5$}\tabularnewline
 &  & $T=1000$ & $T=2000$ & $T=1000$ & $T=2000$ & $T=1000$ & $T=2000$ & $T=1000$ & $T=2000$\tabularnewline
\hline 
\multirow{4}{*}{1} & .0 & 0.916 & 0.906 & 0.872 & 0.862 & 0.884 & 0.888 & 0.846 & 0.88\tabularnewline
 & .3 & 0.906 & 0.906 & 0.916 & 0.848 & 0.886 & 0.88 & 0.862 & 0.868\tabularnewline
 & .6 & 0.902 & 0.908 & 0.91 & 0.886 & 0.88 & 0.89 & 0.844 & 0.872\tabularnewline
 & .9 & 0.902 & 0.902 & 0.912 & 0.88 & 0.876 & 0.886 & 0.842 & 0.862\tabularnewline
\hline 
\multirow{4}{*}{2} & .0 & 0.882 & 0.896 & 0.894 & 0.884 & 0.888 & 0.87 & 0.872 & 0.888\tabularnewline
 & .3 & 0.886 & 0.884 & 0.906 & 0.876 & 0.884 & 0.878 & 0.88 & 0.894\tabularnewline
 & .6 & 0.886 & 0.91 & 0.894 & 0.88 & 0.882 & 0.892 & 0.854 & 0.89\tabularnewline
 & .9 & 0.866 & 0.884 & 0.898 & 0.86 & 0.858 & 0.866 & 0.87 & 0.896\tabularnewline
\hline 
\multirow{4}{*}{3} & .0 & 0.89 & 0.87 & 0.884 & 0.882 & 0.886 & 0.918 & 0.878 & 0.9\tabularnewline
 & .3 & 0.884 & 0.876 & 0.88 & 0.87 & 0.88 & 0.906 & 0.868 & 0.898\tabularnewline
 & .6 & 0.886 & 0.88 & 0.878 & 0.884 & 0.892 & 0.914 & 0.876 & 0.908\tabularnewline
 & .9 & 0.888 & 0.868 & 0.886 & 0.868 & 0.886 & 0.906 & 0.88 & 0.902\tabularnewline
\hline 
\multirow{4}{*}{4} & .0 & 0.898 & 0.916 & 0.882 & 0.876 & 0.876 & 0.894 & 0.846 & 0.878\tabularnewline
 & .3 & 0.926 & 0.912 & 0.9 & 0.86 & 0.878 & 0.892 & 0.84 & 0.884\tabularnewline
 & .6 & 0.914 & 0.912 & 0.918 & 0.864 & 0.88 & 0.892 & 0.842 & 0.878\tabularnewline
 & .9 & 0.91 & 0.898 & 0.91 & 0.876 & 0.876 & 0.874 & 0.844 & 0.882\tabularnewline
\hline 
\hline
\end{tabular}
}
\end{table}

\newpage
\begin{table}
\caption{{\small Results of Monte Carlo simulation study for the mean absolute error.\label{table: simulation results 2}}}
{\footnotesize
\begin{tabular}{cc|cccc|cccc}
\hline
\hline 
\multicolumn{10}{c}{Panel A: DGP1, Homoscedastic Noise}\tabularnewline
\hline 
\hline 
\multirow{3}{*}{$\psi$-type} & \multirow{3}{*}{$\beta_{1}$} & \multicolumn{4}{c|}{$e\sim N(0,1)$} & \multicolumn{4}{c}{$e\sim t(4)$}\tabularnewline
 &  & \multicolumn{2}{c}{$p=2$} & \multicolumn{2}{c|}{$p=5$} & \multicolumn{2}{c}{$p=2$} & \multicolumn{2}{c}{$p=5$}\tabularnewline
 &  & $T=1000$ & $T=2000$ & $T=1000$ & $T=2000$ & $T=1000$ & $T=2000$ & $T=1000$ & $T=2000$\tabularnewline
\hline 
\multirow{4}{*}{1} & .0 & 0.158 & 0.107 & 0.171 & 0.12 & 0.373 & 0.255 & 0.381 & 0.254\tabularnewline
 & .3 & 0.152 & 0.109 & 0.166 & 0.123 & 0.368 & 0.252 & 0.379 & 0.257\tabularnewline
 & .6 & 0.151 & 0.113 & 0.171 & 0.123 & 0.369 & 0.253 & 0.376 & 0.259\tabularnewline
 & .9 & 0.152 & 0.115 & 0.175 & 0.123 & 0.368 & 0.243 & 0.375 & 0.255\tabularnewline
\hline 
\multirow{4}{*}{2} & .0 & 0.214 & 0.152 & 0.215 & 0.154 & 0.478 & 0.339 & 0.485 & 0.316\tabularnewline
 & .3 & 0.212 & 0.147 & 0.224 & 0.153 & 0.463 & 0.332 & 0.472 & 0.319\tabularnewline
 & .6 & 0.21 & 0.147 & 0.219 & 0.158 & 0.464 & 0.333 & 0.465 & 0.318\tabularnewline
 & .9 & 0.2 & 0.15 & 0.229 & 0.159 & 0.467 & 0.321 & 0.483 & 0.319\tabularnewline
\hline 
\multirow{4}{*}{3} & .0 & 0.076 & 0.057 & 0.074 & 0.054 & 0.086 & 0.058 & 0.088 & 0.058\tabularnewline
 & .3 & 0.078 & 0.058 & 0.074 & 0.056 & 0.086 & 0.058 & 0.089 & 0.058\tabularnewline
 & .6 & 0.077 & 0.058 & 0.075 & 0.056 & 0.087 & 0.059 & 0.088 & 0.057\tabularnewline
 & .9 & 0.079 & 0.06 & 0.077 & 0.056 & 0.089 & 0.06 & 0.093 & 0.059\tabularnewline
\hline 
\multirow{4}{*}{4} & .0 & 0.14 & 0.097 & 0.153 & 0.104 & 0.326 & 0.226 & 0.336 & 0.218\tabularnewline
 & .3 & 0.138 & 0.097 & 0.152 & 0.108 & 0.324 & 0.218 & 0.336 & 0.221\tabularnewline
 & .6 & 0.136 & 0.099 & 0.156 & 0.109 & 0.327 & 0.219 & 0.335 & 0.222\tabularnewline
 & .9 & 0.131 & 0.1 & 0.152 & 0.107 & 0.323 & 0.215 & 0.331 & 0.224\tabularnewline
\hline 
\hline 
\multicolumn{10}{c}{Panel B: DGP2, Heteroscedastic noise}\tabularnewline
\hline 
\hline 
\multirow{3}{*}{$\psi$-type} & \multirow{3}{*}{$\beta_{1}$} & \multicolumn{4}{c|}{$e\sim N(0,1)$} & \multicolumn{4}{c}{$e\sim t(4)$}\tabularnewline
 &  & \multicolumn{2}{c}{$p=2$} & \multicolumn{2}{c|}{$p=5$} & \multicolumn{2}{c}{$p=2$} & \multicolumn{2}{c}{$p=5$}\tabularnewline
 &  & $T=1000$ & $T=2000$ & $T=1000$ & $T=2000$ & $T=1000$ & $T=2000$ & $T=1000$ & $T=2000$\tabularnewline
\hline 
\multirow{4}{*}{1} & .0 & 0.201 & 0.143 & 0.232 & 0.159 & 0.479 & 0.335 & 0.497 & 0.32\tabularnewline
 & .3 & 0.204 & 0.147 & 0.219 & 0.16 & 0.487 & 0.336 & 0.494 & 0.327\tabularnewline
 & .6 & 0.206 & 0.152 & 0.234 & 0.168 & 0.483 & 0.332 & 0.494 & 0.331\tabularnewline
 & .9 & 0.203 & 0.146 & 0.222 & 0.161 & 0.48 & 0.324 & 0.482 & 0.329\tabularnewline
\hline 
\multirow{4}{*}{2} & .0 & 0.278 & 0.194 & 0.278 & 0.194 & 0.596 & 0.435 & 0.634 & 0.407\tabularnewline
 & .3 & 0.277 & 0.194 & 0.279 & 0.199 & 0.603 & 0.43 & 0.624 & 0.399\tabularnewline
 & .6 & 0.272 & 0.191 & 0.279 & 0.208 & 0.601 & 0.426 & 0.622 & 0.401\tabularnewline
 & .9 & 0.267 & 0.191 & 0.278 & 0.216 & 0.604 & 0.427 & 0.602 & 0.406\tabularnewline
\hline 
\multirow{4}{*}{3} & .0 & 0.098 & 0.074 & 0.095 & 0.071 & 0.112 & 0.075 & 0.113 & 0.073\tabularnewline
 & .3 & 0.097 & 0.074 & 0.096 & 0.072 & 0.108 & 0.075 & 0.111 & 0.073\tabularnewline
 & .6 & 0.098 & 0.074 & 0.095 & 0.071 & 0.111 & 0.075 & 0.111 & 0.073\tabularnewline
 & .9 & 0.099 & 0.076 & 0.095 & 0.072 & 0.111 & 0.076 & 0.112 & 0.073\tabularnewline
\hline 
\multirow{4}{*}{4} & .0 & 0.183 & 0.127 & 0.2 & 0.138 & 0.425 & 0.294 & 0.429 & 0.278\tabularnewline
 & .3 & 0.184 & 0.128 & 0.202 & 0.144 & 0.429 & 0.291 & 0.432 & 0.282\tabularnewline
 & .6 & 0.187 & 0.132 & 0.205 & 0.149 & 0.42 & 0.291 & 0.439 & 0.286\tabularnewline
 & .9 & 0.182 & 0.131 & 0.204 & 0.15 & 0.427 & 0.291 & 0.434 & 0.284\tabularnewline
\hline 
\hline
\end{tabular}
}
\end{table}

\clearpage

\begin{center}
\begin{figure}[H]
\caption{Coefficient by Percentiles\label{fig:Coefficient-by-Percentiles}}
\caption*{
This figure plots the coefficient estimates and the 95\% confidence
intervals for the 1\% to 10\% quantile regressions. The dependent
variables are the excess returns of the Fama-French 5 industries.
The dependent variables are the Fama-French 5 factors. The estimates
are multiplied by -1 to be consistent with the risk regressions.
}
\begin{centering}
\subfloat{\includegraphics[width=7cm]{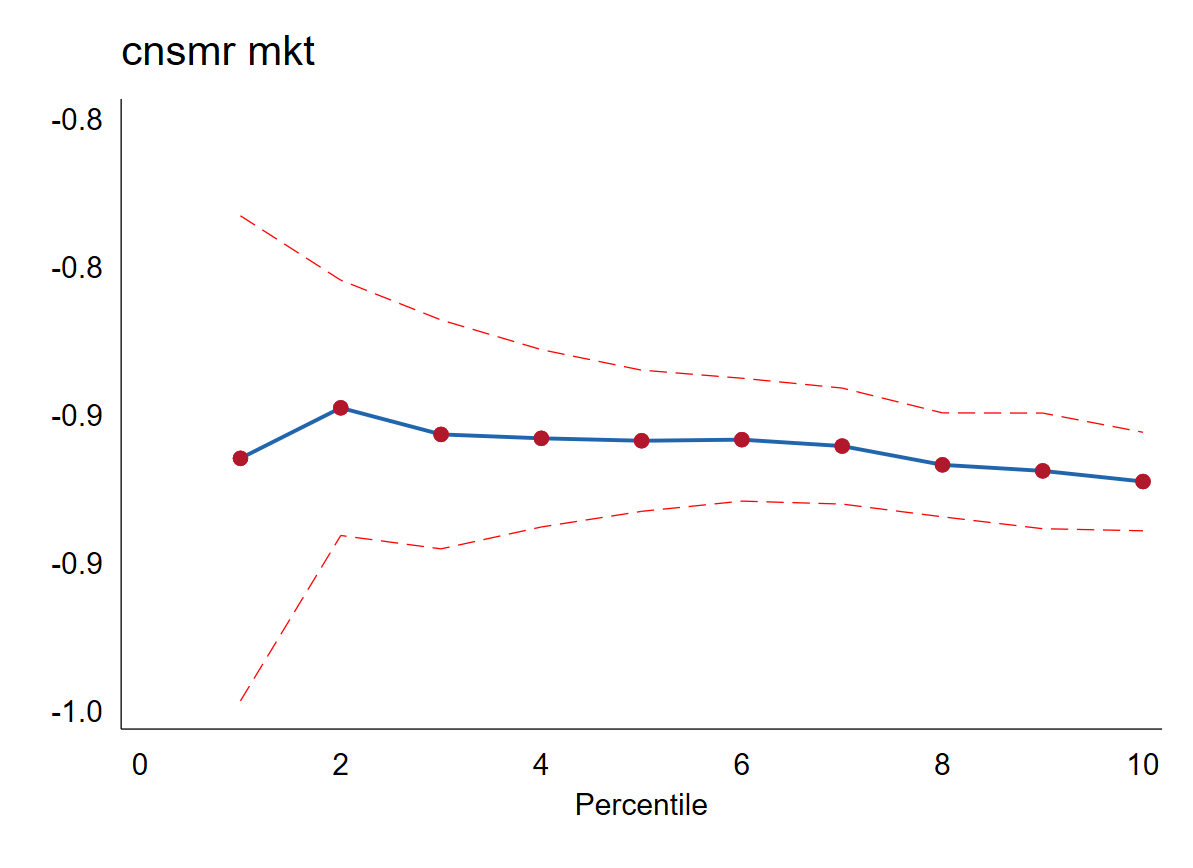}}\subfloat{\includegraphics[width=7cm]{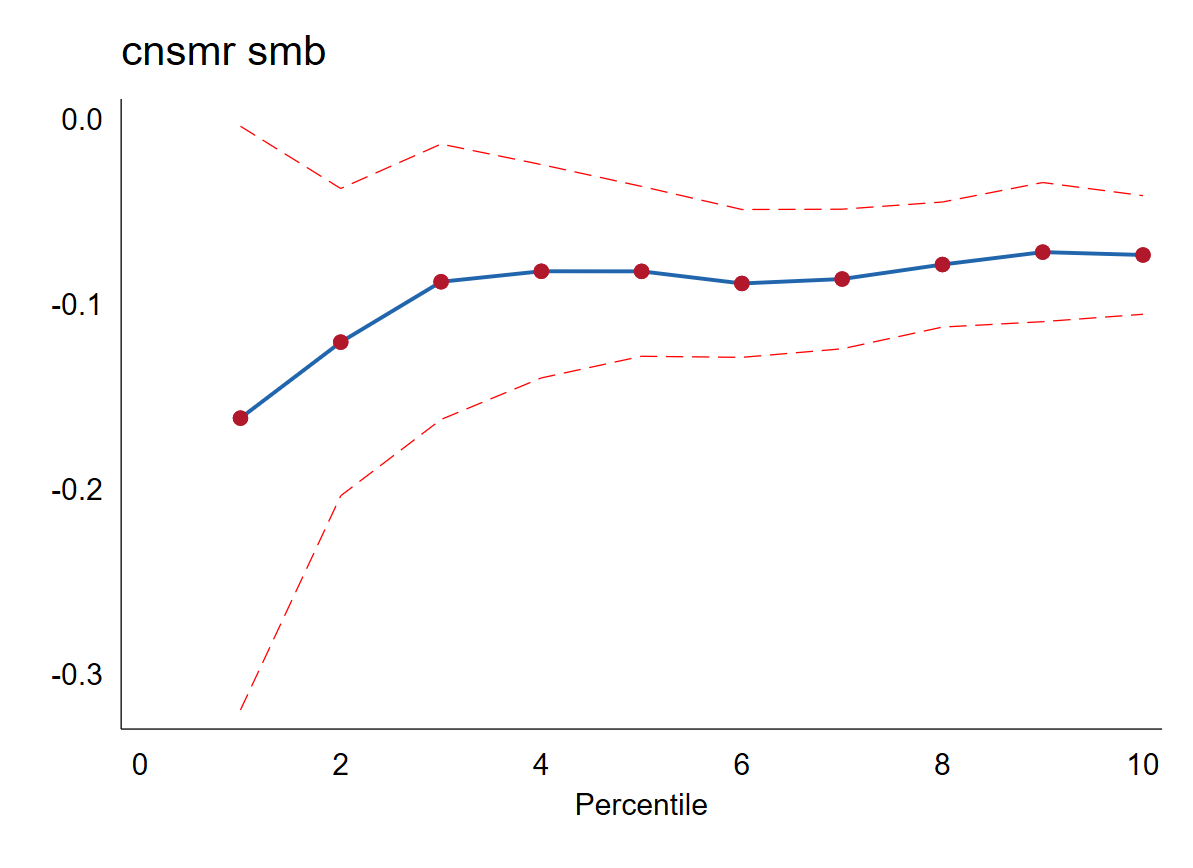}}
\par\end{centering}
\begin{centering}
\subfloat{\includegraphics[width=7cm]{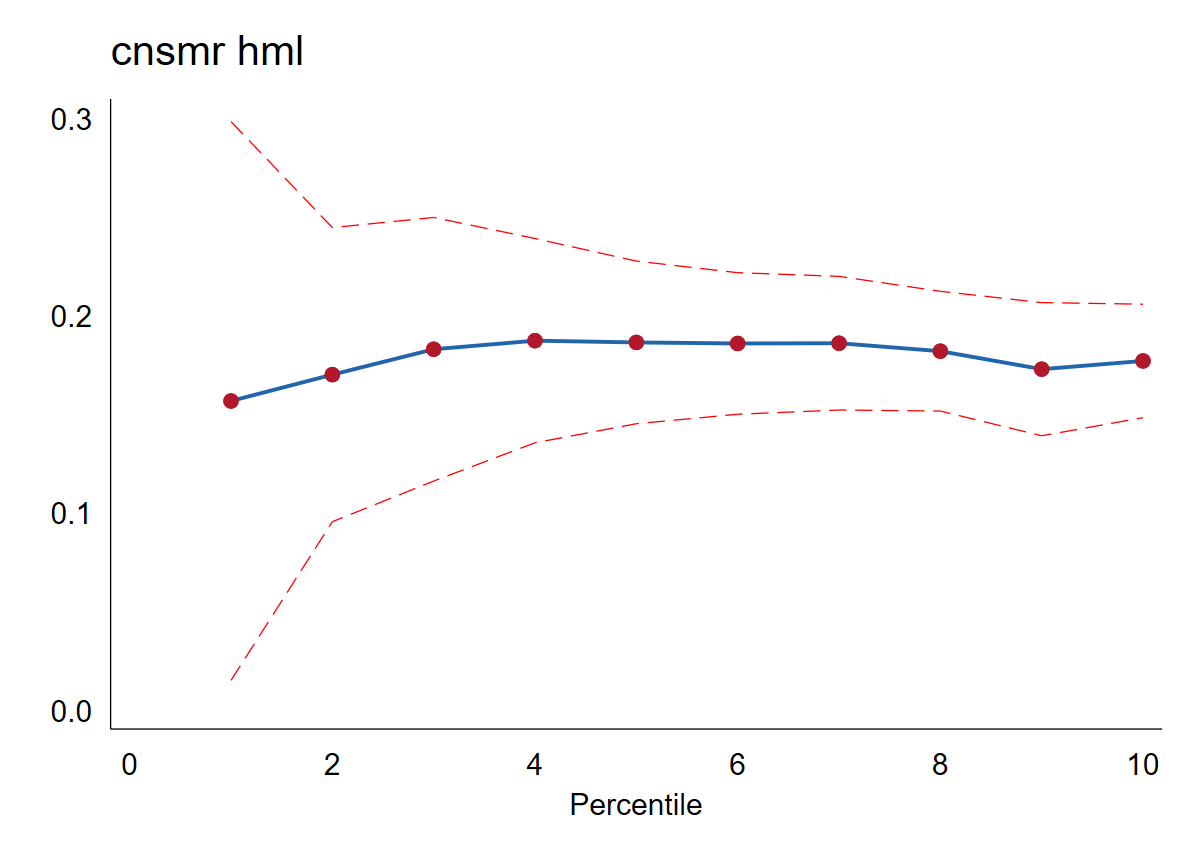}}\subfloat{\includegraphics[width=7cm]{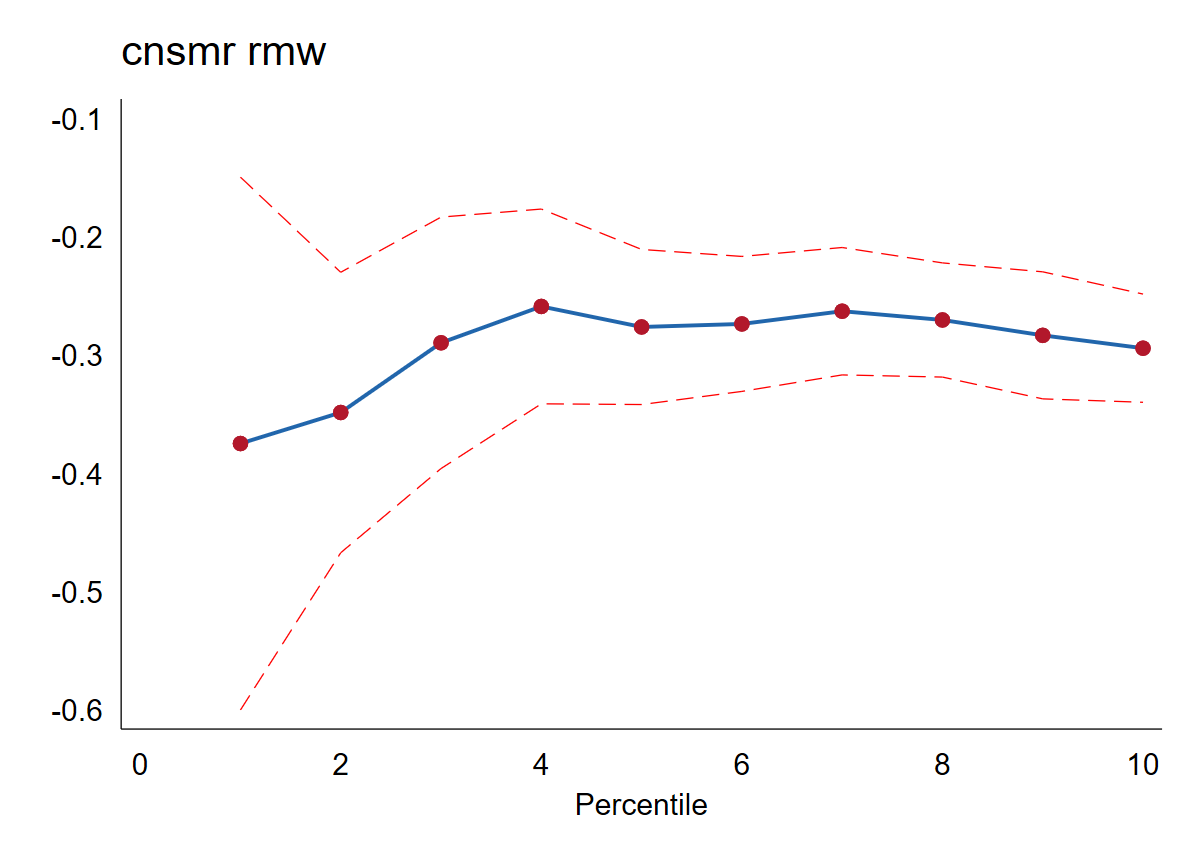}}
\par\end{centering}
\begin{centering}
\subfloat{\includegraphics[width=7cm]{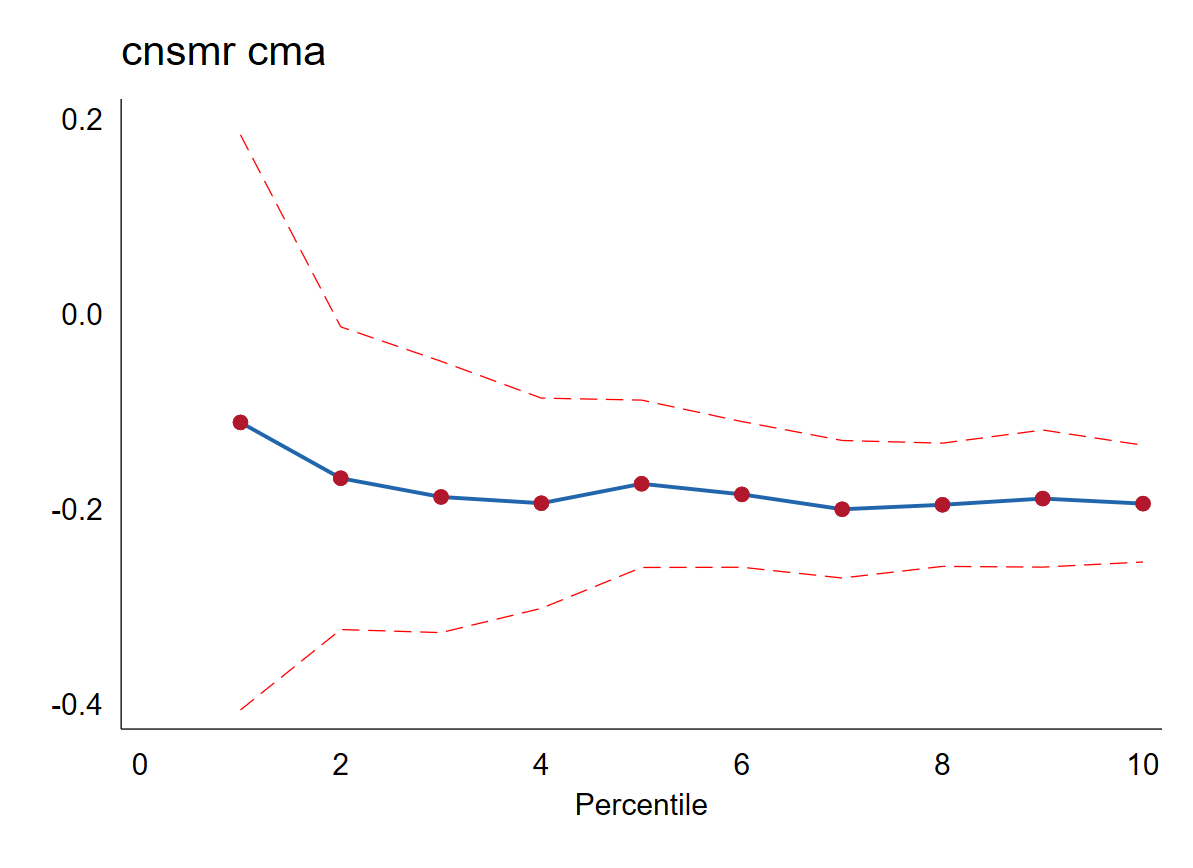}}\subfloat{\includegraphics[width=7cm]{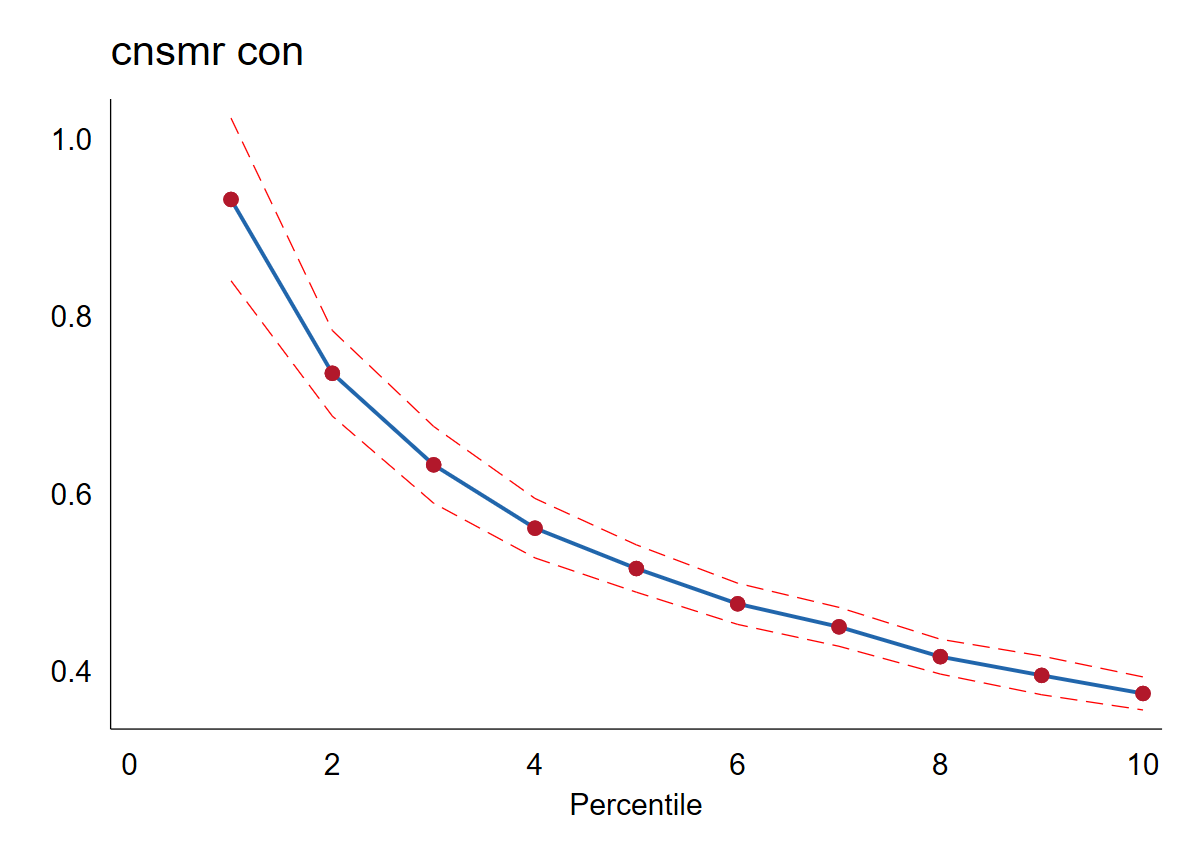}}
\par\end{centering}
\end{figure}
\par\end{center}

\begin{center}
\begin{figure}[H]
\begin{centering}
\subfloat{\includegraphics[width=7cm]{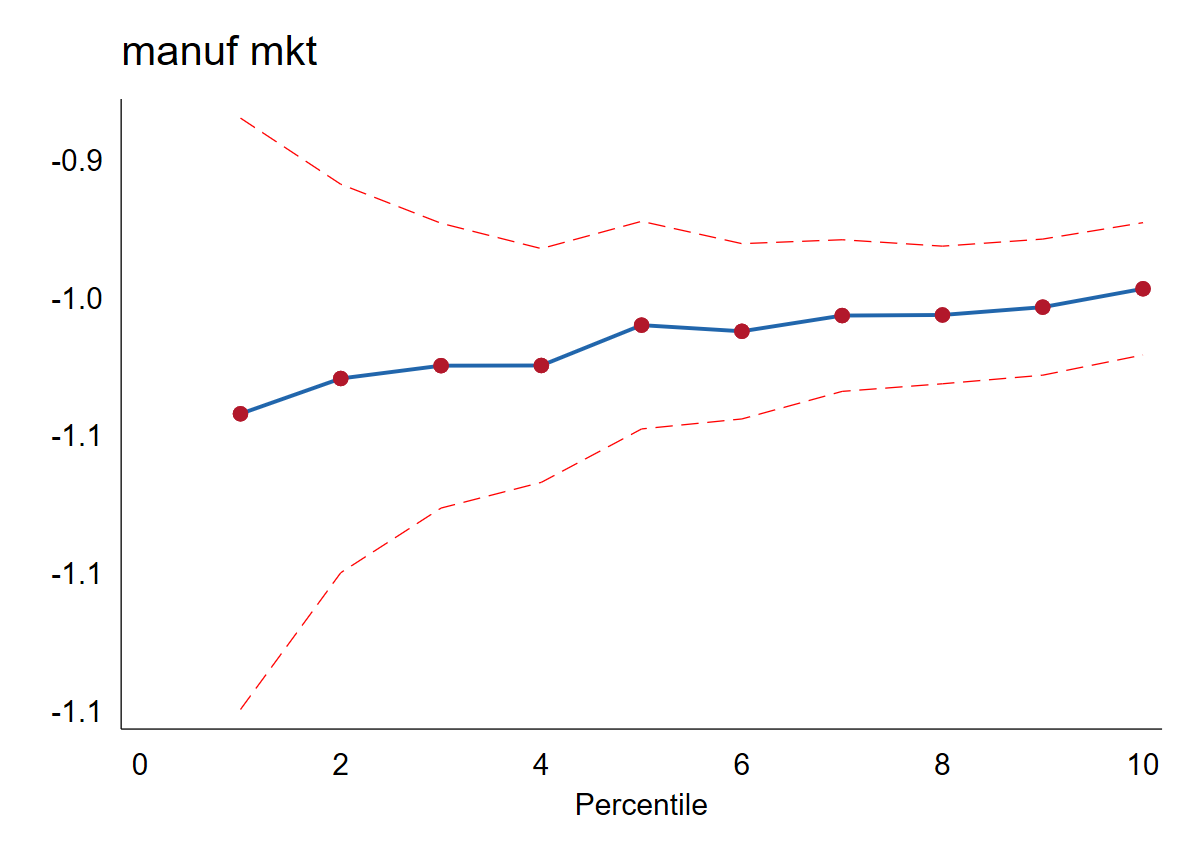}}\subfloat{\includegraphics[width=7cm]{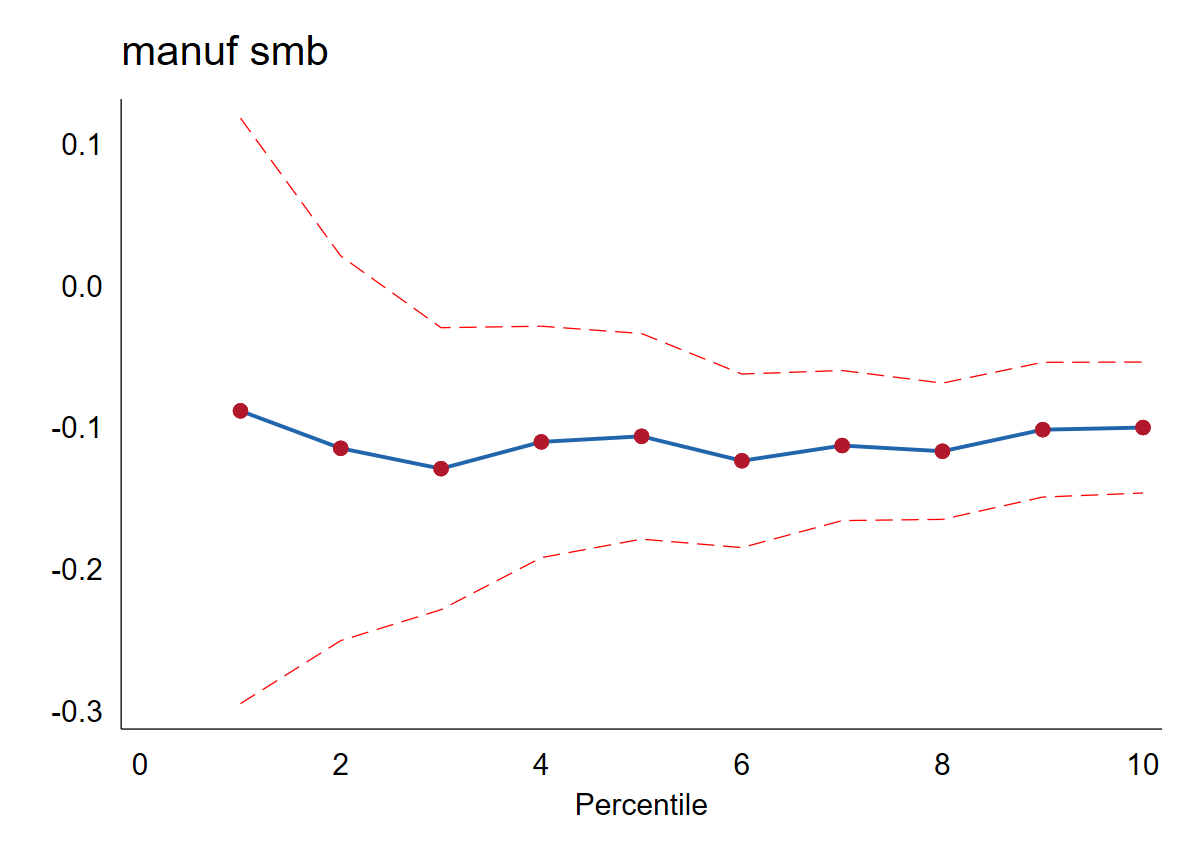}}
\par\end{centering}
\begin{centering}
\subfloat{\includegraphics[width=7cm]{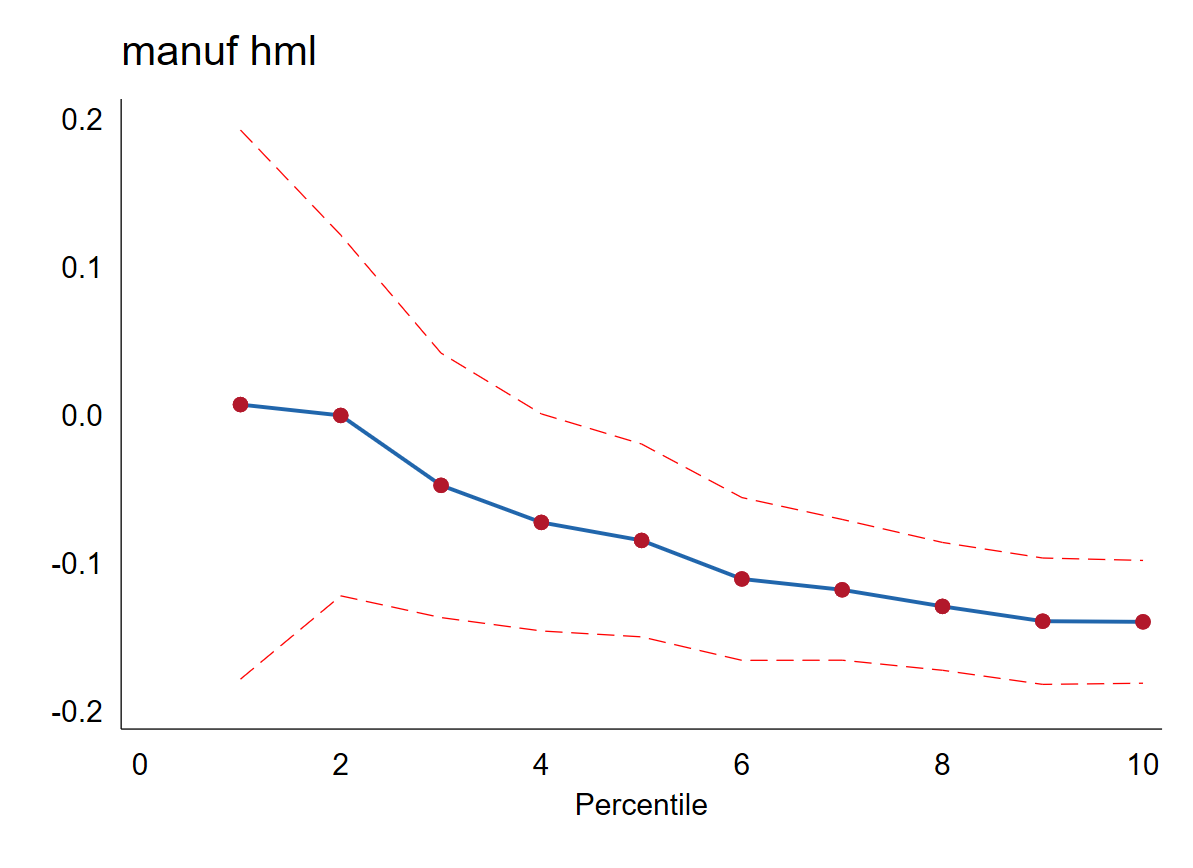}}\subfloat{\includegraphics[width=7cm]{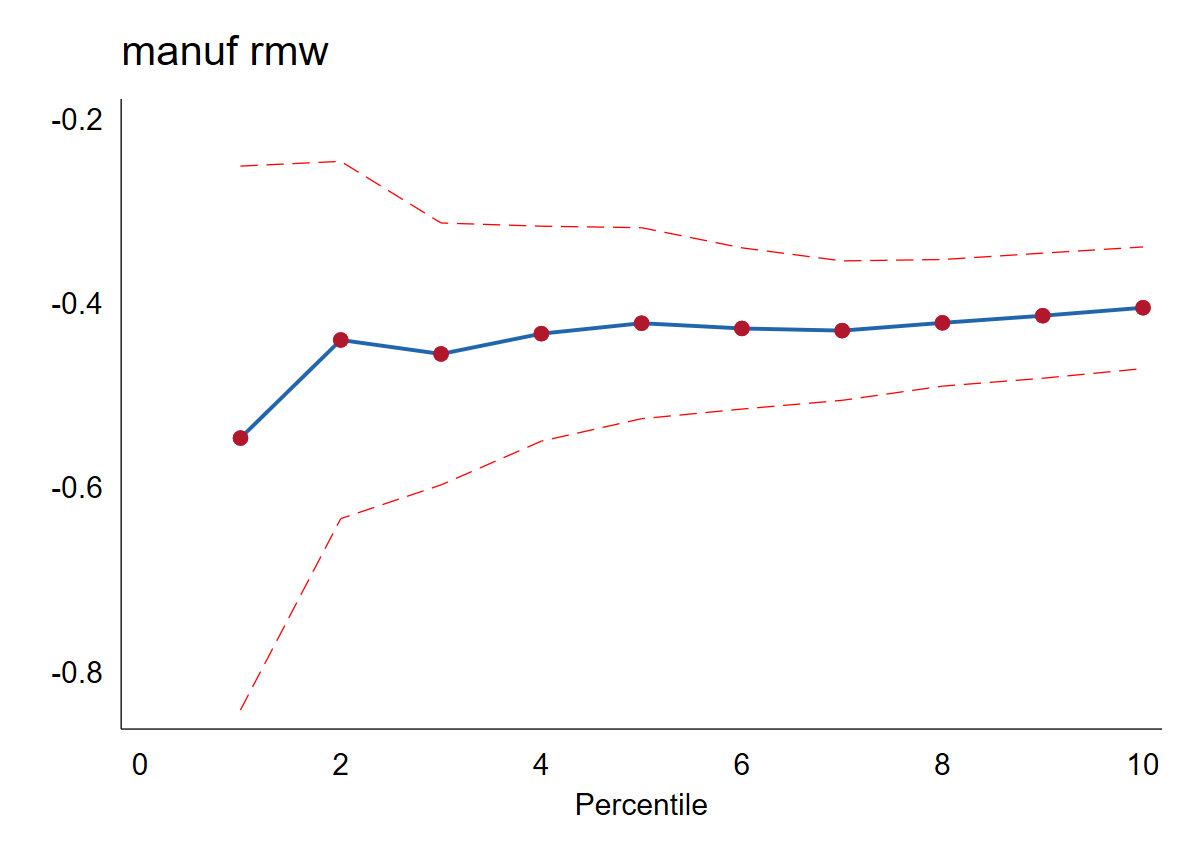}}
\par\end{centering}
\centering{}\subfloat{\includegraphics[width=7cm]{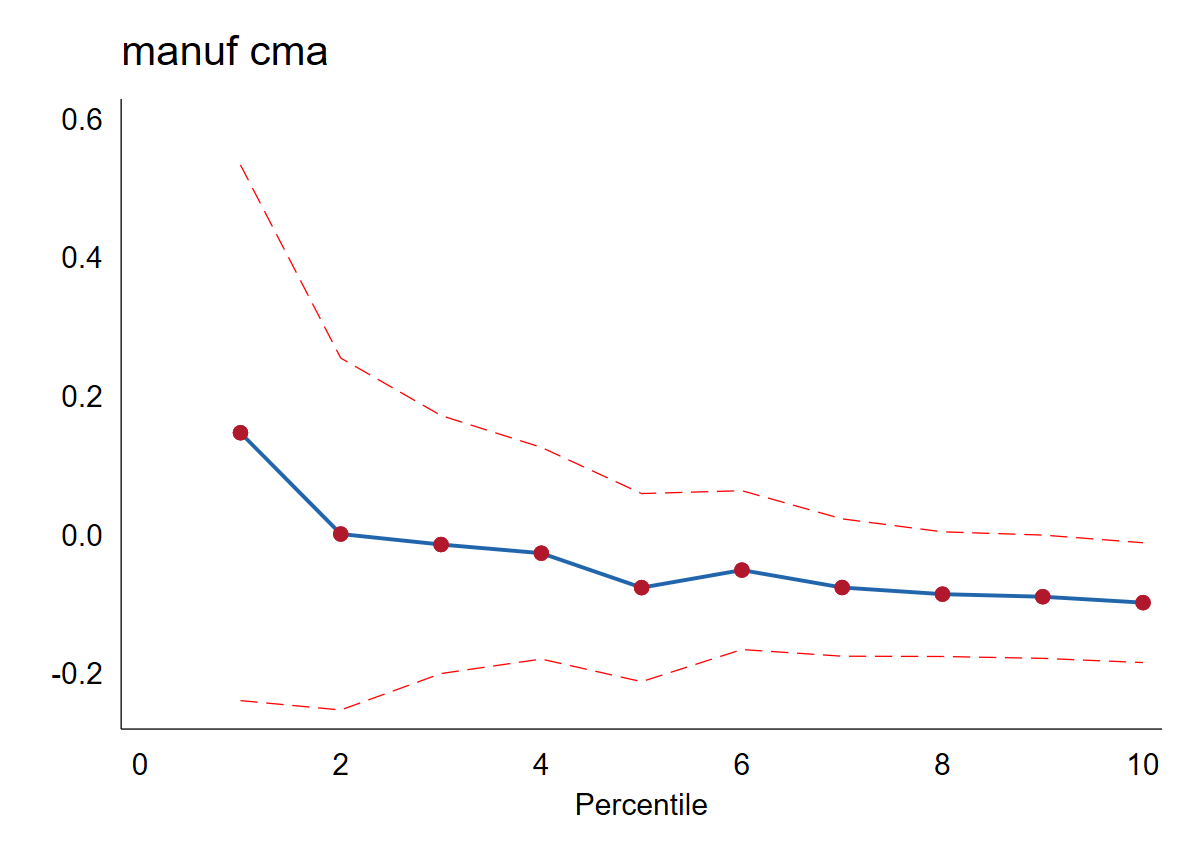}}\subfloat{\includegraphics[width=7cm]{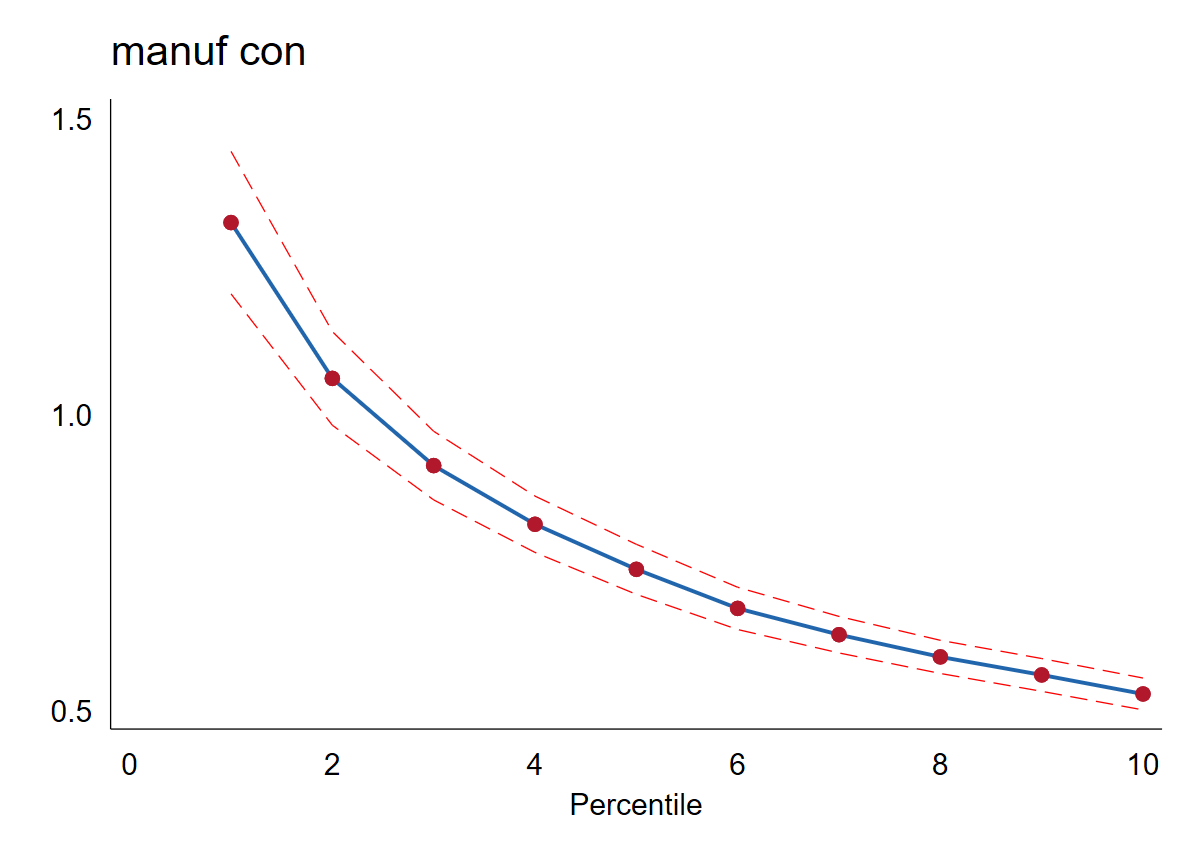}}
\end{figure}
\par\end{center}

\begin{center}
\begin{figure}[H]
\begin{centering}
\subfloat{\includegraphics[width=7cm]{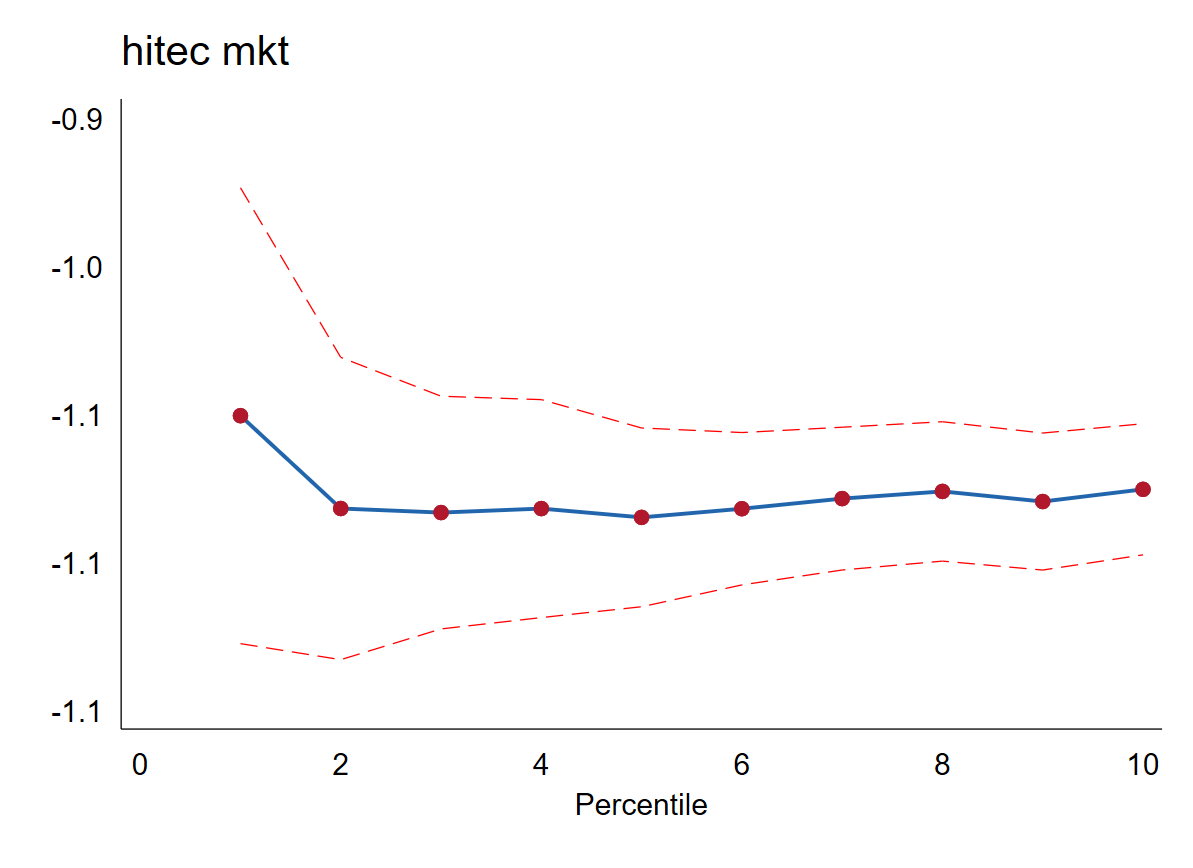}}\subfloat{\includegraphics[width=7cm]{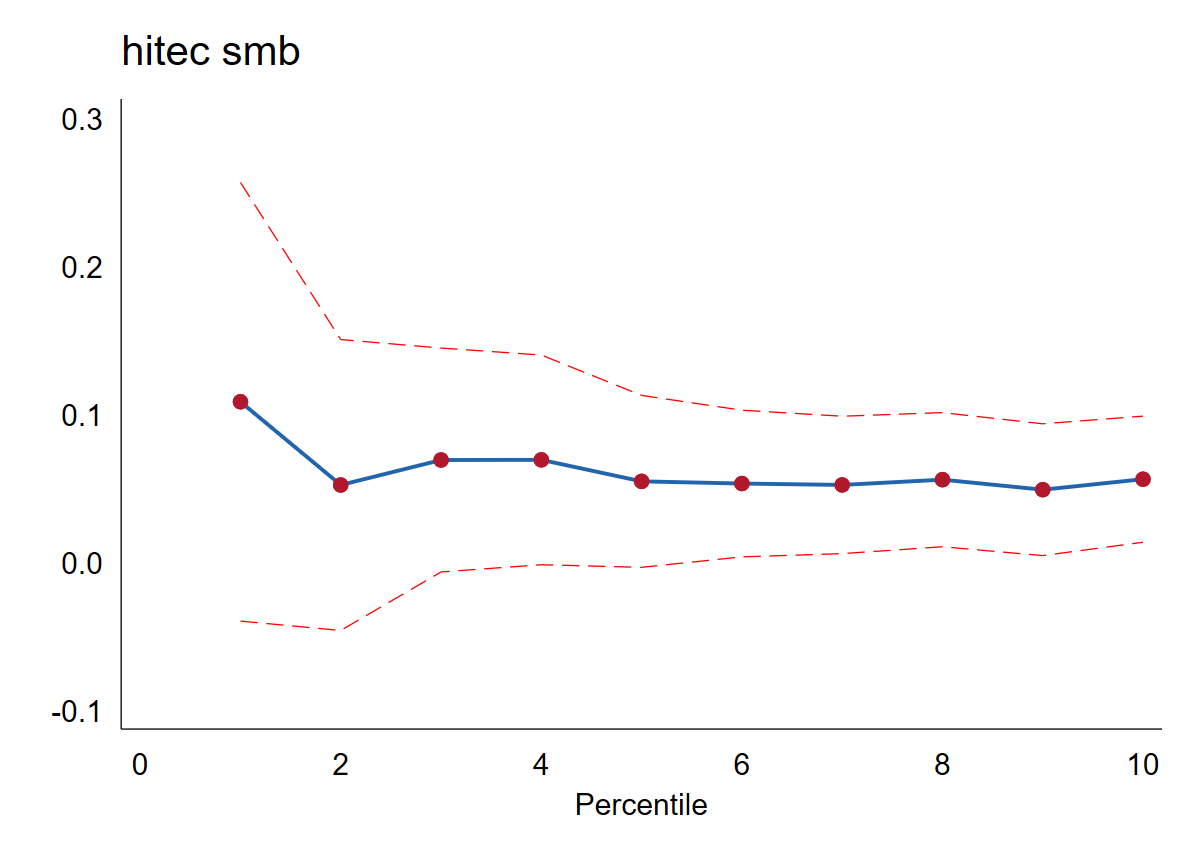}}
\par\end{centering}
\begin{centering}
\subfloat{\includegraphics[width=7cm]{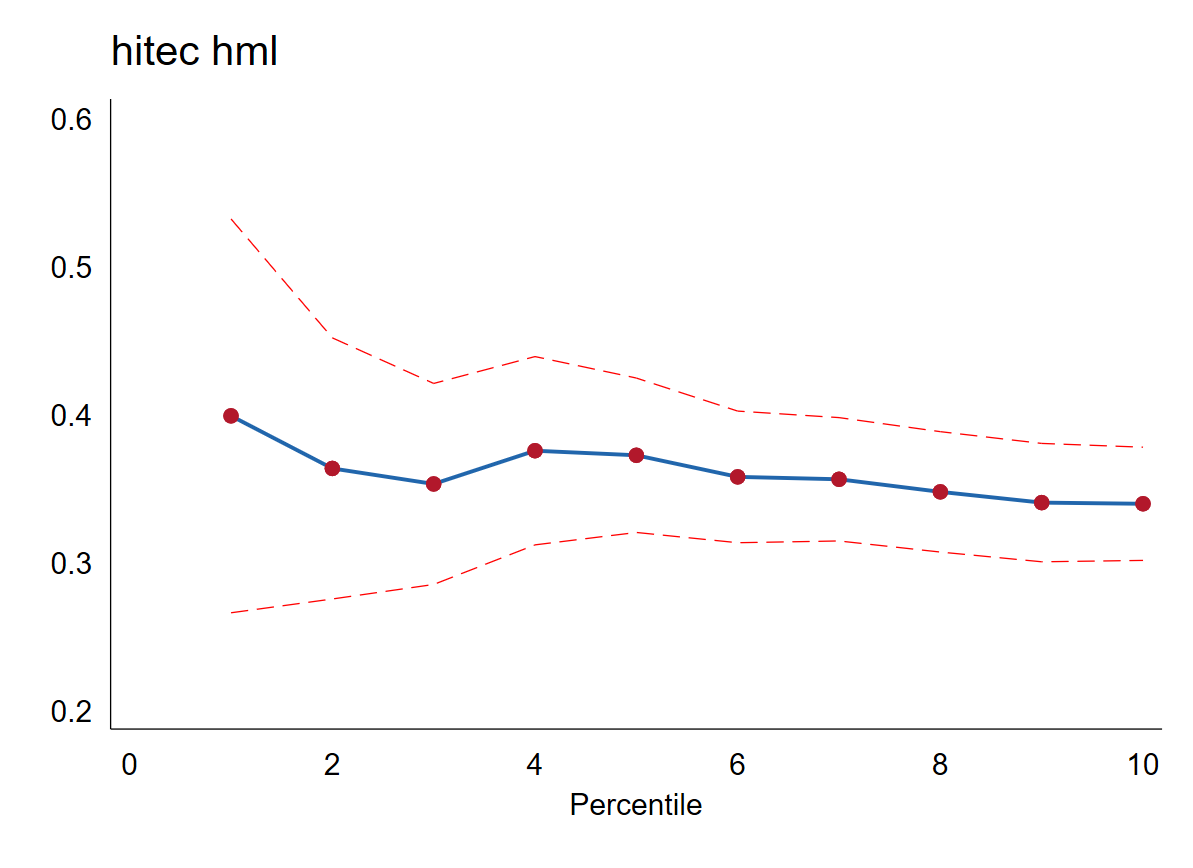}}\subfloat{\includegraphics[width=7cm]{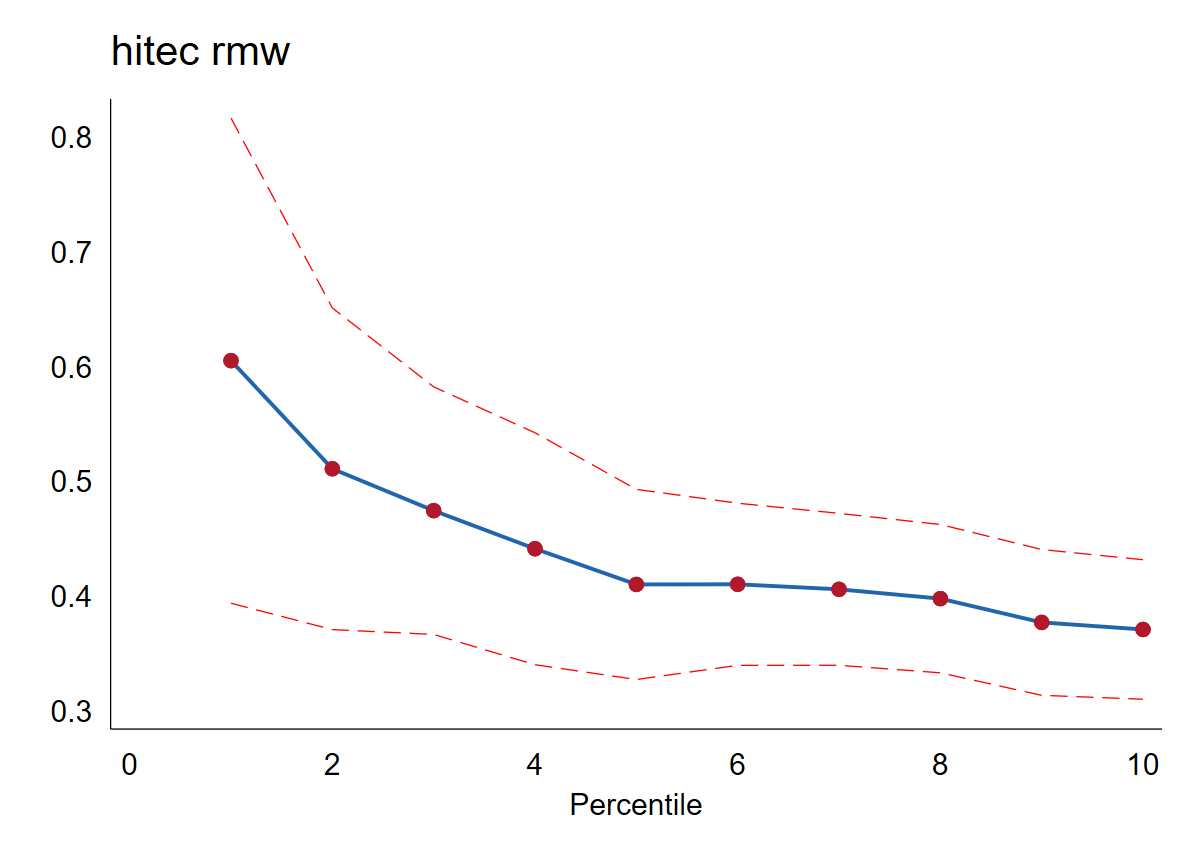}}
\par\end{centering}
\centering{}\subfloat{\includegraphics[width=7cm]{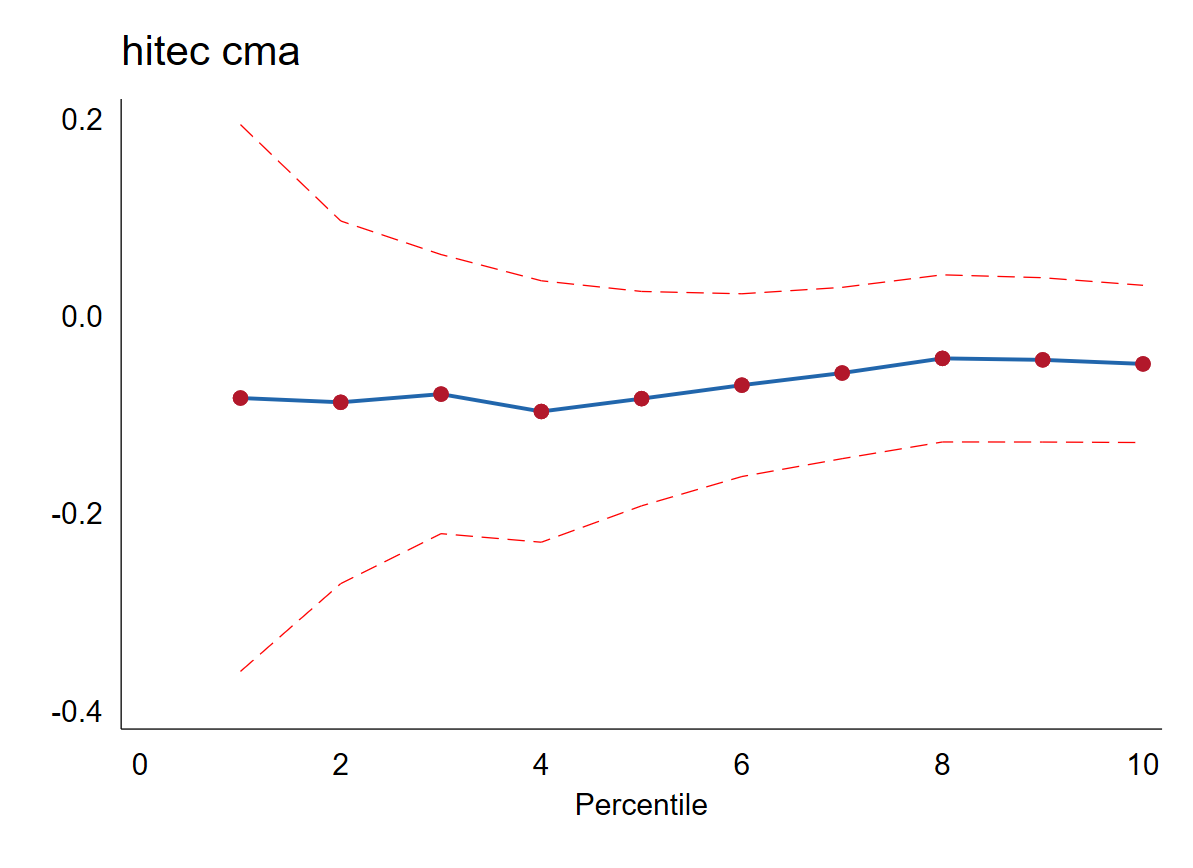}}\subfloat{\includegraphics[width=7cm]{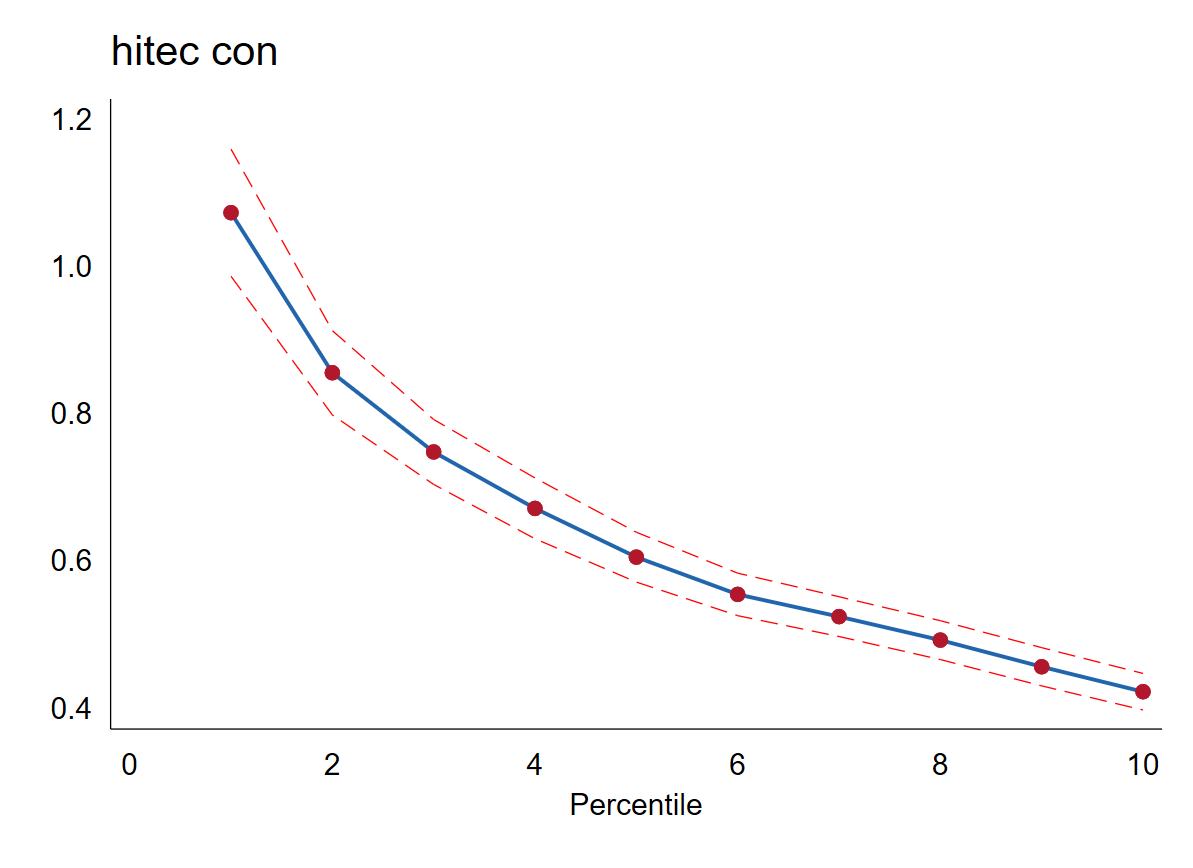}}
\end{figure}
\par\end{center}

\begin{center}
\begin{figure}[H]
\begin{centering}
\subfloat{\includegraphics[width=7cm]{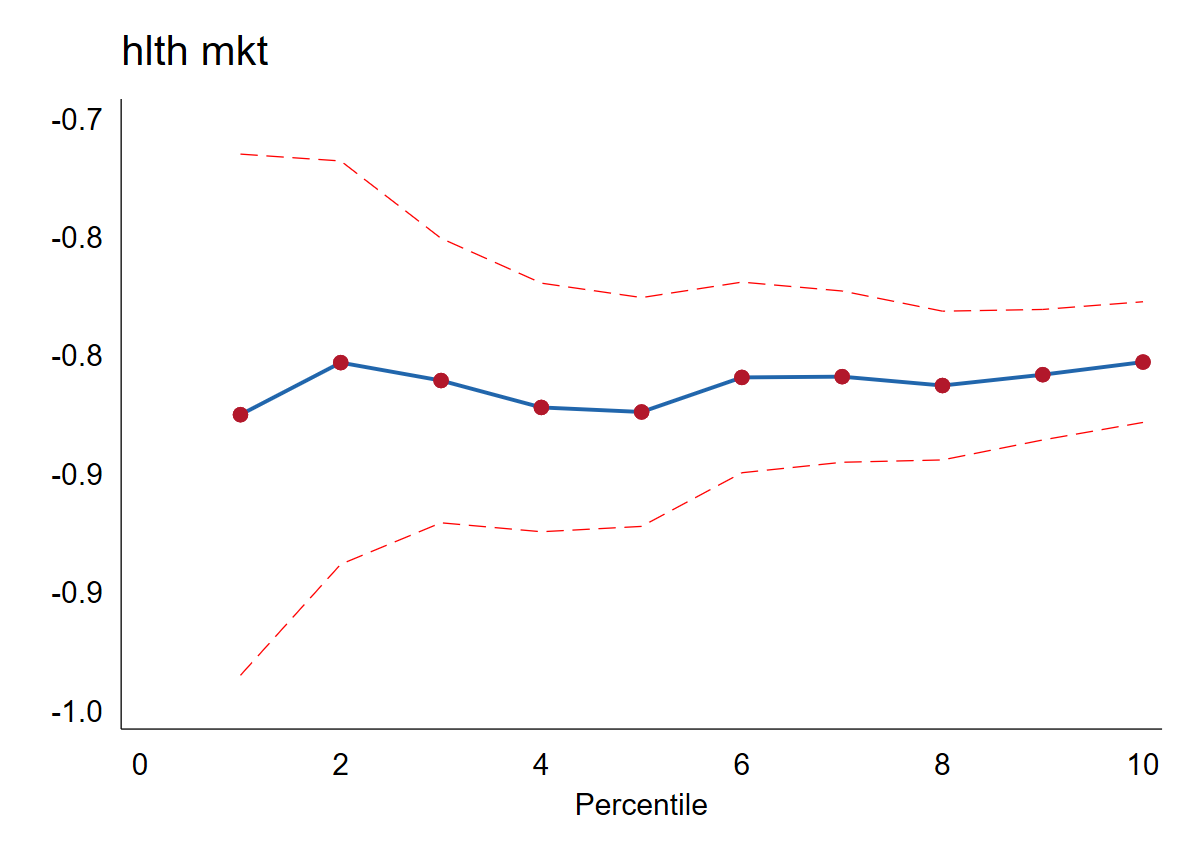}}\subfloat{\includegraphics[width=7cm]{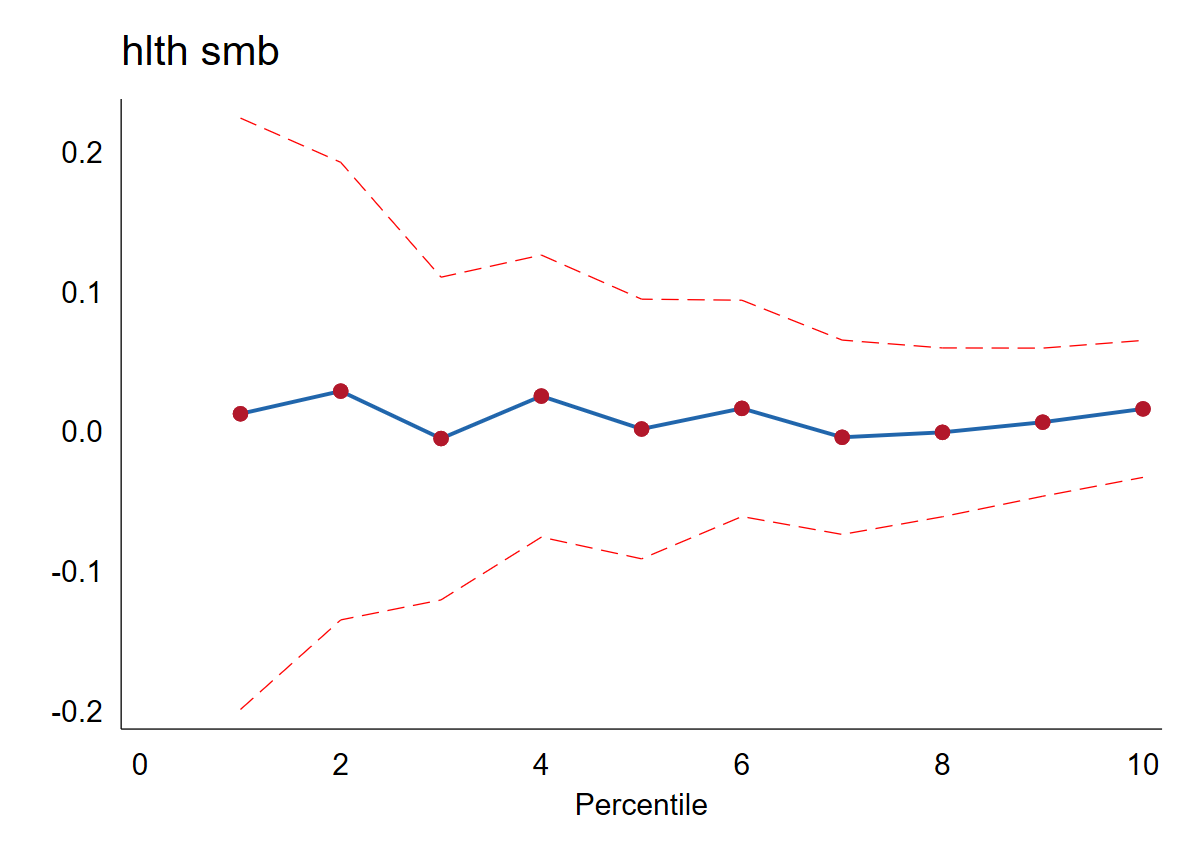}}
\par\end{centering}
\begin{centering}
\subfloat{\includegraphics[width=7cm]{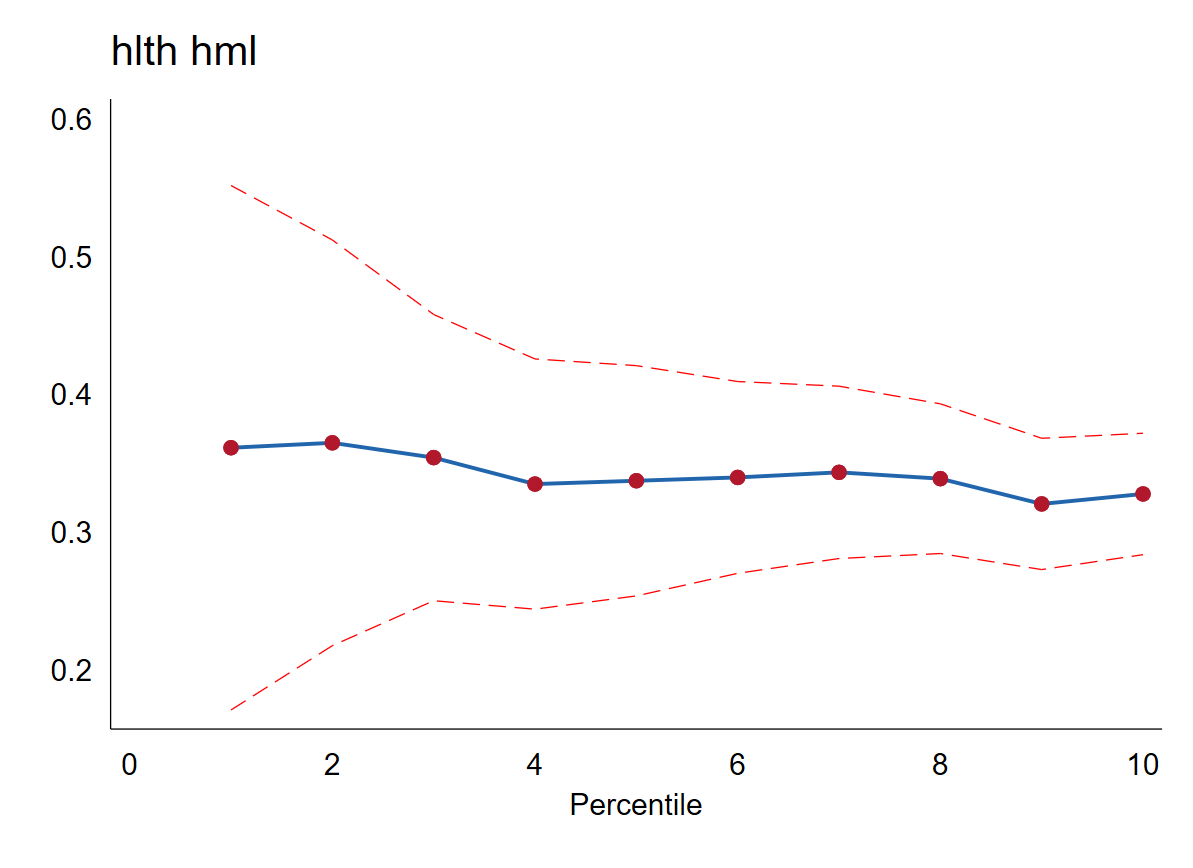}}\subfloat{\includegraphics[width=7cm]{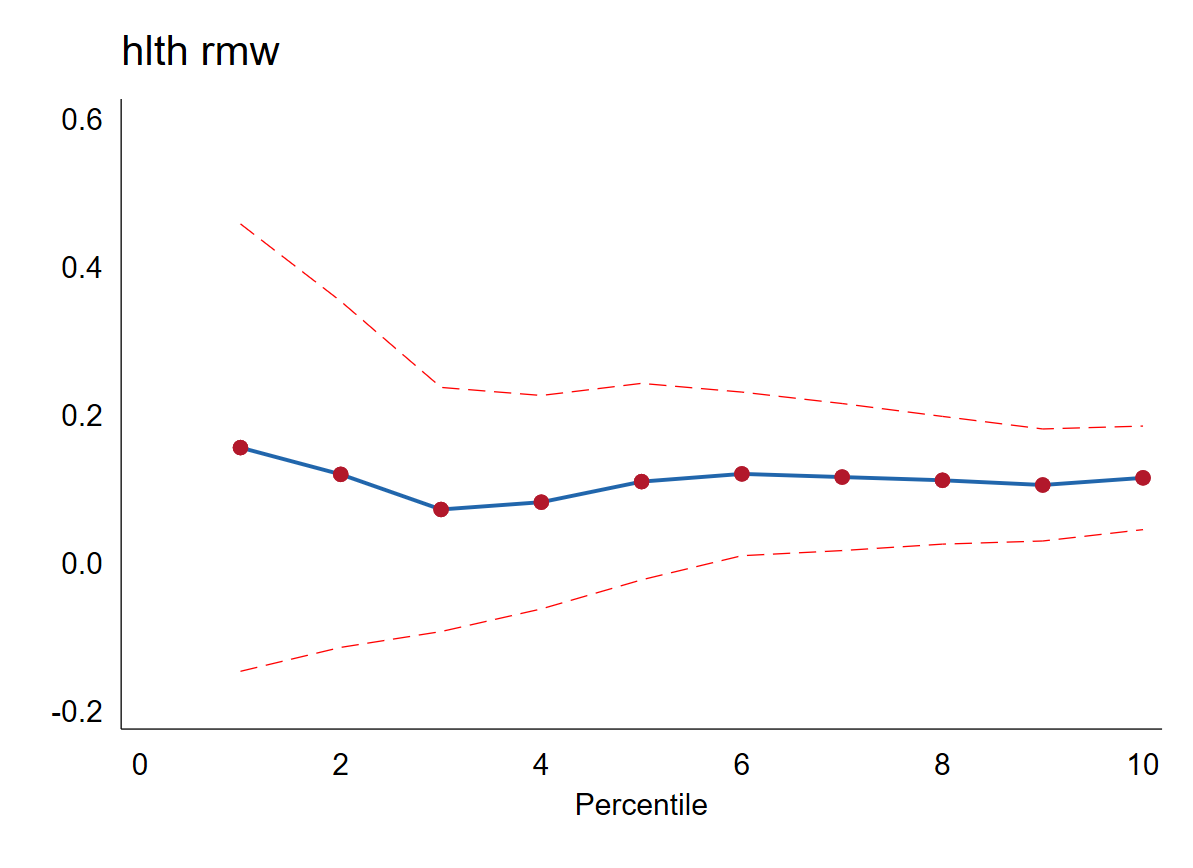}}
\par\end{centering}
\centering{}\subfloat{\includegraphics[width=7cm]{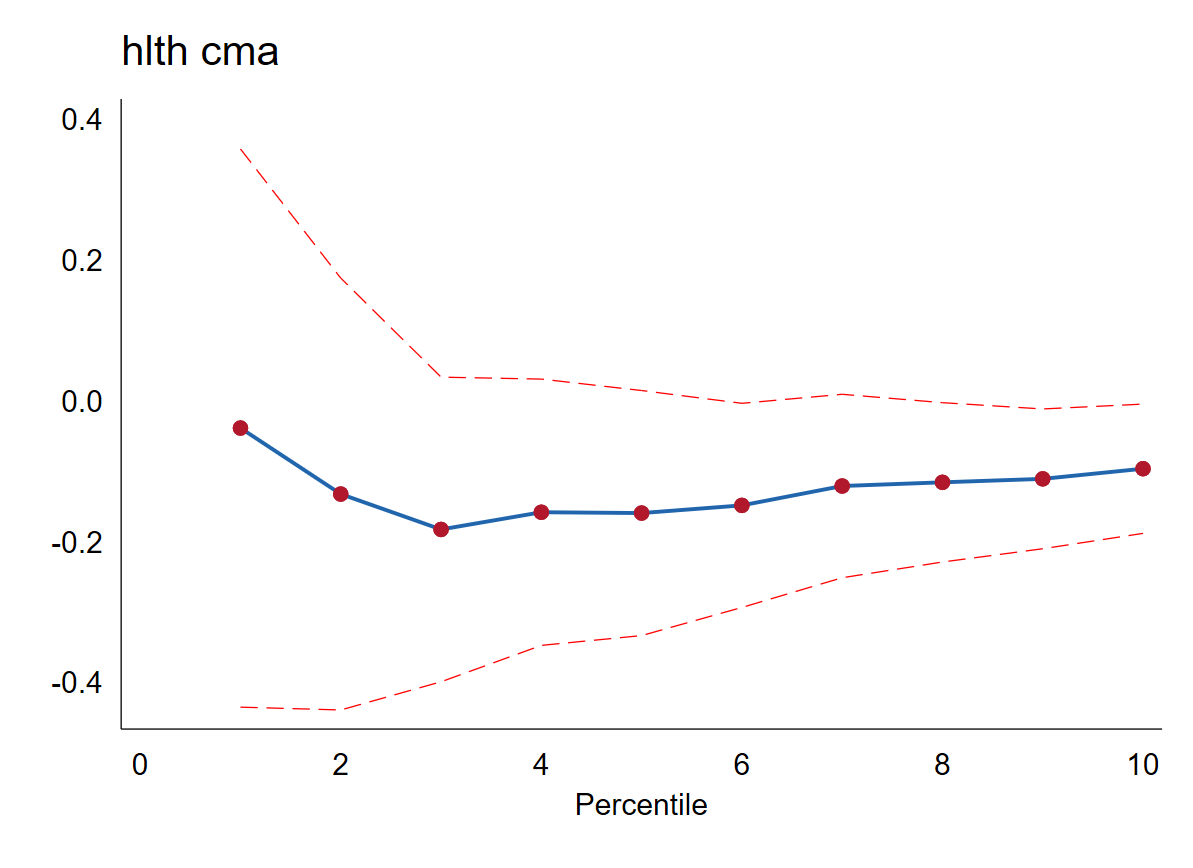}}\subfloat{\includegraphics[width=7cm]{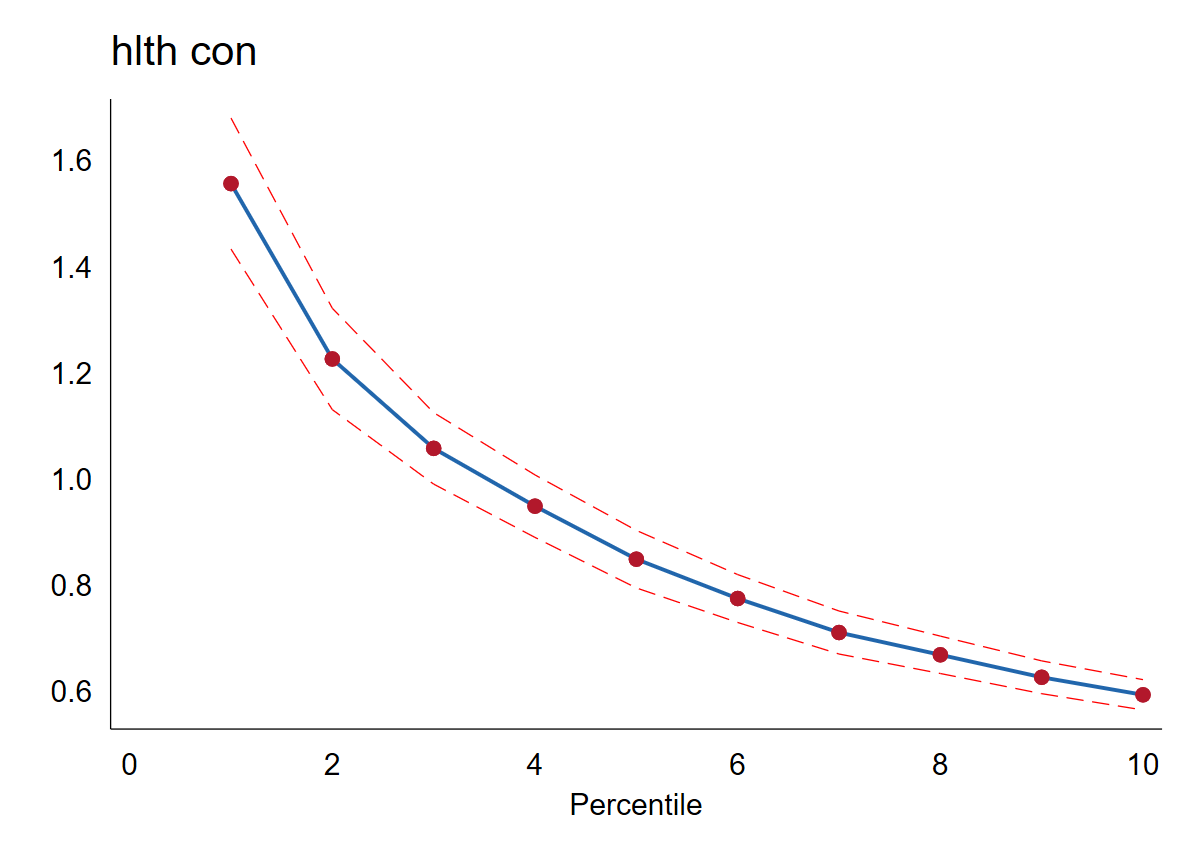}}
\end{figure}
\par\end{center}

\begin{center}
\begin{figure}[H]
\begin{centering}
\subfloat{\includegraphics[width=7cm]{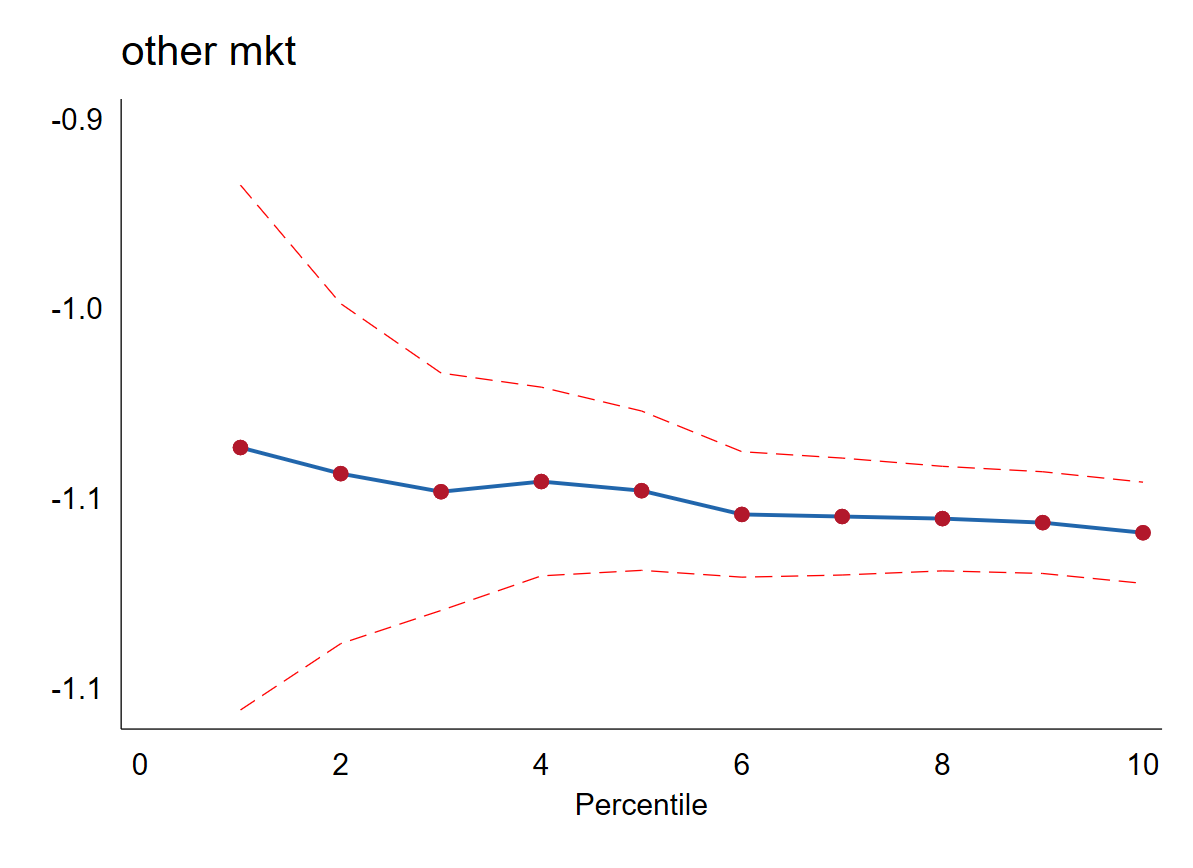}}\subfloat{\includegraphics[width=7cm]{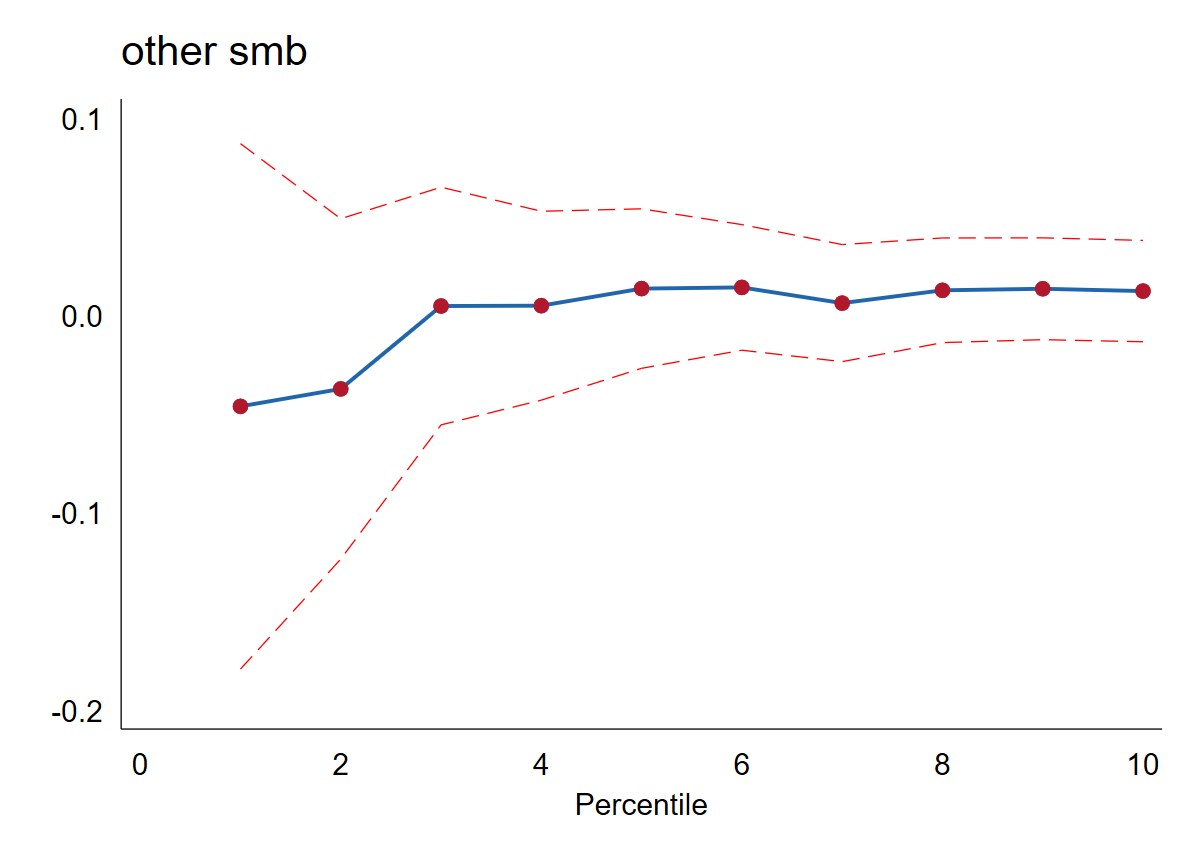}}
\par\end{centering}
\begin{centering}
\subfloat{\includegraphics[width=7cm]{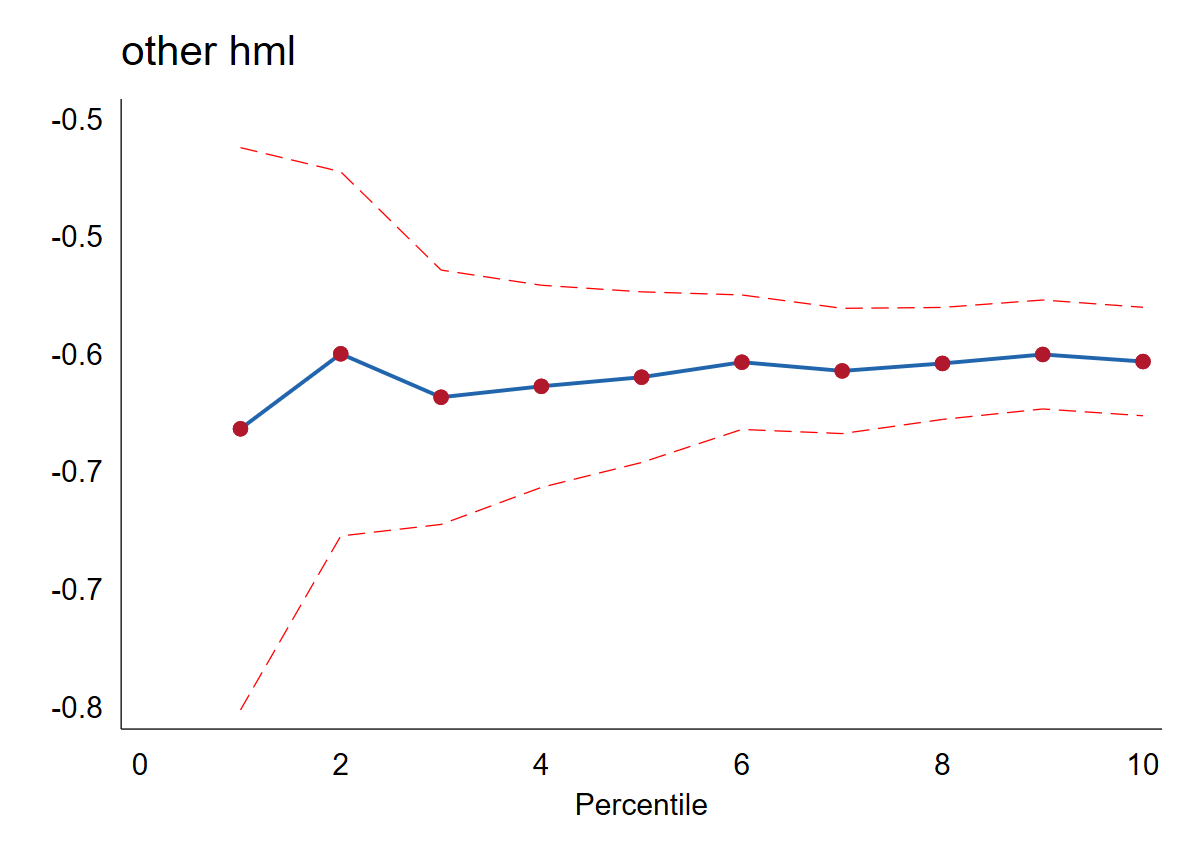}}\subfloat{\includegraphics[width=7cm]{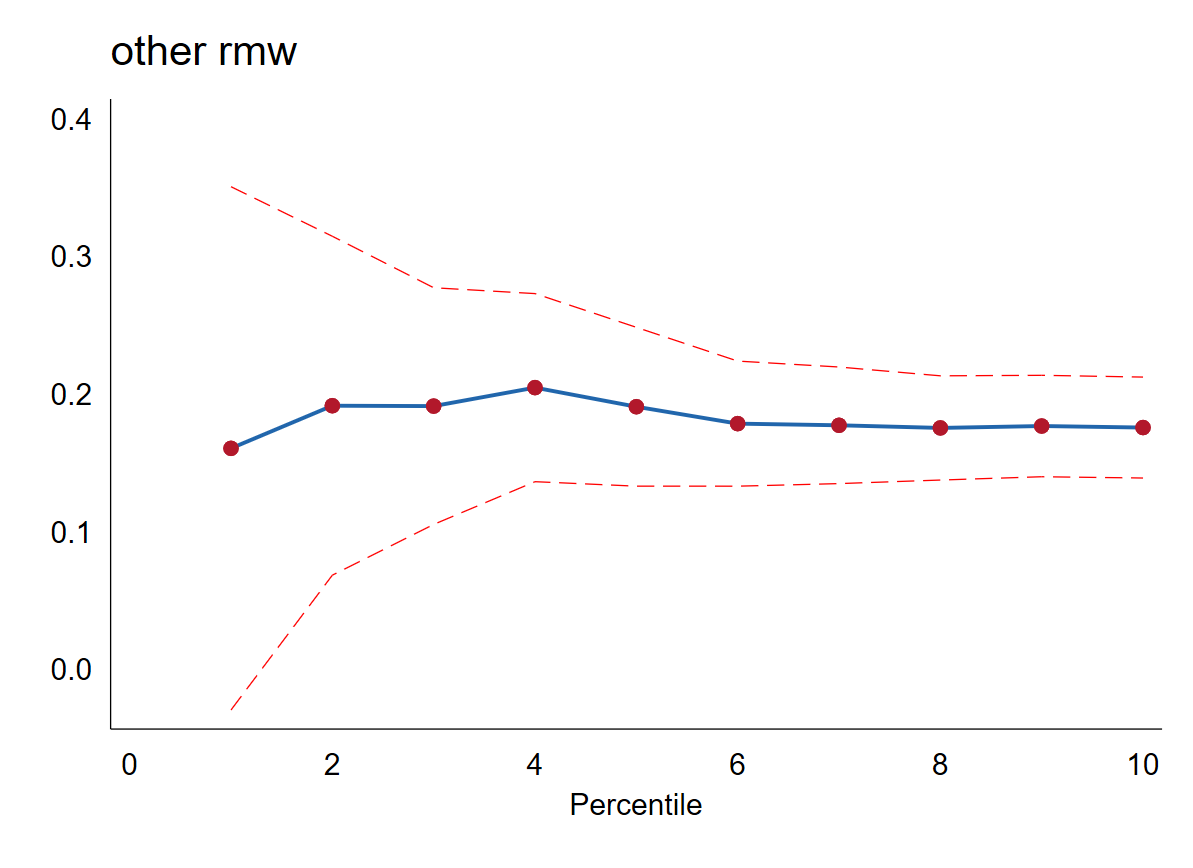}}
\par\end{centering}
\centering{}\subfloat{\includegraphics[width=7cm]{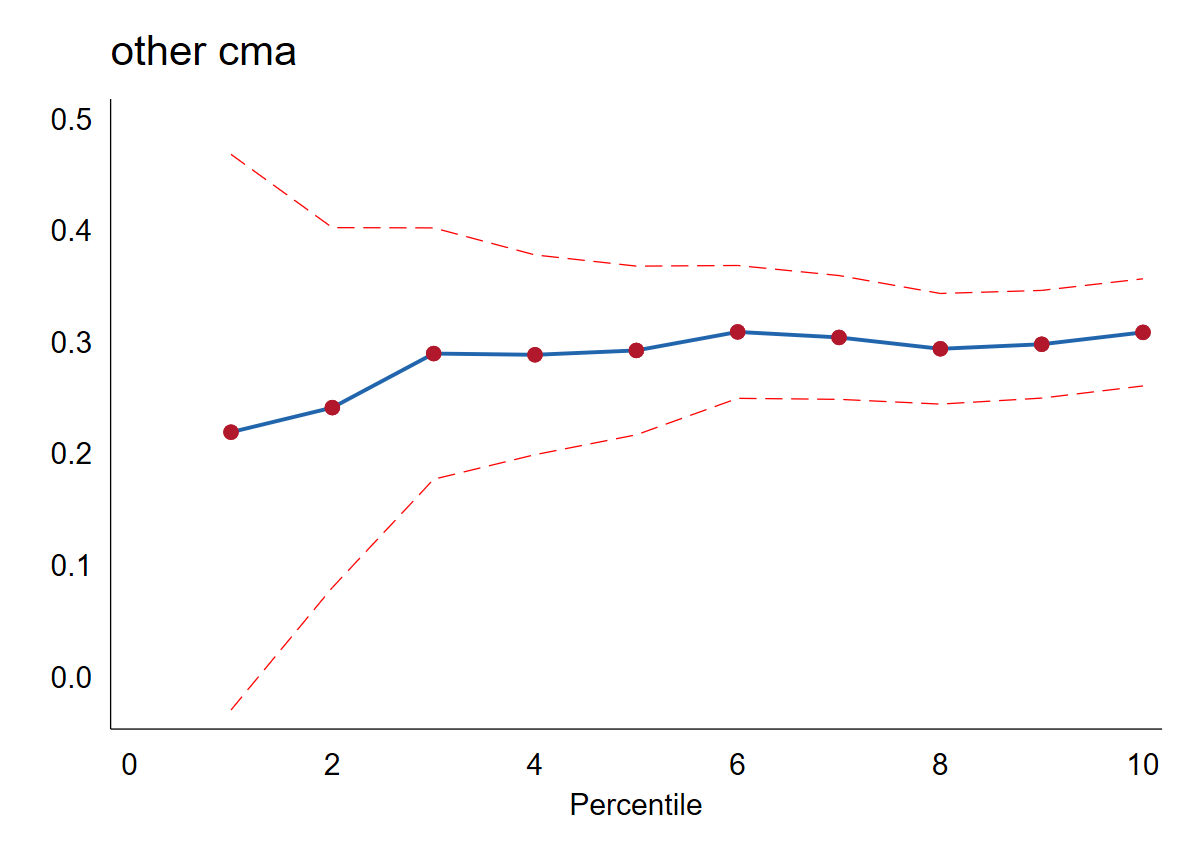}}\subfloat{\includegraphics[width=7cm]{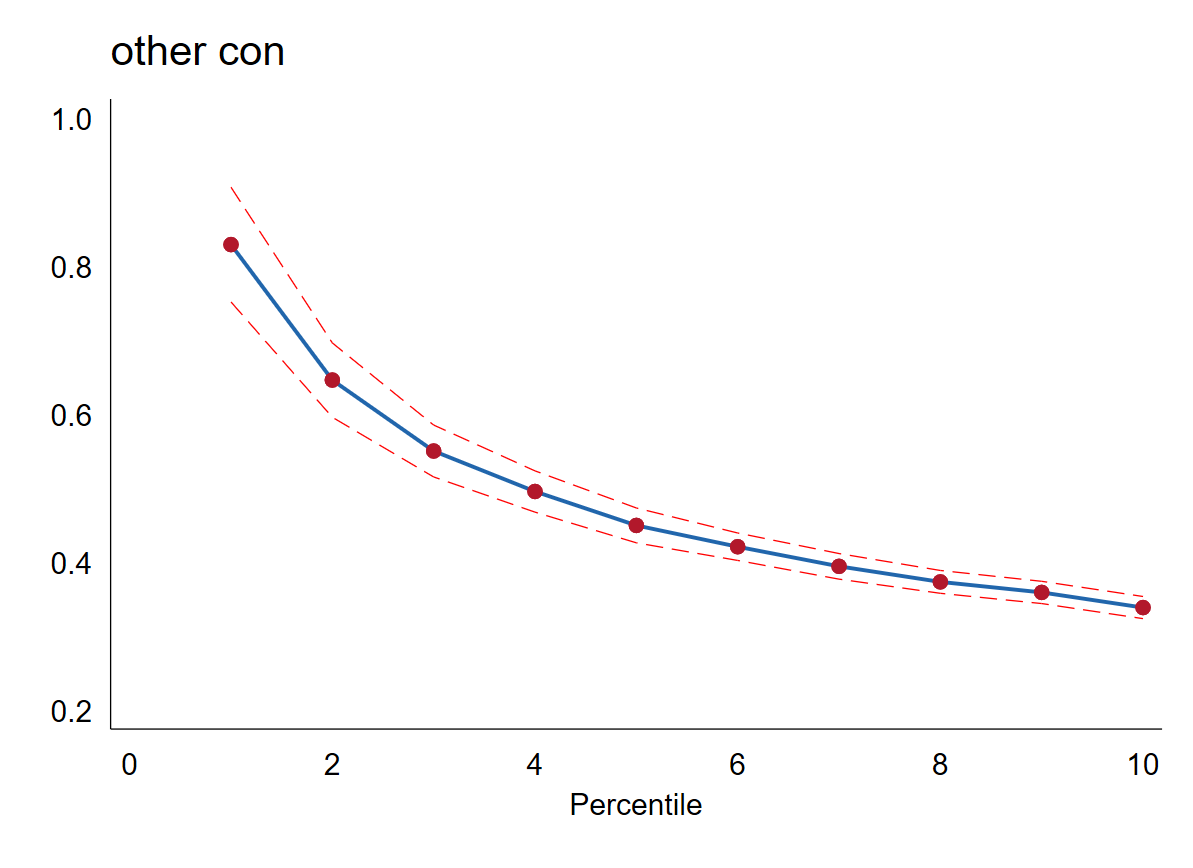}}
\end{figure}
\par\end{center}

\begin{landscape}
\begin{center}
{\tiny{}}
\begin{table}[H]
{\tiny{}\caption{\small{}Industry Factor Loadings -- Newey-West Adjusted\label{tab:newey-west}}
}{\tiny\par}

\caption*{
This table shows the WAQR estimates of industry factor loadings for the Fama-French
5 industries for the 10\% ES regression. The standard errors are reported for both the normal
and Newey-West adjusted method.  Standard errors are reported in parentheses.
}

\centering{}{\tiny{}}%
\begin{tabular}{lccccccccccccl}
\hline 
 & {\tiny{}MKTRF} &  & {\tiny{}SMB} &  & {\tiny{}HML} &  & {\tiny{}RMW} &  & {\tiny{}CMA} &  & {\tiny{}Constant} &  & {\tiny{}Adjustment}\tabularnewline
\hline 
{\tiny{}Cnsmr} & {\tiny{}-1.196} & {\tiny{}(0.022)} & {\tiny{}-0.122} & {\tiny{}(0.043)} & {\tiny{}0.214} & {\tiny{}(0.039)} & {\tiny{}-0.779} & {\tiny{}(0.061)} & {\tiny{}-0.250} & {\tiny{}(0.081)} & {\tiny{}0.909} & {\tiny{}(0.025)} & {\tiny{}No Adjust}\tabularnewline
 & {\tiny{}-1.196} & {\tiny{}(0.050)} & {\tiny{}-0.122} & {\tiny{}(0.093)} & {\tiny{}0.214} & {\tiny{}(0.112)} & {\tiny{}-0.779} & {\tiny{}(0.134)} & {\tiny{}-0.250} & {\tiny{}(0.185)} & {\tiny{}0.909} & {\tiny{}(0.029)} & {\tiny{}Newey-West}\tabularnewline
{\tiny{}Manuf} & {\tiny{}-1.293} & {\tiny{}(0.029)} & {\tiny{}-0.066} & {\tiny{}(0.057)} & {\tiny{}-0.514} & {\tiny{}(0.051)} & {\tiny{}-0.832} & {\tiny{}(0.081)} & {\tiny{}0.122} & {\tiny{}(0.106)} & {\tiny{}1.188} & {\tiny{}(0.033)} & {\tiny{}No Adjust}\tabularnewline
 & {\tiny{}-1.293} & {\tiny{}(0.058)} & {\tiny{}-0.066} & {\tiny{}(0.125)} & {\tiny{}-0.514} & {\tiny{}(0.131)} & {\tiny{}-0.832} & {\tiny{}(0.122)} & {\tiny{}0.122} & {\tiny{}(0.205)} & {\tiny{}1.188} & {\tiny{}(0.037)} & {\tiny{}Newey-West}\tabularnewline
{\tiny{}Hitec} & {\tiny{}-1.400} & {\tiny{}(0.027)} & {\tiny{}0.119} & {\tiny{}(0.052)} & {\tiny{}0.297} & {\tiny{}(0.046)} & {\tiny{}1.125} & {\tiny{}(0.074)} & {\tiny{}-0.280} & {\tiny{}(0.097)} & {\tiny{}1.082} & {\tiny{}(0.030)} & {\tiny{}No Adjust}\tabularnewline
 & {\tiny{}-1.400} & {\tiny{}(0.058)} & {\tiny{}0.119} & {\tiny{}(0.099)} & {\tiny{}0.297} & {\tiny{}(0.117)} & {\tiny{}1.125} & {\tiny{}(0.199)} & {\tiny{}-0.280} & {\tiny{}(0.210)} & {\tiny{}1.082} & {\tiny{}(0.033)} & {\tiny{}Newey-West}\tabularnewline
{\tiny{}Hlth} & {\tiny{}-1.153} & {\tiny{}(0.026)} & {\tiny{}0.096} & {\tiny{}(0.051)} & {\tiny{}0.251} & {\tiny{}(0.045)} & {\tiny{}0.190} & {\tiny{}(0.072)} & {\tiny{}-0.045} & {\tiny{}(0.094)} & {\tiny{}1.207} & {\tiny{}(0.029)} & {\tiny{}No Adjust}\tabularnewline
 & {\tiny{}-1.153} & {\tiny{}(0.056)} & {\tiny{}0.096} & {\tiny{}(0.087)} & {\tiny{}0.251} & {\tiny{}(0.095)} & {\tiny{}0.190} & {\tiny{}(0.105)} & {\tiny{}-0.045} & {\tiny{}(0.168)} & {\tiny{}1.207} & {\tiny{}(0.033)} & {\tiny{}Newey-West}\tabularnewline
{\tiny{}Other} & {\tiny{}-1.330} & {\tiny{}(0.033)} & {\tiny{}0.349} & {\tiny{}(0.064)} & {\tiny{}-1.685} & {\tiny{}(0.058)} & {\tiny{}0.278} & {\tiny{}(0.092)} & {\tiny{}0.765} & {\tiny{}(0.120)} & {\tiny{}1.043} & {\tiny{}(0.037)} & {\tiny{}No Adjust}\tabularnewline
 & {\tiny{}-1.330} & {\tiny{}(0.068)} & {\tiny{}0.349} & {\tiny{}(0.129)} & {\tiny{}-1.685} & {\tiny{}(0.197)} & {\tiny{}0.278} & {\tiny{}(0.144)} & {\tiny{}0.765} & {\tiny{}(0.224)} & {\tiny{}1.043} & {\tiny{}(0.041)} & {\tiny{}Newey-West}\tabularnewline
\hline 
\end{tabular}{\tiny\par}
\end{table}
{\tiny\par}
\par\end{center}
\end{landscape}

\begin{center}
{\tiny{}}
\begin{table}[H]
{\tiny{}\caption{\small{}Quantile Regression -- Higher Order\label{tab:Quantile-Regression-higher-order}}
}{\tiny\par}

\caption*{
This table shows the quantile regression results of regressing the
industry returns to the higher moments and interaction terms of the
Fama-French 5 factors.  Standard errors are reported in parentheses.
}

\centering{}{\scriptsize{}}%
\begin{tabular}{lcccccccccc}
\hline 
{\scriptsize{}Panel A} & {\scriptsize{}(1)} & {\scriptsize{}(2)} & {\scriptsize{}(3)} & {\scriptsize{}(4)} & {\scriptsize{}(5)} & {\scriptsize{}(6)} & {\scriptsize{}(7)} & {\scriptsize{}(8)} & {\scriptsize{}(9)} & {\scriptsize{}(10)}\tabularnewline
{\scriptsize{}Ind} & {\scriptsize{}cnsmr} & {\scriptsize{}cnsmr} & {\scriptsize{}manuf} & {\scriptsize{}manuf} & {\scriptsize{}hitec} & {\scriptsize{}hitec} & {\scriptsize{}hlth} & {\scriptsize{}hlth} & {\scriptsize{}other} & {\scriptsize{}other}\tabularnewline
{\scriptsize{}Quantile} & {\scriptsize{}0.1} & {\scriptsize{}0.05} & {\scriptsize{}0.1} & {\scriptsize{}0.05} & {\scriptsize{}0.1} & {\scriptsize{}0.05} & {\scriptsize{}0.1} & {\scriptsize{}0.05} & {\scriptsize{}0.1} & {\scriptsize{}0.05}\tabularnewline
\hline 
{\scriptsize{}$mktrf$} & {\scriptsize{}0.88} & {\scriptsize{}0.89} & {\scriptsize{}0.99} & {\scriptsize{}0.99} & {\scriptsize{}1.08} & {\scriptsize{}1.10} & {\scriptsize{}0.82} & {\scriptsize{}0.83} & {\scriptsize{}1.07} & {\scriptsize{}1.07}\tabularnewline
 & {\scriptsize{}(0.01)} & {\scriptsize{}(0.01)} & {\scriptsize{}(0.01)} & {\scriptsize{}(0.02)} & {\scriptsize{}(0.01)} & {\scriptsize{}(0.01)} & {\scriptsize{}(0.01)} & {\scriptsize{}(0.02)} & {\scriptsize{}(0.01)} & {\scriptsize{}(0.01)}\tabularnewline
{\scriptsize{}$smb$} & {\scriptsize{}0.04} & {\scriptsize{}0.08} & {\scriptsize{}0.14} & {\scriptsize{}0.13} & {\scriptsize{}-0.06} & {\scriptsize{}-0.10} & {\scriptsize{}-0.06} & {\scriptsize{}-0.09} & {\scriptsize{}-0.03} & {\scriptsize{}-0.03}\tabularnewline
 & {\scriptsize{}(0.02)} & {\scriptsize{}(0.02)} & {\scriptsize{}(0.02)} & {\scriptsize{}(0.03)} & {\scriptsize{}(0.02)} & {\scriptsize{}(0.02)} & {\scriptsize{}(0.03)} & {\scriptsize{}(0.04)} & {\scriptsize{}(0.01)} & {\scriptsize{}(0.02)}\tabularnewline
{\scriptsize{}$hml$} & {\scriptsize{}-0.17} & {\scriptsize{}-0.14} & {\scriptsize{}0.22} & {\scriptsize{}0.19} & {\scriptsize{}-0.34} & {\scriptsize{}-0.33} & {\scriptsize{}-0.43} & {\scriptsize{}-0.49} & {\scriptsize{}0.54} & {\scriptsize{}0.56}\tabularnewline
 & {\scriptsize{}(0.02)} & {\scriptsize{}(0.02)} & {\scriptsize{}(0.02)} & {\scriptsize{}(0.03)} & {\scriptsize{}(0.02)} & {\scriptsize{}(0.02)} & {\scriptsize{}(0.03)} & {\scriptsize{}(0.04)} & {\scriptsize{}(0.01)} & {\scriptsize{}(0.02)}\tabularnewline
{\scriptsize{}$rmw$} & {\scriptsize{}0.31} & {\scriptsize{}0.26} & {\scriptsize{}0.39} & {\scriptsize{}0.44} & {\scriptsize{}-0.33} & {\scriptsize{}-0.34} & {\scriptsize{}-0.31} & {\scriptsize{}-0.41} & {\scriptsize{}-0.20} & {\scriptsize{}-0.24}\tabularnewline
 & {\scriptsize{}(0.03)} & {\scriptsize{}(0.04)} & {\scriptsize{}(0.04)} & {\scriptsize{}(0.05)} & {\scriptsize{}(0.03)} & {\scriptsize{}(0.04)} & {\scriptsize{}(0.04)} & {\scriptsize{}(0.07)} & {\scriptsize{}(0.02)} & {\scriptsize{}(0.03)}\tabularnewline
{\scriptsize{}$cma$} & {\scriptsize{}0.26} & {\scriptsize{}0.20} & {\scriptsize{}-0.05} & {\scriptsize{}-0.09} & {\scriptsize{}0.09} & {\scriptsize{}0.09} & {\scriptsize{}0.08} & {\scriptsize{}0.12} & {\scriptsize{}-0.19} & {\scriptsize{}-0.19}\tabularnewline
 & {\scriptsize{}(0.03)} & {\scriptsize{}(0.04)} & {\scriptsize{}(0.05)} & {\scriptsize{}(0.07)} & {\scriptsize{}(0.03)} & {\scriptsize{}(0.05)} & {\scriptsize{}(0.05)} & {\scriptsize{}(0.08)} & {\scriptsize{}(0.03)} & {\scriptsize{}(0.04)}\tabularnewline
{\scriptsize{}$mktrf^{2}$} & {\scriptsize{}-0.01} & {\scriptsize{}-0.02} & {\scriptsize{}-0.02} & {\scriptsize{}-0.02} & {\scriptsize{}-0.00} & {\scriptsize{}-0.00} & {\scriptsize{}-0.02} & {\scriptsize{}-0.01} & {\scriptsize{}-0.01} & {\scriptsize{}-0.01}\tabularnewline
 & {\scriptsize{}(0.00)} & {\scriptsize{}(0.00)} & {\scriptsize{}(0.00)} & {\scriptsize{}(0.00)} & {\scriptsize{}(0.00)} & {\scriptsize{}(0.00)} & {\scriptsize{}(0.00)} & {\scriptsize{}(0.00)} & {\scriptsize{}(0.00)} & {\scriptsize{}(0.00)}\tabularnewline
{\scriptsize{}$smb^{2}$} & {\scriptsize{}-0.05} & {\scriptsize{}-0.06} & {\scriptsize{}-0.06} & {\scriptsize{}-0.02} & {\scriptsize{}-0.00} & {\scriptsize{}-0.02} & {\scriptsize{}-0.10} & {\scriptsize{}-0.14} & {\scriptsize{}-0.02} & {\scriptsize{}-0.02}\tabularnewline
 & {\scriptsize{}(0.01)} & {\scriptsize{}(0.01)} & {\scriptsize{}(0.01)} & {\scriptsize{}(0.02)} & {\scriptsize{}(0.01)} & {\scriptsize{}(0.01)} & {\scriptsize{}(0.02)} & {\scriptsize{}(0.03)} & {\scriptsize{}(0.01)} & {\scriptsize{}(0.01)}\tabularnewline
{\scriptsize{}$hml^{2}$} & {\scriptsize{}-0.00} & {\scriptsize{}0.01} & {\scriptsize{}-0.06} & {\scriptsize{}-0.10} & {\scriptsize{}-0.02} & {\scriptsize{}-0.02} & {\scriptsize{}-0.05} & {\scriptsize{}-0.08} & {\scriptsize{}-0.02} & {\scriptsize{}-0.03}\tabularnewline
 & {\scriptsize{}(0.01)} & {\scriptsize{}(0.01)} & {\scriptsize{}(0.01)} & {\scriptsize{}(0.01)} & {\scriptsize{}(0.01)} & {\scriptsize{}(0.01)} & {\scriptsize{}(0.01)} & {\scriptsize{}(0.01)} & {\scriptsize{}(0.00)} & {\scriptsize{}(0.01)}\tabularnewline
{\scriptsize{}$rmw^{2}$} & {\scriptsize{}-0.15} & {\scriptsize{}-0.24} & {\scriptsize{}-0.03} & {\scriptsize{}-0.15} & {\scriptsize{}-0.30} & {\scriptsize{}-0.31} & {\scriptsize{}-0.17} & {\scriptsize{}-0.24} & {\scriptsize{}-0.14} & {\scriptsize{}-0.22}\tabularnewline
 & {\scriptsize{}(0.02)} & {\scriptsize{}(0.03)} & {\scriptsize{}(0.03)} & {\scriptsize{}(0.05)} & {\scriptsize{}(0.02)} & {\scriptsize{}(0.03)} & {\scriptsize{}(0.04)} & {\scriptsize{}(0.06)} & {\scriptsize{}(0.02)} & {\scriptsize{}(0.03)}\tabularnewline
{\scriptsize{}$cma^{2}$} & {\scriptsize{}-0.17} & {\scriptsize{}-0.20} & {\scriptsize{}-0.41} & {\scriptsize{}-0.54} & {\scriptsize{}-0.45} & {\scriptsize{}-0.57} & {\scriptsize{}-0.26} & {\scriptsize{}-0.42} & {\scriptsize{}-0.13} & {\scriptsize{}-0.09}\tabularnewline
 & {\scriptsize{}(0.04)} & {\scriptsize{}(0.05)} & {\scriptsize{}(0.05)} & {\scriptsize{}(0.07)} & {\scriptsize{}(0.04)} & {\scriptsize{}(0.05)} & {\scriptsize{}(0.05)} & {\scriptsize{}(0.09)} & {\scriptsize{}(0.03)} & {\scriptsize{}(0.04)}\tabularnewline
{\scriptsize{}$mktrf^{3}$} & {\scriptsize{}-0.00} & {\scriptsize{}-0.00} & {\scriptsize{}0.00} & {\scriptsize{}0.00} & {\scriptsize{}-0.00} & {\scriptsize{}-0.00} & {\scriptsize{}-0.00} & {\scriptsize{}-0.00} & {\scriptsize{}-0.00} & {\scriptsize{}-0.00}\tabularnewline
 & {\scriptsize{}(0.00)} & {\scriptsize{}(0.00)} & {\scriptsize{}(0.00)} & {\scriptsize{}(0.00)} & {\scriptsize{}(0.00)} & {\scriptsize{}(0.00)} & {\scriptsize{}(0.00)} & {\scriptsize{}(0.00)} & {\scriptsize{}(0.00)} & {\scriptsize{}(0.00)}\tabularnewline
{\scriptsize{}$smb^{3}$} & {\scriptsize{}0.01} & {\scriptsize{}-0.00} & {\scriptsize{}-0.02} & {\scriptsize{}-0.02} & {\scriptsize{}0.00} & {\scriptsize{}0.00} & {\scriptsize{}0.02} & {\scriptsize{}0.03} & {\scriptsize{}0.00} & {\scriptsize{}0.00}\tabularnewline
 & {\scriptsize{}(0.00)} & {\scriptsize{}(0.00)} & {\scriptsize{}(0.00)} & {\scriptsize{}(0.01)} & {\scriptsize{}(0.00)} & {\scriptsize{}(0.00)} & {\scriptsize{}(0.00)} & {\scriptsize{}(0.01)} & {\scriptsize{}(0.00)} & {\scriptsize{}(0.00)}\tabularnewline
{\scriptsize{}$hml^{3}$} & {\scriptsize{}-0.00} & {\scriptsize{}-0.00} & {\scriptsize{}-0.02} & {\scriptsize{}-0.02} & {\scriptsize{}0.00} & {\scriptsize{}0.00} & {\scriptsize{}0.02} & {\scriptsize{}0.03} & {\scriptsize{}0.01} & {\scriptsize{}0.01}\tabularnewline
 & {\scriptsize{}(0.00)} & {\scriptsize{}(0.00)} & {\scriptsize{}(0.00)} & {\scriptsize{}(0.00)} & {\scriptsize{}(0.00)} & {\scriptsize{}(0.00)} & {\scriptsize{}(0.00)} & {\scriptsize{}(0.00)} & {\scriptsize{}(0.00)} & {\scriptsize{}(0.00)}\tabularnewline
{\scriptsize{}$rmw^{3}$} & {\scriptsize{}0.03} & {\scriptsize{}0.08} & {\scriptsize{}-0.01} & {\scriptsize{}-0.05} & {\scriptsize{}-0.06} & {\scriptsize{}-0.09} & {\scriptsize{}0.19} & {\scriptsize{}0.23} & {\scriptsize{}0.08} & {\scriptsize{}0.11}\tabularnewline
 & {\scriptsize{}(0.01)} & {\scriptsize{}(0.02)} & {\scriptsize{}(0.02)} & {\scriptsize{}(0.03)} & {\scriptsize{}(0.01)} & {\scriptsize{}(0.02)} & {\scriptsize{}(0.02)} & {\scriptsize{}(0.03)} & {\scriptsize{}(0.01)} & {\scriptsize{}(0.01)}\tabularnewline
{\scriptsize{}$cma^{3}$} & {\scriptsize{}-0.07} & {\scriptsize{}0.10} & {\scriptsize{}0.02} & {\scriptsize{}0.09} & {\scriptsize{}-0.20} & {\scriptsize{}-0.25} & {\scriptsize{}0.13} & {\scriptsize{}-0.05} & {\scriptsize{}-0.08} & {\scriptsize{}-0.03}\tabularnewline
 & {\scriptsize{}(0.02)} & {\scriptsize{}(0.03)} & {\scriptsize{}(0.03)} & {\scriptsize{}(0.04)} & {\scriptsize{}(0.02)} & {\scriptsize{}(0.03)} & {\scriptsize{}(0.03)} & {\scriptsize{}(0.06)} & {\scriptsize{}(0.02)} & {\scriptsize{}(0.02)}\tabularnewline
{\scriptsize{}Cons} & {\scriptsize{}-0.30} & {\scriptsize{}-0.40} & {\scriptsize{}-0.40} & {\scriptsize{}-0.55} & {\scriptsize{}-0.30} & {\scriptsize{}-0.42} & {\scriptsize{}-0.46} & {\scriptsize{}-0.62} & {\scriptsize{}-0.27} & {\scriptsize{}-0.36}\tabularnewline
 & {\scriptsize{}(0.01)} & {\scriptsize{}(0.01)} & {\scriptsize{}(0.01)} & {\scriptsize{}(0.02)} & {\scriptsize{}(0.01)} & {\scriptsize{}(0.01)} & {\scriptsize{}(0.02)} & {\scriptsize{}(0.03)} & {\scriptsize{}(0.01)} & {\scriptsize{}(0.01)}\tabularnewline
\hline 
\end{tabular}{\scriptsize\par}
\end{table}
{\tiny\par}
\par\end{center}

\begin{center}
{\tiny{}}
\begin{table}[H]
\centering{}{\scriptsize{}}%
\begin{tabular}{lcccccccccc}
\hline 
{\scriptsize{}Panel B} & {\scriptsize{}(1)} & {\scriptsize{}(2)} & {\scriptsize{}(3)} & {\scriptsize{}(4)} & {\scriptsize{}(5)} & {\scriptsize{}(6)} & {\scriptsize{}(7)} & {\scriptsize{}(8)} & {\scriptsize{}(9)} & {\scriptsize{}(10)}\tabularnewline
{\scriptsize{}Ind} & {\scriptsize{}cnsmr} & {\scriptsize{}cnsmr} & {\scriptsize{}manuf} & {\scriptsize{}manuf} & {\scriptsize{}hitec} & {\scriptsize{}hitec} & {\scriptsize{}hlth} & {\scriptsize{}hlth} & {\scriptsize{}other} & {\scriptsize{}other}\tabularnewline
{\scriptsize{}Quantile} & {\scriptsize{}0.1} & {\scriptsize{}0.05} & {\scriptsize{}0.1} & {\scriptsize{}0.05} & {\scriptsize{}0.1} & {\scriptsize{}0.05} & {\scriptsize{}0.1} & {\scriptsize{}0.05} & {\scriptsize{}0.1} & {\scriptsize{}0.05}\tabularnewline
\hline 
{\scriptsize{}$mktrf$} & {\scriptsize{}0.87} & {\scriptsize{}0.87} & {\scriptsize{}1.02} & {\scriptsize{}1.03} & {\scriptsize{}1.07} & {\scriptsize{}1.08} & {\scriptsize{}0.80} & {\scriptsize{}0.79} & {\scriptsize{}1.06} & {\scriptsize{}1.04}\tabularnewline
 & {\scriptsize{}(0.01)} & {\scriptsize{}(0.01)} & {\scriptsize{}(0.01)} & {\scriptsize{}(0.02)} & {\scriptsize{}(0.01)} & {\scriptsize{}(0.02)} & {\scriptsize{}(0.01)} & {\scriptsize{}(0.02)} & {\scriptsize{}(0.01)} & {\scriptsize{}(0.01)}\tabularnewline
{\scriptsize{}$smb$} & {\scriptsize{}0.06} & {\scriptsize{}0.09} & {\scriptsize{}0.10} & {\scriptsize{}0.07} & {\scriptsize{}-0.07} & {\scriptsize{}-0.08} & {\scriptsize{}-0.02} & {\scriptsize{}-0.05} & {\scriptsize{}0.00} & {\scriptsize{}-0.01}\tabularnewline
 & {\scriptsize{}(0.02)} & {\scriptsize{}(0.02)} & {\scriptsize{}(0.02)} & {\scriptsize{}(0.03)} & {\scriptsize{}(0.02)} & {\scriptsize{}(0.03)} & {\scriptsize{}(0.02)} & {\scriptsize{}(0.04)} & {\scriptsize{}(0.01)} & {\scriptsize{}(0.02)}\tabularnewline
{\scriptsize{}$hml$} & {\scriptsize{}-0.19} & {\scriptsize{}-0.18} & {\scriptsize{}0.14} & {\scriptsize{}0.11} & {\scriptsize{}-0.33} & {\scriptsize{}-0.32} & {\scriptsize{}-0.32} & {\scriptsize{}-0.32} & {\scriptsize{}0.60} & {\scriptsize{}0.62}\tabularnewline
 & {\scriptsize{}(0.01)} & {\scriptsize{}(0.02)} & {\scriptsize{}(0.02)} & {\scriptsize{}(0.03)} & {\scriptsize{}(0.02)} & {\scriptsize{}(0.03)} & {\scriptsize{}(0.02)} & {\scriptsize{}(0.04)} & {\scriptsize{}(0.01)} & {\scriptsize{}(0.02)}\tabularnewline
{\scriptsize{}$rmw$} & {\scriptsize{}0.30} & {\scriptsize{}0.30} & {\scriptsize{}0.41} & {\scriptsize{}0.45} & {\scriptsize{}-0.40} & {\scriptsize{}-0.47} & {\scriptsize{}-0.14} & {\scriptsize{}-0.19} & {\scriptsize{}-0.16} & {\scriptsize{}-0.17}\tabularnewline
 & {\scriptsize{}(0.02)} & {\scriptsize{}(0.03)} & {\scriptsize{}(0.04)} & {\scriptsize{}(0.05)} & {\scriptsize{}(0.03)} & {\scriptsize{}(0.04)} & {\scriptsize{}(0.03)} & {\scriptsize{}(0.06)} & {\scriptsize{}(0.02)} & {\scriptsize{}(0.03)}\tabularnewline
{\scriptsize{}$cma$} & {\scriptsize{}0.21} & {\scriptsize{}0.19} & {\scriptsize{}0.03} & {\scriptsize{}-0.02} & {\scriptsize{}0.07} & {\scriptsize{}0.05} & {\scriptsize{}0.07} & {\scriptsize{}0.09} & {\scriptsize{}-0.30} & {\scriptsize{}-0.30}\tabularnewline
 & {\scriptsize{}(0.03)} & {\scriptsize{}(0.04)} & {\scriptsize{}(0.05)} & {\scriptsize{}(0.06)} & {\scriptsize{}(0.04)} & {\scriptsize{}(0.05)} & {\scriptsize{}(0.04)} & {\scriptsize{}(0.08)} & {\scriptsize{}(0.02)} & {\scriptsize{}(0.04)}\tabularnewline
{\scriptsize{}$mktrf\times smb$} & {\scriptsize{}0.03} & {\scriptsize{}0.02} & {\scriptsize{}0.06} & {\scriptsize{}0.09} & {\scriptsize{}0.03} & {\scriptsize{}0.04} & {\scriptsize{}0.04} & {\scriptsize{}0.05} & {\scriptsize{}0.02} & {\scriptsize{}0.01}\tabularnewline
 & {\scriptsize{}(0.01)} & {\scriptsize{}(0.01)} & {\scriptsize{}(0.01)} & {\scriptsize{}(0.02)} & {\scriptsize{}(0.01)} & {\scriptsize{}(0.01)} & {\scriptsize{}(0.01)} & {\scriptsize{}(0.02)} & {\scriptsize{}(0.01)} & {\scriptsize{}(0.01)}\tabularnewline
{\scriptsize{}$mktrf\times hml$} & {\scriptsize{}-0.01} & {\scriptsize{}-0.02} & {\scriptsize{}-0.06} & {\scriptsize{}-0.06} & {\scriptsize{}-0.01} & {\scriptsize{}-0.00} & {\scriptsize{}-0.03} & {\scriptsize{}-0.04} & {\scriptsize{}-0.01} & {\scriptsize{}-0.01}\tabularnewline
 & {\scriptsize{}(0.01)} & {\scriptsize{}(0.01)} & {\scriptsize{}(0.01)} & {\scriptsize{}(0.01)} & {\scriptsize{}(0.01)} & {\scriptsize{}(0.01)} & {\scriptsize{}(0.01)} & {\scriptsize{}(0.01)} & {\scriptsize{}(0.00)} & {\scriptsize{}(0.01)}\tabularnewline
{\scriptsize{}$mktrf\times rmw$} & {\scriptsize{}0.02} & {\scriptsize{}0.00} & {\scriptsize{}-0.01} & {\scriptsize{}-0.02} & {\scriptsize{}0.06} & {\scriptsize{}0.06} & {\scriptsize{}0.08} & {\scriptsize{}0.08} & {\scriptsize{}-0.00} & {\scriptsize{}-0.02}\tabularnewline
 & {\scriptsize{}(0.01)} & {\scriptsize{}(0.02)} & {\scriptsize{}(0.02)} & {\scriptsize{}(0.03)} & {\scriptsize{}(0.02)} & {\scriptsize{}(0.02)} & {\scriptsize{}(0.02)} & {\scriptsize{}(0.03)} & {\scriptsize{}(0.01)} & {\scriptsize{}(0.02)}\tabularnewline
{\scriptsize{}$mktrf\times cma$} & {\scriptsize{}0.01} & {\scriptsize{}-0.00} & {\scriptsize{}0.14} & {\scriptsize{}0.13} & {\scriptsize{}0.03} & {\scriptsize{}0.02} & {\scriptsize{}0.03} & {\scriptsize{}0.01} & {\scriptsize{}0.04} & {\scriptsize{}0.05}\tabularnewline
 & {\scriptsize{}(0.02)} & {\scriptsize{}(0.02)} & {\scriptsize{}(0.03)} & {\scriptsize{}(0.03)} & {\scriptsize{}(0.02)} & {\scriptsize{}(0.03)} & {\scriptsize{}(0.02)} & {\scriptsize{}(0.04)} & {\scriptsize{}(0.01)} & {\scriptsize{}(0.02)}\tabularnewline
{\scriptsize{}$smb\times hml$} & {\scriptsize{}-0.03} & {\scriptsize{}-0.03} & {\scriptsize{}-0.10} & {\scriptsize{}-0.13} & {\scriptsize{}-0.05} & {\scriptsize{}-0.05} & {\scriptsize{}-0.03} & {\scriptsize{}-0.06} & {\scriptsize{}-0.06} & {\scriptsize{}-0.05}\tabularnewline
 & {\scriptsize{}(0.01)} & {\scriptsize{}(0.02)} & {\scriptsize{}(0.02)} & {\scriptsize{}(0.02)} & {\scriptsize{}(0.01)} & {\scriptsize{}(0.02)} & {\scriptsize{}(0.02)} & {\scriptsize{}(0.03)} & {\scriptsize{}(0.01)} & {\scriptsize{}(0.02)}\tabularnewline
{\scriptsize{}$smb\times rmw$} & {\scriptsize{}-0.01} & {\scriptsize{}-0.03} & {\scriptsize{}0.04} & {\scriptsize{}0.13} & {\scriptsize{}0.08} & {\scriptsize{}0.08} & {\scriptsize{}-0.02} & {\scriptsize{}-0.06} & {\scriptsize{}0.03} & {\scriptsize{}0.01}\tabularnewline
 & {\scriptsize{}(0.03)} & {\scriptsize{}(0.04)} & {\scriptsize{}(0.04)} & {\scriptsize{}(0.05)} & {\scriptsize{}(0.03)} & {\scriptsize{}(0.05)} & {\scriptsize{}(0.04)} & {\scriptsize{}(0.07)} & {\scriptsize{}(0.02)} & {\scriptsize{}(0.03)}\tabularnewline
{\scriptsize{}$smb\times cma$} & {\scriptsize{}0.05} & {\scriptsize{}0.03} & {\scriptsize{}0.21} & {\scriptsize{}0.20} & {\scriptsize{}-0.06} & {\scriptsize{}-0.04} & {\scriptsize{}0.10} & {\scriptsize{}0.13} & {\scriptsize{}0.04} & {\scriptsize{}0.04}\tabularnewline
 & {\scriptsize{}(0.03)} & {\scriptsize{}(0.05)} & {\scriptsize{}(0.05)} & {\scriptsize{}(0.07)} & {\scriptsize{}(0.04)} & {\scriptsize{}(0.06)} & {\scriptsize{}(0.05)} & {\scriptsize{}(0.09)} & {\scriptsize{}(0.03)} & {\scriptsize{}(0.04)}\tabularnewline
{\scriptsize{}$hml\times rmw$} & {\scriptsize{}0.02} & {\scriptsize{}0.03} & {\scriptsize{}0.09} & {\scriptsize{}0.15} & {\scriptsize{}0.08} & {\scriptsize{}0.09} & {\scriptsize{}-0.05} & {\scriptsize{}-0.00} & {\scriptsize{}0.01} & {\scriptsize{}0.04}\tabularnewline
 & {\scriptsize{}(0.02)} & {\scriptsize{}(0.03)} & {\scriptsize{}(0.03)} & {\scriptsize{}(0.04)} & {\scriptsize{}(0.02)} & {\scriptsize{}(0.04)} & {\scriptsize{}(0.03)} & {\scriptsize{}(0.05)} & {\scriptsize{}(0.02)} & {\scriptsize{}(0.03)}\tabularnewline
{\scriptsize{}$hml\times cma$} & {\scriptsize{}-0.05} & {\scriptsize{}-0.08} & {\scriptsize{}-0.22} & {\scriptsize{}-0.30} & {\scriptsize{}-0.15} & {\scriptsize{}-0.15} & {\scriptsize{}-0.13} & {\scriptsize{}-0.21} & {\scriptsize{}-0.11} & {\scriptsize{}-0.09}\tabularnewline
 & {\scriptsize{}(0.03)} & {\scriptsize{}(0.04)} & {\scriptsize{}(0.04)} & {\scriptsize{}(0.05)} & {\scriptsize{}(0.03)} & {\scriptsize{}(0.05)} & {\scriptsize{}(0.04)} & {\scriptsize{}(0.07)} & {\scriptsize{}(0.02)} & {\scriptsize{}(0.03)}\tabularnewline
{\scriptsize{}$rmw\times cma$} & {\scriptsize{}-0.04} & {\scriptsize{}-0.07} & {\scriptsize{}0.22} & {\scriptsize{}0.20} & {\scriptsize{}0.08} & {\scriptsize{}-0.04} & {\scriptsize{}-0.10} & {\scriptsize{}-0.16} & {\scriptsize{}0.04} & {\scriptsize{}-0.03}\tabularnewline
 & {\scriptsize{}(0.05)} & {\scriptsize{}(0.07)} & {\scriptsize{}(0.07)} & {\scriptsize{}(0.09)} & {\scriptsize{}(0.06)} & {\scriptsize{}(0.08)} & {\scriptsize{}(0.07)} & {\scriptsize{}(0.12)} & {\scriptsize{}(0.04)} & {\scriptsize{}(0.06)}\tabularnewline
{\scriptsize{}Cons} & {\scriptsize{}-0.36} & {\scriptsize{}-0.50} & {\scriptsize{}-0.50} & {\scriptsize{}-0.69} & {\scriptsize{}-0.39} & {\scriptsize{}-0.56} & {\scriptsize{}-0.57} & {\scriptsize{}-0.80} & {\scriptsize{}-0.32} & {\scriptsize{}-0.43}\tabularnewline
 & {\scriptsize{}(0.01)} & {\scriptsize{}(0.01)} & {\scriptsize{}(0.02)} & {\scriptsize{}(0.02)} & {\scriptsize{}(0.01)} & {\scriptsize{}(0.02)} & {\scriptsize{}(0.01)} & {\scriptsize{}(0.03)} & {\scriptsize{}(0.01)} & {\scriptsize{}(0.01)}\tabularnewline
\hline 
\end{tabular}
\end{table}
{\tiny\par}
\par\end{center}

\begin{center}
{\footnotesize{}}
\begin{table}[H]
{\footnotesize{}\caption{\small{}Inequality Regression\label{tab:Inequality-regression}}
}{\footnotesize\par}

{\footnotesize{}\medskip{}
}{\footnotesize\par}
\caption*{
This table reports results of the inequality regression for each year
from 2001 to 2018.  Standard errors are reported in parentheses.
}
{\footnotesize{}\medskip{}
}{\footnotesize\par}
\centering{}{\scriptsize{}}%
\begin{tabular}{lccccccccc}
\hline 
 & {\scriptsize{}2001} & {\scriptsize{}2002} & {\scriptsize{}2003} & {\scriptsize{}2004} & {\scriptsize{}2005} & {\scriptsize{}2006} & {\scriptsize{}2007} & {\scriptsize{}2008} & {\scriptsize{}2009}\tabularnewline
\hline 
 &  &  &  &  &  &  &  &  & \tabularnewline
{\scriptsize{}Fam Size} & {\scriptsize{}0.147} & {\scriptsize{}0.061} & {\scriptsize{}0.021} & {\scriptsize{}0.089} & {\scriptsize{}0.051} & {\scriptsize{}-0.070} & {\scriptsize{}-0.004} & {\scriptsize{}-0.034} & {\scriptsize{}0.126}\tabularnewline
 & {\scriptsize{}(0.052)} & {\scriptsize{}(0.044)} & {\scriptsize{}(0.049)} & {\scriptsize{}(0.054)} & {\scriptsize{}(0.053)} & {\scriptsize{}(0.051)} & {\scriptsize{}(0.048)} & {\scriptsize{}(0.051)} & {\scriptsize{}(0.049)}\tabularnewline
{\scriptsize{}No Child} & {\scriptsize{}-0.151} & {\scriptsize{}-0.111} & {\scriptsize{}-0.024} & {\scriptsize{}-0.134} & {\scriptsize{}-0.115} & {\scriptsize{}0.028} & {\scriptsize{}-0.026} & {\scriptsize{}0.038} & {\scriptsize{}-0.194}\tabularnewline
 & {\scriptsize{}(0.062)} & {\scriptsize{}(0.053)} & {\scriptsize{}(0.060)} & {\scriptsize{}(0.064)} & {\scriptsize{}(0.063)} & {\scriptsize{}(0.061)} & {\scriptsize{}(0.057)} & {\scriptsize{}(0.061)} & {\scriptsize{}(0.058)}\tabularnewline
{\scriptsize{}Age} & {\scriptsize{}0.019} & {\scriptsize{}0.010} & {\scriptsize{}0.034} & {\scriptsize{}0.011} & {\scriptsize{}0.013} & {\scriptsize{}0.013} & {\scriptsize{}0.001} & {\scriptsize{}0.019} & {\scriptsize{}0.020}\tabularnewline
 & {\scriptsize{}(0.010)} & {\scriptsize{}(0.009)} & {\scriptsize{}(0.010)} & {\scriptsize{}(0.011)} & {\scriptsize{}(0.010)} & {\scriptsize{}(0.010)} & {\scriptsize{}(0.010)} & {\scriptsize{}(0.010)} & {\scriptsize{}(0.010)}\tabularnewline
{\scriptsize{}Edu} & {\scriptsize{}0.140} & {\scriptsize{}0.122} & {\scriptsize{}0.134} & {\scriptsize{}0.061} & {\scriptsize{}0.118} & {\scriptsize{}0.129} & {\scriptsize{}0.117} & {\scriptsize{}0.099} & {\scriptsize{}0.113}\tabularnewline
 & {\scriptsize{}(0.015)} & {\scriptsize{}(0.014)} & {\scriptsize{}(0.015)} & {\scriptsize{}(0.016)} & {\scriptsize{}(0.015)} & {\scriptsize{}(0.015)} & {\scriptsize{}(0.014)} & {\scriptsize{}(0.015)} & {\scriptsize{}(0.014)}\tabularnewline
{\scriptsize{}Constant} & {\scriptsize{}0.055} & {\scriptsize{}0.705} & {\scriptsize{}-0.419} & {\scriptsize{}1.138} & {\scriptsize{}0.778} & {\scriptsize{}0.886} & {\scriptsize{}1.365} & {\scriptsize{}0.784} & {\scriptsize{}0.377}\tabularnewline
 & {\scriptsize{}(0.480)} & {\scriptsize{}(0.429)} & {\scriptsize{}(0.482)} & {\scriptsize{}(0.504)} & {\scriptsize{}(0.489)} & {\scriptsize{}(0.498)} & {\scriptsize{}(0.467)} & {\scriptsize{}(0.480)} & {\scriptsize{}(0.463)}\tabularnewline
 &  &  &  &  &  &  &  &  & \tabularnewline
\hline 
 & {\scriptsize{}2010} & {\scriptsize{}2011} & {\scriptsize{}2012} & {\scriptsize{}2013} & {\scriptsize{}2014} & {\scriptsize{}2015} & {\scriptsize{}2016} & {\scriptsize{}2017} & {\scriptsize{}2018}\tabularnewline
\hline 
 &  &  &  &  &  &  &  &  & \tabularnewline
{\scriptsize{}Fam Size} & {\scriptsize{}-0.011} & {\scriptsize{}-0.134} & {\scriptsize{}-0.114} & {\scriptsize{}-0.094} & {\scriptsize{}-0.019} & {\scriptsize{}-0.123} & {\scriptsize{}-0.090} & {\scriptsize{}-0.128} & {\scriptsize{}-0.083}\tabularnewline
 & {\scriptsize{}(0.048)} & {\scriptsize{}(0.044)} & {\scriptsize{}(0.048)} & {\scriptsize{}(0.052)} & {\scriptsize{}(0.044)} & {\scriptsize{}(0.047)} & {\scriptsize{}(0.047)} & {\scriptsize{}(0.047)} & {\scriptsize{}(0.049)}\tabularnewline
{\scriptsize{}No Child} & {\scriptsize{}-0.042} & {\scriptsize{}0.060} & {\scriptsize{}0.067} & {\scriptsize{}0.077} & {\scriptsize{}-0.071} & {\scriptsize{}0.059} & {\scriptsize{}0.041} & {\scriptsize{}0.085} & {\scriptsize{}0.040}\tabularnewline
 & {\scriptsize{}(0.058)} & {\scriptsize{}(0.053)} & {\scriptsize{}(0.058)} & {\scriptsize{}(0.063)} & {\scriptsize{}(0.053)} & {\scriptsize{}(0.057)} & {\scriptsize{}(0.057)} & {\scriptsize{}(0.056)} & {\scriptsize{}(0.060)}\tabularnewline
{\scriptsize{}Age} & {\scriptsize{}0.000} & {\scriptsize{}0.030} & {\scriptsize{}0.012} & {\scriptsize{}0.016} & {\scriptsize{}0.007} & {\scriptsize{}0.022} & {\scriptsize{}0.010} & {\scriptsize{}-0.006} & {\scriptsize{}0.018}\tabularnewline
 & {\scriptsize{}(0.010)} & {\scriptsize{}(0.009)} & {\scriptsize{}(0.010)} & {\scriptsize{}(0.011)} & {\scriptsize{}(0.010)} & {\scriptsize{}(0.011)} & {\scriptsize{}(0.011)} & {\scriptsize{}(0.010)} & {\scriptsize{}(0.011)}\tabularnewline
{\scriptsize{}Edu} & {\scriptsize{}0.067} & {\scriptsize{}0.094} & {\scriptsize{}0.057} & {\scriptsize{}0.050} & {\scriptsize{}0.066} & {\scriptsize{}0.065} & {\scriptsize{}0.080} & {\scriptsize{}0.051} & {\scriptsize{}0.071}\tabularnewline
 & {\scriptsize{}(0.015)} & {\scriptsize{}(0.014)} & {\scriptsize{}(0.015)} & {\scriptsize{}(0.017)} & {\scriptsize{}(0.014)} & {\scriptsize{}(0.015)} & {\scriptsize{}(0.016)} & {\scriptsize{}(0.015)} & {\scriptsize{}(0.016)}\tabularnewline
{\scriptsize{}Constant} & {\scriptsize{}1.912} & {\scriptsize{}0.618} & {\scriptsize{}1.696} & {\scriptsize{}1.623} & {\scriptsize{}1.748} & {\scriptsize{}1.274} & {\scriptsize{}1.704} & {\scriptsize{}2.643} & {\scriptsize{}1.366}\tabularnewline
 & {\scriptsize{}(0.469)} & {\scriptsize{}(0.442)} & {\scriptsize{}(0.490)} & {\scriptsize{}(0.542)} & {\scriptsize{}(0.459)} & {\scriptsize{}(0.499)} & {\scriptsize{}(0.520)} & {\scriptsize{}(0.496)} & {\scriptsize{}(0.528)}\tabularnewline
 &  &  &  &  &  &  &  &  & \tabularnewline
\hline 
\end{tabular}{\scriptsize\par}
\end{table}
{\footnotesize\par}
\par\end{center}

\newpage{}
\begin{center}
{\footnotesize{}}
\begin{table}[H]
{\footnotesize{}\caption{\small{}Social Welfare Regression\label{tab:lowincome-regression}}
}{\footnotesize\par}

{\footnotesize{}\medskip{}
}{\footnotesize\par}
\begin{flushleft}
This table reports results of the social welfare (exponential) regression for each year
from 2001 to 2018.  Standard errors are reported in parentheses.
\end{flushleft}
{\footnotesize{}\medskip{}
}{\footnotesize\par}
\centering{}{\scriptsize{}}%
\begin{tabular}{lccccccccc}
\hline 
 & {\scriptsize{}2001} & {\scriptsize{}2002} & {\scriptsize{}2003} & {\scriptsize{}2004} & {\scriptsize{}2005} & {\scriptsize{}2006} & {\scriptsize{}2007} & {\scriptsize{}2008} & {\scriptsize{}2009}\tabularnewline
\hline 
 &  &  &  &  &  &  &  &  & \tabularnewline
{\scriptsize{}Fam Size} & {\scriptsize{}-0.061} & {\scriptsize{}-0.051} & {\scriptsize{}-0.026} & {\scriptsize{}-0.060} & {\scriptsize{}-0.050} & {\scriptsize{}0.030} & {\scriptsize{}0.009} & {\scriptsize{}0.031} & {\scriptsize{}-0.080}\tabularnewline
 & {\scriptsize{}(0.038)} & {\scriptsize{}(0.031)} & {\scriptsize{}(0.039)} & {\scriptsize{}(0.045)} & {\scriptsize{}(0.039)} & {\scriptsize{}(0.038)} & {\scriptsize{}(0.033)} & {\scriptsize{}(0.036)} & {\scriptsize{}(0.036)}\tabularnewline
{\scriptsize{}No Child} & {\scriptsize{}0.146} & {\scriptsize{}0.166} & {\scriptsize{}0.113} & {\scriptsize{}0.166} & {\scriptsize{}0.199} & {\scriptsize{}0.091} & {\scriptsize{}0.094} & {\scriptsize{}0.061} & {\scriptsize{}0.220}\tabularnewline
 & {\scriptsize{}(0.046)} & {\scriptsize{}(0.037)} & {\scriptsize{}(0.047)} & {\scriptsize{}(0.053)} & {\scriptsize{}(0.046)} & {\scriptsize{}(0.046)} & {\scriptsize{}(0.040)} & {\scriptsize{}(0.043)} & {\scriptsize{}(0.043)}\tabularnewline
{\scriptsize{}Age} & {\scriptsize{}-0.005} & {\scriptsize{}0.004} & {\scriptsize{}-0.013} & {\scriptsize{}0.007} & {\scriptsize{}0.004} & {\scriptsize{}0.014} & {\scriptsize{}0.014} & {\scriptsize{}0.000} & {\scriptsize{}0.002}\tabularnewline
 & {\scriptsize{}(0.008)} & {\scriptsize{}(0.006)} & {\scriptsize{}(0.008)} & {\scriptsize{}(0.009)} & {\scriptsize{}(0.008)} & {\scriptsize{}(0.008)} & {\scriptsize{}(0.007)} & {\scriptsize{}(0.007)} & {\scriptsize{}(0.007)}\tabularnewline
{\scriptsize{}Edu} & {\scriptsize{}0.109} & {\scriptsize{}0.104} & {\scriptsize{}0.088} & {\scriptsize{}0.117} & {\scriptsize{}0.106} & {\scriptsize{}0.110} & {\scriptsize{}0.123} & {\scriptsize{}0.128} & {\scriptsize{}0.126}\tabularnewline
 & {\scriptsize{}(0.011)} & {\scriptsize{}(0.010)} & {\scriptsize{}(0.012)} & {\scriptsize{}(0.013)} & {\scriptsize{}(0.011)} & {\scriptsize{}(0.012)} & {\scriptsize{}(0.010)} & {\scriptsize{}(0.010)} & {\scriptsize{}(0.010)}\tabularnewline
{\scriptsize{}Constant} & {\scriptsize{}5.179} & {\scriptsize{}4.836} & {\scriptsize{}5.675} & {\scriptsize{}4.613} & {\scriptsize{}4.738} & {\scriptsize{}4.162} & {\scriptsize{}4.222} & {\scriptsize{}4.750} & {\scriptsize{}4.828}\tabularnewline
 & {\scriptsize{}(0.356)} & {\scriptsize{}(0.302)} & {\scriptsize{}(0.380)} & {\scriptsize{}(0.418)} & {\scriptsize{}(0.361)} & {\scriptsize{}(0.375)} & {\scriptsize{}(0.327)} & {\scriptsize{}(0.339)} & {\scriptsize{}(0.339)}\tabularnewline
 &  &  &  &  &  &  &  &  & \tabularnewline
\hline 
 & {\scriptsize{}2010} & {\scriptsize{}2011} & {\scriptsize{}2012} & {\scriptsize{}2013} & {\scriptsize{}2014} & {\scriptsize{}2015} & {\scriptsize{}2016} & {\scriptsize{}2017} & {\scriptsize{}2018}\tabularnewline
\hline 
 &  &  &  &  &  &  &  &  & \tabularnewline
{\scriptsize{}Fam Size} & {\scriptsize{}0.043} & {\scriptsize{}0.054} & {\scriptsize{}0.091} & {\scriptsize{}0.073} & {\scriptsize{}0.009} & {\scriptsize{}0.071} & {\scriptsize{}0.052} & {\scriptsize{}0.123} & {\scriptsize{}0.081}\tabularnewline
 & {\scriptsize{}(0.036)} & {\scriptsize{}(0.032)} & {\scriptsize{}(0.035)} & {\scriptsize{}(0.040)} & {\scriptsize{}(0.031)} & {\scriptsize{}(0.034)} & {\scriptsize{}(0.035)} & {\scriptsize{}(0.035)} & {\scriptsize{}(0.037)}\tabularnewline
{\scriptsize{}No Child} & {\scriptsize{}0.082} & {\scriptsize{}0.101} & {\scriptsize{}0.025} & {\scriptsize{}0.052} & {\scriptsize{}0.149} & {\scriptsize{}0.090} & {\scriptsize{}0.092} & {\scriptsize{}0.011} & {\scriptsize{}0.038}\tabularnewline
 & {\scriptsize{}(0.043)} & {\scriptsize{}(0.039)} & {\scriptsize{}(0.042)} & {\scriptsize{}(0.049)} & {\scriptsize{}(0.037)} & {\scriptsize{}(0.041)} & {\scriptsize{}(0.042)} & {\scriptsize{}(0.042)} & {\scriptsize{}(0.045)}\tabularnewline
{\scriptsize{}Age} & {\scriptsize{}0.023} & {\scriptsize{}-0.005} & {\scriptsize{}0.015} & {\scriptsize{}0.012} & {\scriptsize{}0.018} & {\scriptsize{}0.008} & {\scriptsize{}0.020} & {\scriptsize{}0.018} & {\scriptsize{}0.003}\tabularnewline
 & {\scriptsize{}(0.007)} & {\scriptsize{}(0.007)} & {\scriptsize{}(0.007)} & {\scriptsize{}(0.009)} & {\scriptsize{}(0.007)} & {\scriptsize{}(0.008)} & {\scriptsize{}(0.008)} & {\scriptsize{}(0.008)} & {\scriptsize{}(0.008)}\tabularnewline
{\scriptsize{}Edu} & {\scriptsize{}0.157} & {\scriptsize{}0.150} & {\scriptsize{}0.171} & {\scriptsize{}0.176} & {\scriptsize{}0.164} & {\scriptsize{}0.168} & {\scriptsize{}0.149} & {\scriptsize{}0.170} & {\scriptsize{}0.169}\tabularnewline
 & {\scriptsize{}(0.011)} & {\scriptsize{}(0.010)} & {\scriptsize{}(0.011)} & {\scriptsize{}(0.013)} & {\scriptsize{}(0.010)} & {\scriptsize{}(0.011)} & {\scriptsize{}(0.012)} & {\scriptsize{}(0.011)} & {\scriptsize{}(0.012)}\tabularnewline
{\scriptsize{}Constant} & {\scriptsize{}3.370} & {\scriptsize{}4.651} & {\scriptsize{}3.533} & {\scriptsize{}3.598} & {\scriptsize{}3.622} & {\scriptsize{}3.890} & {\scriptsize{}3.516} & {\scriptsize{}3.370} & {\scriptsize{}4.178}\tabularnewline
 & {\scriptsize{}(0.347)} & {\scriptsize{}(0.325)} & {\scriptsize{}(0.359)} & {\scriptsize{}(0.415)} & {\scriptsize{}(0.326)} & {\scriptsize{}(0.360)} & {\scriptsize{}(0.382)} & {\scriptsize{}(0.368)} & {\scriptsize{}(0.396)}\tabularnewline
 &  &  &  &  &  &  &  &  & \tabularnewline
\hline 
\end{tabular}{\scriptsize\par}
\end{table}
\end{center}

\clearpage

\section{Data}\label{sec: data}

\subsection*{Financial Market Data}

The Fama-French industry daily returns are obtained from Kenneth French's
website. In our main specifications, we use the Fama-French 5 industry
definition. In the Appendix, we also provide results based on the
Fama-French 30-industry definition. The factor model data, including
the Fama-French 3-factor and the Fama-French 5-factor, are also obtained
from Kenneth French's website. The industry returns and the factor
model returns span from 1963 to 2021. 

\subsection*{Wage Data}

The data are drawn from the IPUMS website. We apply filters similar
to \cite{ACF06}. The sample for the calculations
consists of US-born black and white men with age 40-49 with at least
5 years of eduation, with positive wages and hours worked. The data
span from 2001 to 2018. For each year, we use a 30,000 random sample.
The logged wage variable is the average logged weekly wage and is
calculated as the log of the annual income from work divided by weeks
worked.

\end{document}